\documentclass[nohyperref]{article}

\usepackage{microtype}
\usepackage{graphicx}
\graphicspath{{figures/}}
\usepackage{subcaption}
\usepackage{makecell}
\usepackage{booktabs} % for professional tables
\usepackage[inline=true, margin=false]{fixme}

\usepackage{hyperref}

\usepackage[accepted]{icml2022}

\makeatletter
\renewcommand{\ICML@appearing}{}
\makeatother

\usepackage{amsmath}
\usepackage{amssymb}
\usepackage{mathtools}
\usepackage{amsthm}

\theoremstyle{plain}
\newtheorem{theorem}{Theorem}[section]
\newtheorem{proposition}[theorem]{Proposition}

\theoremstyle{definition}

\theoremstyle{remark}

\usepackage{bm}
\renewcommand{\vec}[1]{\bm{#1}}
\newcommand{\bp}{\text{bp}}
\newcommand{\alphaMat}{\bm{\alpha}}

\icmltitlerunning{The Multivariate Community Hawkes Model for Continuous-time Networks}

\begin{document}

\twocolumn[
\icmltitle{The Multivariate Community Hawkes Model for Dependent Relational Events in Continuous-time Networks}

% It is OKAY to include author information, even for blind
% submissions: the style file will automatically remove it for you
% unless you've provided the [accepted] option to the icml2022
% package.

% List of affiliations: The first argument should be a (short)
% identifier you will use later to specify author affiliations
% Academic affiliations should list Department, University, City, Region, Country
% Industry affiliations should list Company, City, Region, Country

% You can specify symbols, otherwise they are numbered in order.
% Ideally, you should not use this facility. Affiliations will be numbered
% in order of appearance and this is the preferred way.
\icmlsetsymbol{equal}{*}

\begin{icmlauthorlist}
\icmlauthor{Hadeel Soliman}{ut}
\icmlauthor{Lingfei Zhao}{osu}
\icmlauthor{Zhipeng Huang}{ut}
\icmlauthor{Subhadeep Paul}{osu}
\icmlauthor{Kevin S. Xu}{ut}
\end{icmlauthorlist}

\icmlaffiliation{ut}{Department of Electrical Engineering and Computer Science, University of Toledo, Toledo, OH, USA}
\icmlaffiliation{osu}{Department of Statistics, The Ohio State University, Columbus, OH, USA}

\icmlcorrespondingauthor{Subhadeep Paul}{paul.963@osu.edu}
\icmlcorrespondingauthor{Kevin S. Xu}{kevin.xu@utoledo.edu}

% You may provide any keywords that you
% find helpful for describing your paper; these are used to populate
% the "keywords" metadata in the PDF but will not be shown in the document
\icmlkeywords{Machine Learning, ICML}

\vskip 0.3in
]

% this must go after the closing bracket ] following \twocolumn[ ...

% This command actually creates the footnote in the first column
% listing the affiliations and the copyright notice.
% The command takes one argument, which is text to display at the start of the footnote.
% The \icmlEqualContribution command is standard text for equal contribution.
% Remove it (just {}) if you do not need this facility.

\printAffiliationsAndNotice{}  % leave blank if no need to mention equal contribution
%\printAffiliationsAndNotice{\icmlEqualContribution} % otherwise use the standard text.

\begin{abstract}
The stochastic block model (SBM) is one of the most widely used generative models for network data. 
Many continuous-time dynamic network models are built upon the same assumption as the SBM: edges or events between all pairs of nodes are conditionally independent given the block or community memberships, which prevents them from reproducing higher-order motifs such as triangles that are commonly observed in real networks.
We propose the \emph{multivariate community Hawkes (MULCH) model}, an extremely flexible community-based model for continuous-time networks that introduces \emph{dependence between node pairs} using structured multivariate Hawkes processes. 
We fit the model using a spectral clustering and likelihood-based local refinement procedure. 
We find that our proposed MULCH model is far more accurate than existing models both for predictive and generative tasks.
\end{abstract}

\section{Introduction}

Networks are often used to represent data in the form of relations (edges) between a set of entities (nodes). 
In many settings, the nodes and edges change over time, resulting in dynamic or temporal networks. 
We consider networks observed through \emph{timestamped relational events}, where each event is a triplet $(i,j,t)$ denoting events from node $i$ (sender) to node $j$ (receiver) at timestamp $t$. 
Application settings involving timestamped relational events include interactions (messages, likes, shares, etc.) between users on social media, financial transactions between buyers and sellers, and military actions between countries in diplomatic conflicts.

There has been significant recent interest in generative models for timestamped relational event data. 
Such models typically combine a Temporal Point Process (TPP) model such as a Hawkes process \citep{laub2021elements} for event times with a latent variable network model such as a Stochastic Block Model (SBM) \citep{nowicki2001estimation} for the sender and receiver of the event \citep{Blundell2012,Dubois2013,yang2017decoupling, miscouridou2018modelling,matias2015semiparametric,junuthula2019block,arastuie2020chip, huang2022mutually}. 
We call such models \emph{continuous-time network models} because they provide probabilities of observing events between nodes at arbitrary times.

Continuous-time network models typically assume that the probability of an event occurring between a pair of nodes $(i,j)$ at some time $t$ is conditionally independent of all other node pairs given the latent variable representation of the network. 
The conditional independence between node pairs allows them to be modeled separately using univariate or bivariate TPPs. 
Such an approach makes the model more tractable but prevents it from replicating higher-order structures including triangles and other network motifs \citep{benson2016higher}. 
Of particular interest in the dynamic network setting are \emph{temporal motifs}, which require multiple events to be formed between different nodes within a certain time window \citep{paranjape2017motifs}.

We propose the \emph{multivariate community Hawkes (MULCH) model}, a highly flexible continuous-time network model that introduces dependence between node pairs in a controlled manner. 
We jointly model all node pairs using a multivariate Hawkes process where an event between a node pair $(x,y)$ can increase the probability of an event between a different node pair $(i,j)$. 
To keep the model tractable, we impose an SBM-inspired structure on the excitation matrix $\alphaMat$ of the multivariate Hawkes process. 
We consider several different types of excitation to encourage formation of temporal motifs, inspired by the notion of participation shifts \citep{gibson2003participation, gibson2005taking} from sociology.

Our main contributions are as follows: 
(1) We propose the highly flexible MULCH model for continuous-time networks using multivariate Hawkes processes that incorporate many types of dependence between node pairs including reciprocity and participation shifts.
(2) We develop an efficient estimation approach for MULCH using spectral clustering and likelihood-based local refinement.
(3) We demonstrate the MULCH fits to several real data sets are superior to existing continuous-time network models both at predicting future events and at replicating temporal motifs.
(4) We present a case study using MULCH to analyze military dispute data, revealing how groups of countries act and respond to other actions in such disputes.

\section{Background}

\subsection{Multivariate Hawkes Process}
\label{sec:hawkes}
Temporal point processes (TPPs) are used to model events that occur randomly (i.e.~not at regularly spaced intervals) over time. 
A univariate TPP is characterized by a conditional intensity function $\lambda(t|\mathcal{H}_t)$ such that the expected number of events in an infinitesimal interval $dt$ around $t$ is given by $\lambda(t|\mathcal{H}_t) \, dt$ \citep{rasmussen2018lecture}. 
$\mathcal{H}_t$ denotes the history of the TPP (all event times up to time $t$); we drop it for ease of notation and simply use $\lambda(t)$. 
A univariate Hawkes process is a self-exciting TPP where the occurrence of an event increases the probability of another event occurring shortly afterwards \citep{laub2021elements}. 
Given a sequence of timestamps $\{t_1, t_2, \ldots, t_l\}$ for $l$ events, the conditional intensity function takes the form 
$\lambda(t) = \mu + \alpha\sum_{s: t_s<t}{\gamma(t - t_s)},$ 
where $\mu$ is the background or base intensity, $\alpha$ is the jump size or excitation, and $\gamma(\cdot)$ is the kernel function. 

A multivariate TPP is characterized by a set of conditional intensity functions $\{\lambda_j(t)\}_{j=1}^d$ for the $d$ different variables. 
A multivariate Hawkes process is both self and mutually exciting, so that an event in one variable can also increase the probability of an event in another variable \citep{zhou2013learning, hawkes2018hawkes}.
The conditional intensity function of the $j$th dimension is given by 
$\lambda_j(t) = \mu_j + \sum_{s: t_s<t}\alpha_{ij }{\gamma_{ij}(t - t_s)},$
where the background intensity $\mu_j$ can differ for each dimension, $\alpha_{ij}$ denotes the jump size that an event in dimension $i$ causes to dimension $j$, and $\gamma_{ij}(\cdot)$ denotes the kernel, which may differ across dimensions. 
Notice that the background intensities are now characterized by a $d$-dimensional vector $\bm{\mu}$, while the jump sizes and kernels are characterized by $d \times d$ matrices $\bm{\alpha}$ and $\bm{\gamma}$, respectively.

\subsection{Stochastic Block Model}

The \emph{stochastic block model} (SBM), first formalized by \citet{holland1983stochastic}, is one of the most widely used generative models for network data. 
The SBM was designed for a static network, but many dynamic extensions have since been proposed, which we discuss in the related work section. 
In a (static) SBM with $n$ nodes, every node $i$ is assigned to one and only \emph{block} $Z_i \in \{1, \ldots, K\}$, where $K$ denotes the number of blocks.
For a directed SBM, given the node membership vector $\vec Z = [Z_i]_{i=1}^n$, all off-diagonal entries of the adjacency matrix $A_{ij}$ are independent Bernoulli random variables with parameter $p_{Z_i,Z_j}$, where $p$ is a $K\times K$ matrix of probabilities which is not symmetric. 
Thus, the probability of forming an edge between nodes $i$ and $j$ depends only on the node memberships $Z_i$ and $Z_j$.

Real networks often have many reciprocated edges ($A_{ij} = 1 \Rightarrow A_{ji} = 1$) and triangles ($A_{ij} = 1, A_{jm} = 1 \Rightarrow A_{im} = 1)$.
These are not replicated by the SBM due to the independence of the adjacency matrix entries. 
Replicating triangles usually requires an additional generative process on top of the SBM \citep{peixoto2021disentangling} or generative models which induce dependence \cite{paul2018higher,bollobas2011sparse}.

Fitting an SBM to data involves estimating both the node memberships $\vec{Z}$ and the edge probabilities $p_{Z_i,Z_j}$ between all pairs of blocks \citep{nowicki2001estimation}. 
Several variants of spectral clustering, including regularized versions \cite{chaudhuri2012spectral, amini2013pseudo}, have been shown to be consistent estimators of the node memberships in the SBM and various extensions in several asymptotic settings \citep{rcy11,Sussman2012,qin13,lei2015consistency,chin2015stochastic,han2015consistent,gao2017achieving}. 
Spectral clustering scales to large networks with hundreds of thousands of nodes and is generally not sensitive to initialization, so it is also a practically useful estimator.

\subsection{Related Work}

Most prior models for continuous-time networks utilize low-dimensional latent variable representations of the networks to parameterize univariate or bivariate TPPs, typically Hawkes processes, for the node pairs. 
The self excitation in Hawkes processes has been found to be a good model for conversation event sequences \citep{masuda2013self} among other temporal relational event data. 
The latent variable representations are often inspired by generative models for static networks such as latent space models \citep{hrh02} and stochastic block models. 
Continuous-time network models have been built with continuous latent space representations \citep{yang2017decoupling, huang2022mutually} and latent block or community representations \citep{Blundell2012,Dubois2013,xin2015continuous,matias2015semiparametric, miscouridou2018modelling,corneli2018multiple,junuthula2019block,arastuie2020chip}. 

\begin{figure*}[t]
    \centering
    \hfill
    \begin{subfigure}[c]{0.3\textwidth}
        \includegraphics[width=\textwidth]{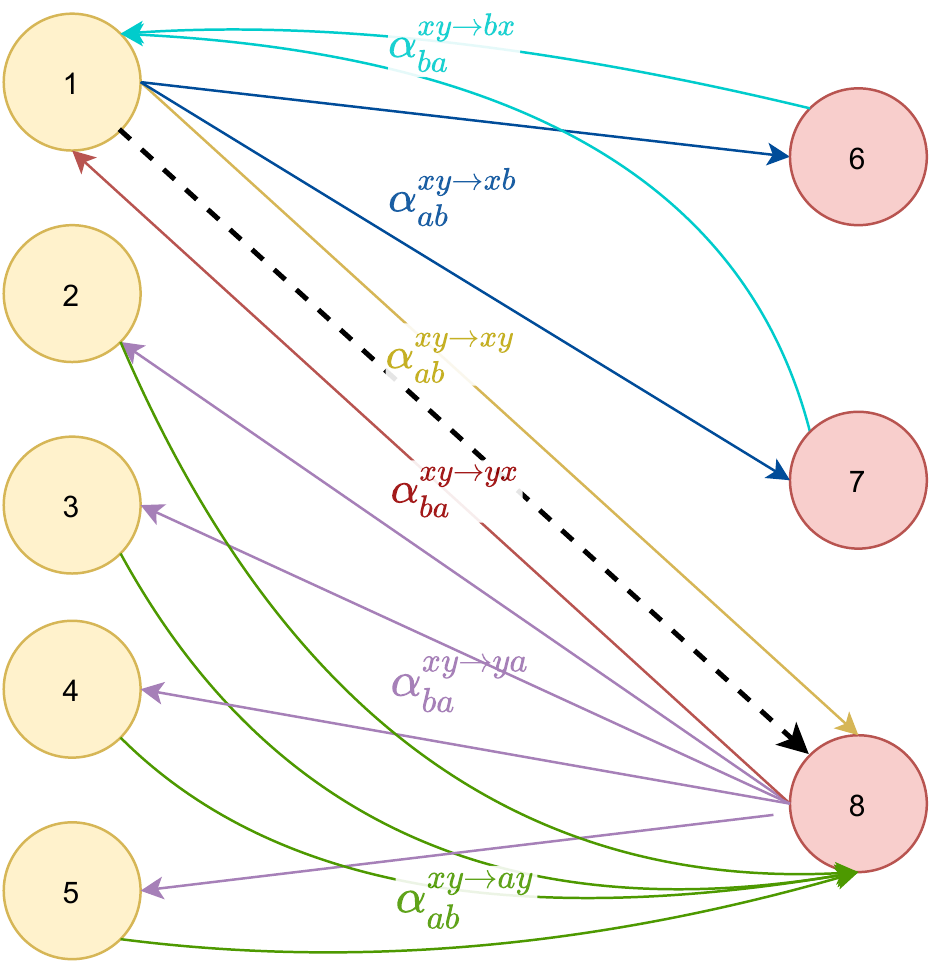}
        \caption{}
        \label{fig:excitationDiagram}
    \end{subfigure}
    \hfill
    \begin{subfigure}[c]{0.39\textwidth}
        \centering
        \includegraphics[width=\textwidth, clip=true, trim={0 0 20 0}]{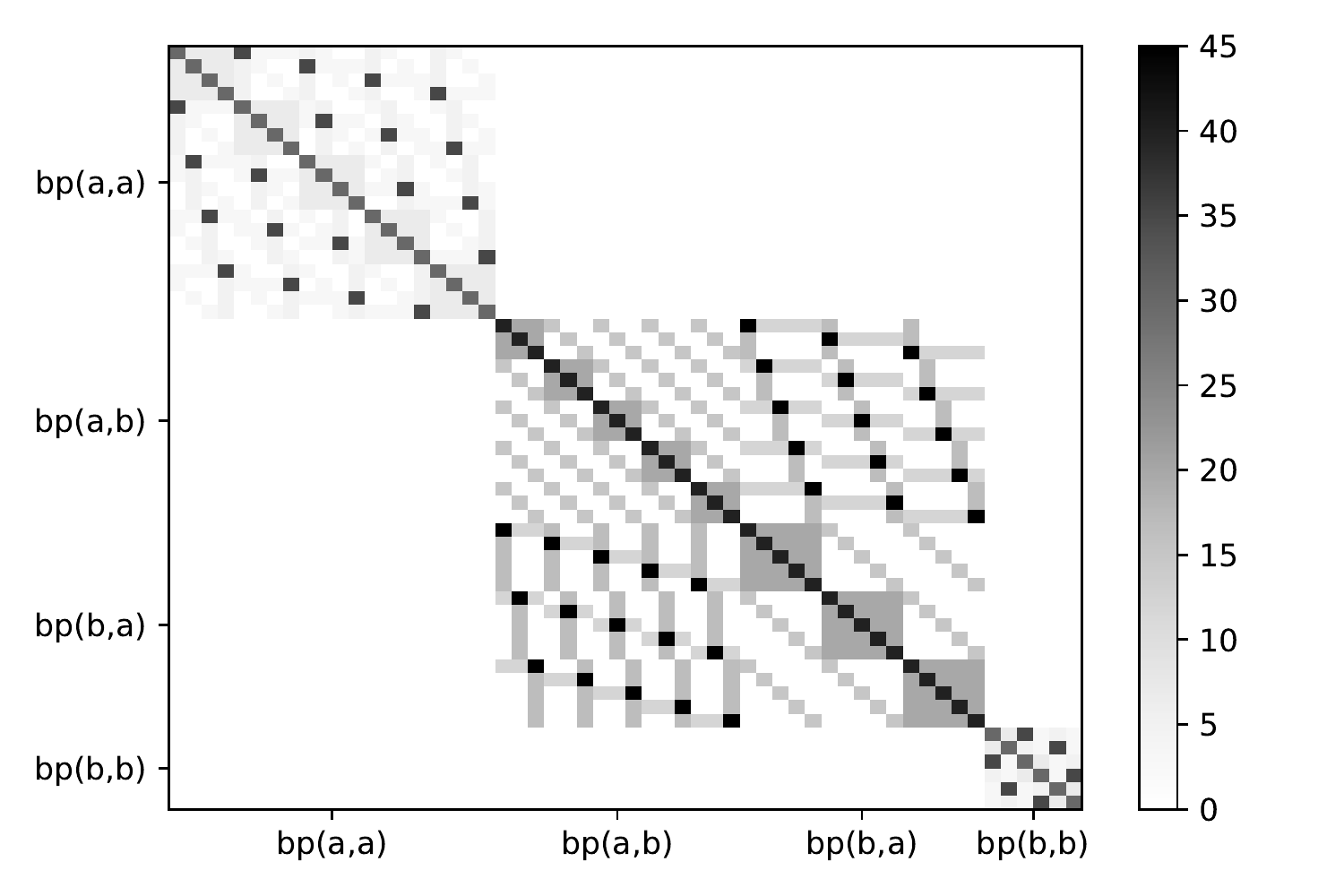}
        \caption{}
        \label{fig:Excitation_matrix_inside}
    \end{subfigure}
    \hfill
    \begin{subfigure}[c]{0.3\textwidth}
        \centering
        \includegraphics[width=\textwidth, clip=true, trim={10 0 20 0}]{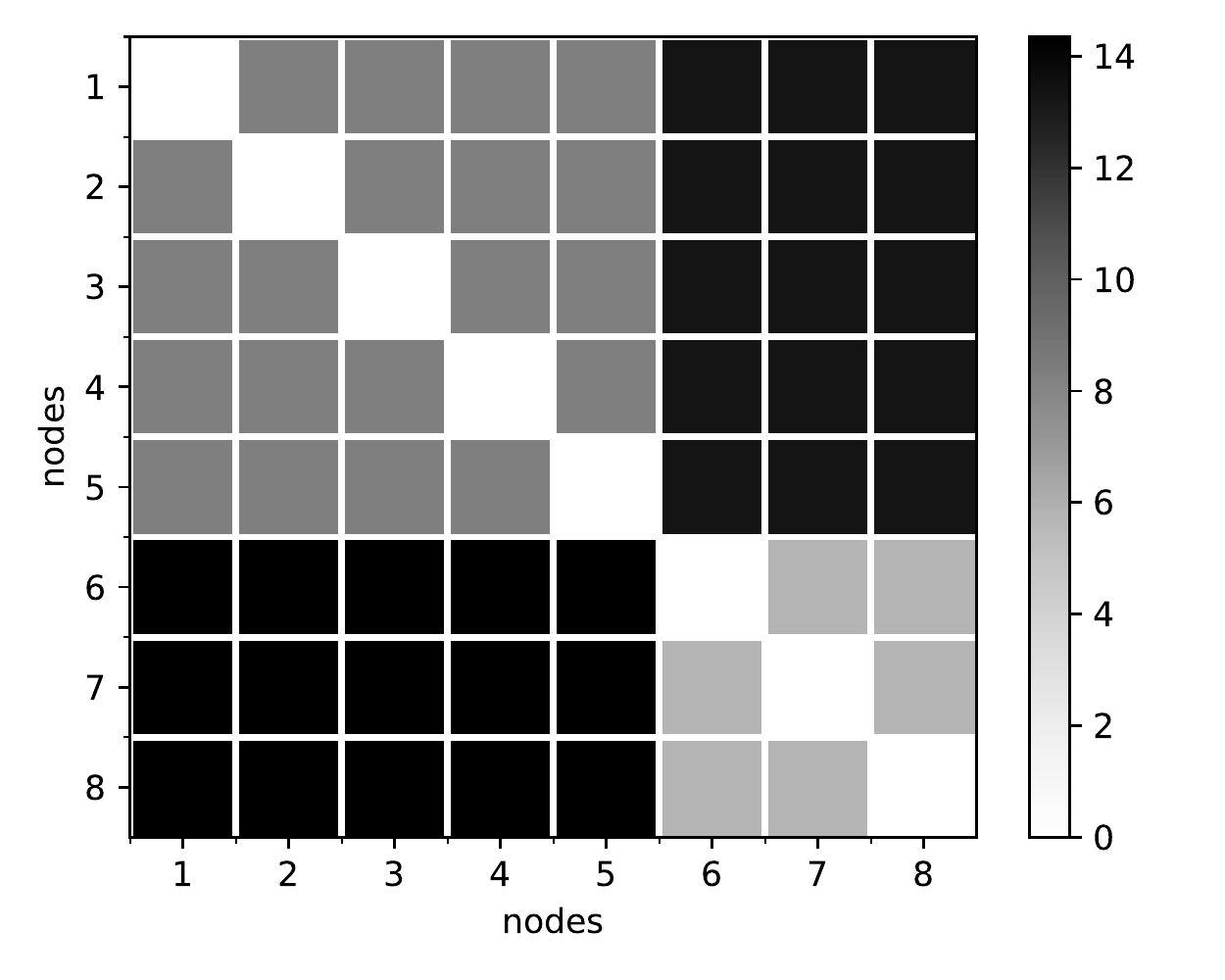}
        \caption{}
        \label{fig:Block_structure_count}
    \end{subfigure}
    \hfill
    \caption{Toy example illustrating the different excitations in the proposed multivariate block Hawkes model.
    Nodes 1-5 are in block $a$, while nodes 6-8 are in block $b$. 
    \subref{fig:excitationDiagram} The event $(1,8)$ (dashed line) excites the processes for the other node pairs shown in solid lines by the specified jump size $\alpha_{ab}$ or $\alpha_{ba}$.
    \subref{fig:Excitation_matrix_inside} Block diagonal structure of the excitation matrix $\alphaMat$ for this toy example.
    \subref{fig:Block_structure_count} Block structure of the expected count matrix (expected number of events between each node pair) for this toy example.}
    \label{fig:MultiBHM_illustration}
\end{figure*}

The CHIP \citep{arastuie2020chip} model uses a univariate Hawkes process to model self excitation for each node pair, with node pairs in the same community pair sharing parameters. 
Bivariate Hawkes process models \citep{Blundell2012,yang2017decoupling,miscouridou2018modelling, huang2022mutually} allow events $i \rightarrow j$, which we denote by the directed pair $(i,j)$, to influence the probability of events $(j,i)$. 
This encourages reciprocal events, which are commonly seen in email and messaging networks, where a reciprocal event typically denotes a user replying to a message.
A weakness of such models is that they still have no mechanism to encourage the formation of higher-order motifs such as triangles given their bivariate nature.
One way to accomplish this would be to move to higher-dimensional Hawkes processes, as we propose in this paper.

High-dimensional multivariate Hawkes processes are frequently used for estimating the structure of a latent or unobserved network from observed events at the nodes 
\citep{zhou2013learning,Linderman2014,Farajtabar2015,Tran2015}. 
These models are often used to estimate \emph{static} networks of diffusion from information cascades.
The information cascade is modeled using an $n$-dimensional Hawkes process with each dimension corresponding to a node, and the objective is to estimate the $n \times n$ excitation matrix $\bm{\alpha}$ of influence strengths between nodes.

\section{Multivariate Community Hawkes Model}

In order to replicate higher-order motifs such as triangles, we must move beyond univariate and bivariate Hawkes process-based models. 
One approach would be to model all (ordered) node pairs $(i,j)$ using an $n(n-1)$-dimensional Hawkes process where an event for any node pair can excite any other node pair, similar to the models used for information cascades \citep{zhou2013learning,Tran2015}. 
Such a model would have an $n(n-1) \times n(n-1)$ excitation matrix $\bm{\alpha}$, which would be computationally intractable even for relatively small networks with a few hundred nodes.

We propose the \emph{multivariate community Hawkes (MULCH)} model, an extremely flexible continuous-time network model based on the SBM where each node $i$ belongs to a block $a \in \{1, \ldots, K\}$, which we denote by $i \in a$ or $Z_i=a$. 
Each node pair $(i,j)$ then belongs to a block pair $(a,b)$, which we denote by $(i,j) \in \bp(a,b)$. 
We assume a block diagonal structure on $\alphaMat$ so that events for node pair $(x,y) \in \bp(a,b)$ can only excite node pairs $(i,j) \in \{\bp(a,b), \bp(b,a)\}$. 
This block diagonal structure is shown in Figure \ref{fig:Excitation_matrix_inside}. 
This structure allows for higher-order motifs to form in the same $(a,b)$ and reciprocal $(b,a)$ block pairs of the initial event for node pair $(x,y)$.

Within the non-zero diagonal blocks of $\alphaMat$, we introduce dependence between node pairs in a controlled manner using different types of excitations. 
Consider a node pair $(i,j) \in \bp(a,b)$, which has conditional intensity function
\begin{equation*}
    \lambda_{ij}(t) = \mu_{ab} + \sum_{\substack{(x,y) \in \bp(a,b) \\ (x,y) \in \bp(b,a)}} \alpha_{ab}^{xy \rightarrow ij} \sum_{t_s \in T_{xy}} \gamma_{xy \rightarrow ij}(t - t_s),
\end{equation*}
where $T_{xy}$ denotes the set of times that event $(x,y)$ happens. 
The parameters $\alpha_{ab}^{xy \rightarrow ij}$ control the types of excitations in the model by denoting which node pairs $(x,y)$ can increase the probability of an event for node pair $(i,j)$ occurring shortly after time $t$.

\subsection{Excitation Parameters}

\begin{table}[t]
    \centering
    \caption{Descriptions of the 6 types of excitations (illustrated in Figure \ref{fig:excitationDiagram}) we consider following an event from node $x$ in block $a$ to node $y$ in block $b$.
    Excitations that trigger events from block $b$ back to block $a$ are modeled by $\alpha_{ba}$, not $\alpha_{ab}$.}
    \label{tab:excitationDescriptions}
    \begin{tabular}{p{0.5in}p{2.5in}}
        \toprule
        Parameter                         & Excitation Type \\
        \midrule
        $\alpha_{ab}^{xy \rightarrow xy}$ & \emph{Self excitation}: continuation of event $(x,y)$ \\
        $\alpha_{ba}^{xy \rightarrow yx}$ & \emph{Reciprocal excitation}: event $(y,x)$ taken in response to event $(x,y)$ \\
        $\alpha_{ab}^{xy \rightarrow xb}$ & \emph{Turn continuation}: $(x,b)$ following $(x,y)$ to other nodes except for $y$ in the same block $b$\\
        $\alpha_{ba}^{xy \rightarrow ya}$ & \emph{Generalized reciprocity}: $(y,a)$ following $(x,y)$ to other nodes except $x$ in block $a$\\
        $\alpha_{ab}^{xy \rightarrow ay}$ & \emph{Allied continuation}: event $(a,y)$ following $(x,y)$ from other nodes except $x$ in block $a$ \\
        $\alpha_{ba}^{xy \rightarrow bx}$ & \emph{Allied reciprocity}: event $(b,x)$ following $(x,y)$ from other nodes except $y$ in block $b$ \\
        \bottomrule
    \end{tabular}
\end{table}

Many variants of our proposed multivariate block Hawkes model are possible depending on the structure of the parameters $\alpha_{ab}^{xy \rightarrow ij}$ corresponding to different types of excitations. 
We consider 6 types of excitations, listed in Table \ref{tab:excitationDescriptions} and illustrated for a sample network in Figure \ref{fig:excitationDiagram}. 
If we had only the excitation parameter $\alpha_{ab}^{xy \rightarrow xy}$, then our model would only incorporate self excitation and reduces to the CHIP model \citep{arastuie2020chip}. 
The remaining 5 parameters denote mutual excitation. 
These different types of mutual excitation correspond to  the notion of participation shifts \citep{gibson2003participation, gibson2005taking} from sociology. 
The terms ``turn continuation'' \citep{gibson2003participation} and ``generalized reciprocity'' \citep{yamagishi2000group} also have origins from sociology.

Each excitation is associated with its own jump size parameter that controls probability of such an event following the initial event $(x,y)$. 
The structure of a non-zero diagonal block of $\alphaMat$ under these 6 excitations is shown in Figure \ref{fig:Excitation_matrix_inside}.
By using specific types of excitations in a multivariate Hawkes process model, we not only reduce the number of parameters, but also improve interpretability through the set of excitation parameters for each block pair.

While other models have incorporated self excitation \citep{arastuie2020chip, junuthula2019block} and reciprocal excitation \citep{Blundell2012, yang2017decoupling, miscouridou2018modelling}, our model is the first to incorporate the 4 other types of excitations.
These newly-introduced excitations incorporate higher-order dependencies beyond a node pair to encourage the formation of motifs such as triangles.

\subsection{Generative Process}
To generate events between node pairs in the network, we first sample the membership of each node $Z_i$ from a categorical distribution with block probability $\bm{\pi}=[\pi_1, \cdots, \pi_K]$. There is a total of $n(n-1)$ node pairs, which are split into $K\times K$ block pairs. All node pairs within one block pair $(a, b)$ share the same base intensity $\mu_{ab}$ and same set of 6 excitation parameters from Table \ref{tab:excitationDescriptions}. 
We denote these 7 parameters for a block pair by $\bm{\theta}_{ab}$. 

By ordering node pairs by block pairs, we can form the  excitation matrix $\alphaMat$ for a high-dimensional multivariate Hawkes process as shown in Figure \ref{fig:Excitation_matrix_inside}.
As a result of MULCH's block diagonal structure on $\alphaMat$, we can simulate events between $(i,j) \in \{\bp(a,b), \bp(b,a)\}$ as separate multivariate Hawkes processes. 
We use a variant of Ogata's thinning algorithm proposed by \citet{xu2020modeling} to simulate the multivariate Hawkes processes. 
Given the block probabilities $\bm{\pi}$, the Hawkes process parameters for all block pairs $\bm{\Theta} = [\bm{\theta}_{ab}]$, and simulation duration $T$, the generative process is summarized in Algorithm \ref{alg:Generation}. 
The generated events are sets of triplets $(i,j,t)$ denoting nodes and timestamps concatenated into a single array $E$.

\begin{algorithm}[t]
\caption{MULCH Model Generative Process}
\label{alg:Generation}
\begin{algorithmic}[1]
\REQUIRE Block probabilities $\bm{\pi}$, Hawkes process parameters $\bm{\Theta}$, Duration $T$
\ENSURE Node memberships $\vec Z$, Array of event triplets $E$
    \FOR{$i=1$ to $n$}
        \STATE $Z_i \sim \text{Categorical}(\bm{\pi})$
    \ENDFOR
    \STATE Split $n\times(n-1)$ node pairs into $K\times K$ block pairs
    \FORALL{{diagonal block pair $(a,a)$}}
        \STATE $E_{aa} \leftarrow \text{MultivariateHawkesProcess}(\bm{\theta}_{aa},T)$
    \ENDFOR
    \FORALL{off-diagonal block pairs $(a,b)$ and $(b,a)$}
        \STATE $E_{ab}, E_{ba} \leftarrow \text{MultivariateHawkesProcess}(\bm{\theta}_{ab},\bm{\theta}_{ba},T)$
    \ENDFOR
    \STATE Concatenate all $E_{aa}$ and $E_{ab}$ into $E$ \\
    \RETURN $\vec Z, E$
\end{algorithmic}
\end{algorithm}

\subsection{Hawkes Process Kernel Selection}
\label{sec:KernelSelection}
A frequent choice of kernel for univariate Hawkes processes is an exponential kernel $\gamma(t) = e^{-\beta (t)}$. 
Estimation of $\beta$ is difficult, and estimators are typically not well-behaved \citep{santos2021surfacing}. 
In many applications involving multivariate Hawkes processes, the exponential kernel is normalized so that $\gamma_{xy \rightarrow ij}(t) = \beta e^{-\beta (t)}$, and $\beta$ is typically assumed to be a fixed rather than estimated parameter \citep{zhou2013learning, bacry2015hawkes, bacry2017tick}. 
While assuming $\beta$ to be fixed greatly simplifies estimation, a poor choice of $\beta$ may result in a much worse fit compared to estimating $\beta$.

One way to mitigate this possibility is to use a weighted sum of different kernels \citep{zhou2013a_learning}. 
Inspired by \citet{yang2017decoupling}, we use a sum of multiple exponential kernels with different decay rates. 
Let $\bm{\beta} = (\beta_1, \cdots, \beta_Q)$ denote a set of $Q$ fixed decays shared among all block pairs. For each $\bp(a,b)$, we introduce a block pair-specific kernel scaling parameter $\bm{C}_{ab} = \big(C_{ab}^1, \cdots, C_{ab}^Q\big)$ that is estimated simultaneously with $\bm{\Theta}$. 
For identifiability, we assume $\sum_{q=1}^Q C_{ab}^{q}=1$, and $C_{ab}^{q}\in[0,1]$. 
The full expression for the MULCH conditional intensity function with the sum of exponential kernels is provided in Appendix \ref{sec:supp_sum_of_kernels}.

\section{Estimation Procedure}

Fitting the MULCH model involves both estimation of node membership vector $\vec Z$ and Hawkes process parameters $\bm{\theta}_{ab}$ for all block pairs $(a,b)$. 
We first aggregate the number of events between each node pair $(i,j)$ to form entry $N_{ij}(T)$ of the count adjacency matrix $\bm{N}(T)$. 
We then apply spectral clustering to the count matrix $\bm{N}(T)$ to obtain an initial estimate of the node memberships $\vec Z$. 
Since the network is directed, we use singular vectors for spectral clustering according to the algorithm of \citet{Sussman2012} with the added step of row normalization of the singular vectors \citep{rohe2016co}, which is beneficial in the presence of degree heterogeneity.

The expectation of the count matrix has a block structure induced by MULCH, as shown in Figure \ref{fig:Block_structure_count}. The expected count matrix is simply $E[\bm{N}(T)] =\boldsymbol{\lambda} T$, where $\boldsymbol{\lambda}$ denotes the expected intensity function.
When the process is stationary, we have the following theorem to ensure that the expected count matrix has the block structure.
\begin{theorem}
\label{thm:block_structure}
For any $(i,j) \in (a,b)$ and $(i',j') \in (a,b)$, $\boldsymbol{\lambda}_{(i,j)} = \boldsymbol{\lambda}_{(i',j')} =  g_{ab} \mu_{ab} + g_{ba} \mu_{ba}$,
where $g_{ab }$ and $g_{ba}$ are real valued functions which depend on model parameters.
\end{theorem}
Proof of this theorem can be found in Appendix \ref{sec:supp_block_structure}. 
As the duration $T \rightarrow \infty$, the count matrix $\bm{N}(T)$ should approach the expected count matrix, and spectral clustering should succeed in recovering the correct node memberships.
For finite $T$, the estimated node memberships for spectral clustering may not be optimal, so we run an iterative refinement procedure using the likelihood to improve the estimated node memberships, which we describe in Section \ref{sec:refinement}. 

We do not have a theoretical guarantee for spectral clustering estimation accuracy due to the \emph{dependent} adjacency matrix entries in MULCH resulting from the different excitations in Table \ref{tab:excitationDescriptions}. 
This is unlike CHIP \citep{arastuie2020chip}, which had independent adjacency matrix entries and used proof techniques that assumed independence.

After we get the node membership $\vec Z$, we then estimate $\bm{\Theta}$, using a maximum likelihood estimation (MLE) approach. Under the MULCH assumptions, the model's log-likelihood can be expressed as a sum over block pair log-likelihoods
\begin{equation}
\label{eq:total_log_lik}
\ell(\bm \Theta| \vec Z, \mathcal{H}_t) = \sum_{a=1}^K\sum_{b=1}^K \ell_{ab}(\vec \theta_{ab}| \vec Z,\mathcal{H}_t)
\end{equation}
The detailed form for $\ell_{ab}(\vec \theta_{ab}| \vec Z, \mathcal{H}_t)$ is provided in Appendix \ref{sec:full_log_lik}.
Each block pair log-likelihood function $\ell_{ab}(\vec \theta_{ab}| \vec Z,\mathcal{H}_t)$ can be maximized by standard non-linear optimization methods. 
We use the L-BFGS-B \citep{byrd1995limited} optimizer implemented in SciPy. Parameters are initialized to small random numbers, and we set bounds to be between $(10^{-7}, +\infty)$ to ensure that they are positive.

\subsection{Likelihood Refinement Procedure}
\label{sec:refinement}
Next, we propose a likelihood refinement procedure using the MULCH model likelihood to improve the estimation of the node memberships starting from a spectral clustering initial solution. The algorithm is motivated by leave-one-out likelihood maximization procedures that are commonly employed in conjunction with spectral clustering to achieve minimax optimal error rates of community detection  in the context of stochastic block models and extensions \cite{gao2017achieving,gao2018community,chen2020global}. However, in contrast to the previous literature, we do not use the likelihood of the count matrix elements, but instead use the likelihood of the event times to obtain the refined node memberships. 

Our refinement procedure has 3 steps. First, we perform spectral clustering on the count matrix to get the initial node memberships $\vec Z^0$, and then estimate the parameters through maximum likelihood estimation as described in the previous section. Then for each node $i$, we assign it to the block which can maximize the log-likelihood \eqref{eq:total_log_lik} given the estimated node memberships of all other nodes and the estimated parameters. 
Finally, we re-estimate the model parameters using the new block assignment.

The refinement step involves computing a likelihood over significantly less number of events and is thus computationally efficient. 
To get better results, we run the refinement multiple times, each time using the latest refined estimates as the initial values. 
We end the refinement when no node memberships change or at a maximum of $15$ iterations.

\subsection{Model Selection}
Up to this point, we have assumed that both the decay values $\beta_q$ used in the sum of exponential kernels and the number of blocks $K$ are fixed and known. 
The $\beta_q$ values should be chosen appropriately for the type of network. 
For example, for networks with rapid dynamics such as instant messaging networks, suitable $\beta_q$ values will likely be on the order of minutes, hours, or days; whereas, for networks with slower dynamics, suitable $\beta_q$ values may be in the order of weeks or months. 
We choose three $\beta_q$ values in our experiments.

In most application settings, the number of blocks $K$ is usually also unknown. 
We use task-based model selection in our experiments; that is, we choose the value of $K$ that maximizes the evaluation metric we use for predictive accuracy (Section \ref{sec:exp_pred}) or generative accuracy (Section \ref{sec:exp_gen})

\begin{figure*}[t]
    \newcommand{\figwidth}{0.33\textwidth}
    \centering
    \hfill
    \begin{subfigure}[c]{\figwidth}
        \centering
        \includegraphics[width=\textwidth]{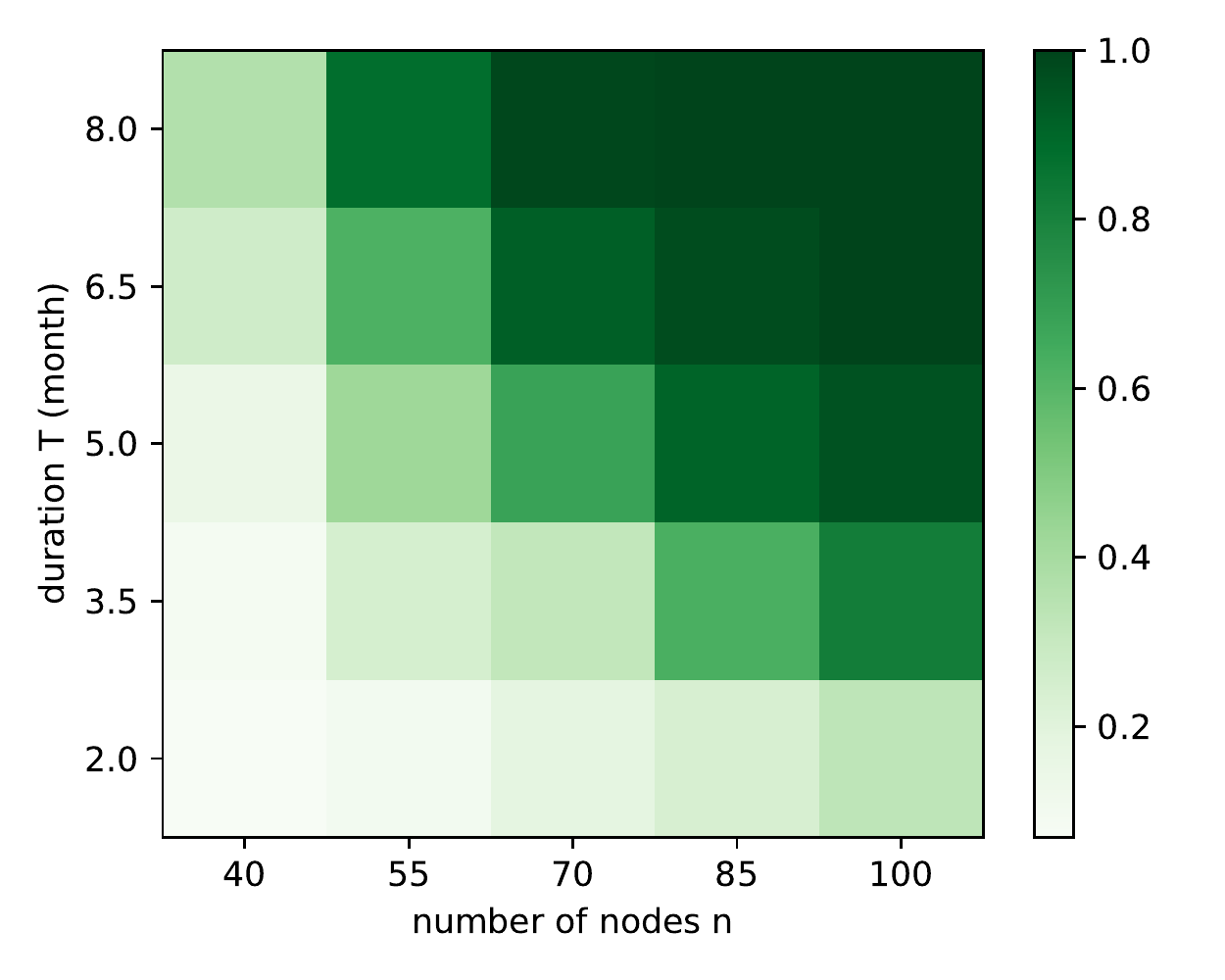}
        \caption{}
        \label{fig:RI_matrix}
    \end{subfigure}
    \hfill
    \begin{subfigure}[c]{\figwidth}
        \centering
        \includegraphics[width=\textwidth]{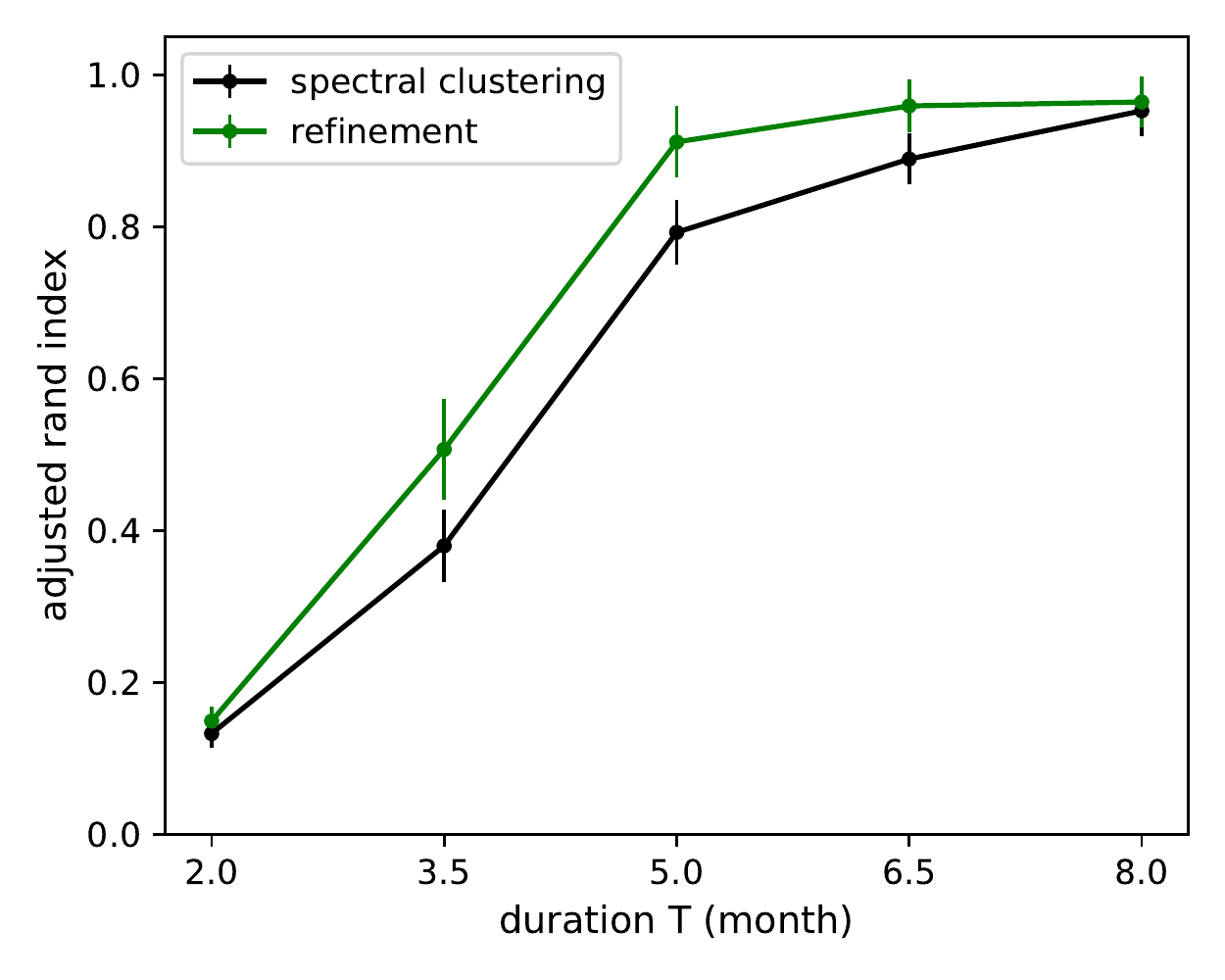}
        \caption{}
        \label{fig:RI_ref_sp_N70}
    \end{subfigure}
    \hfill
    \begin{subfigure}[c]{\figwidth}
        \centering
        \includegraphics[width=\textwidth]{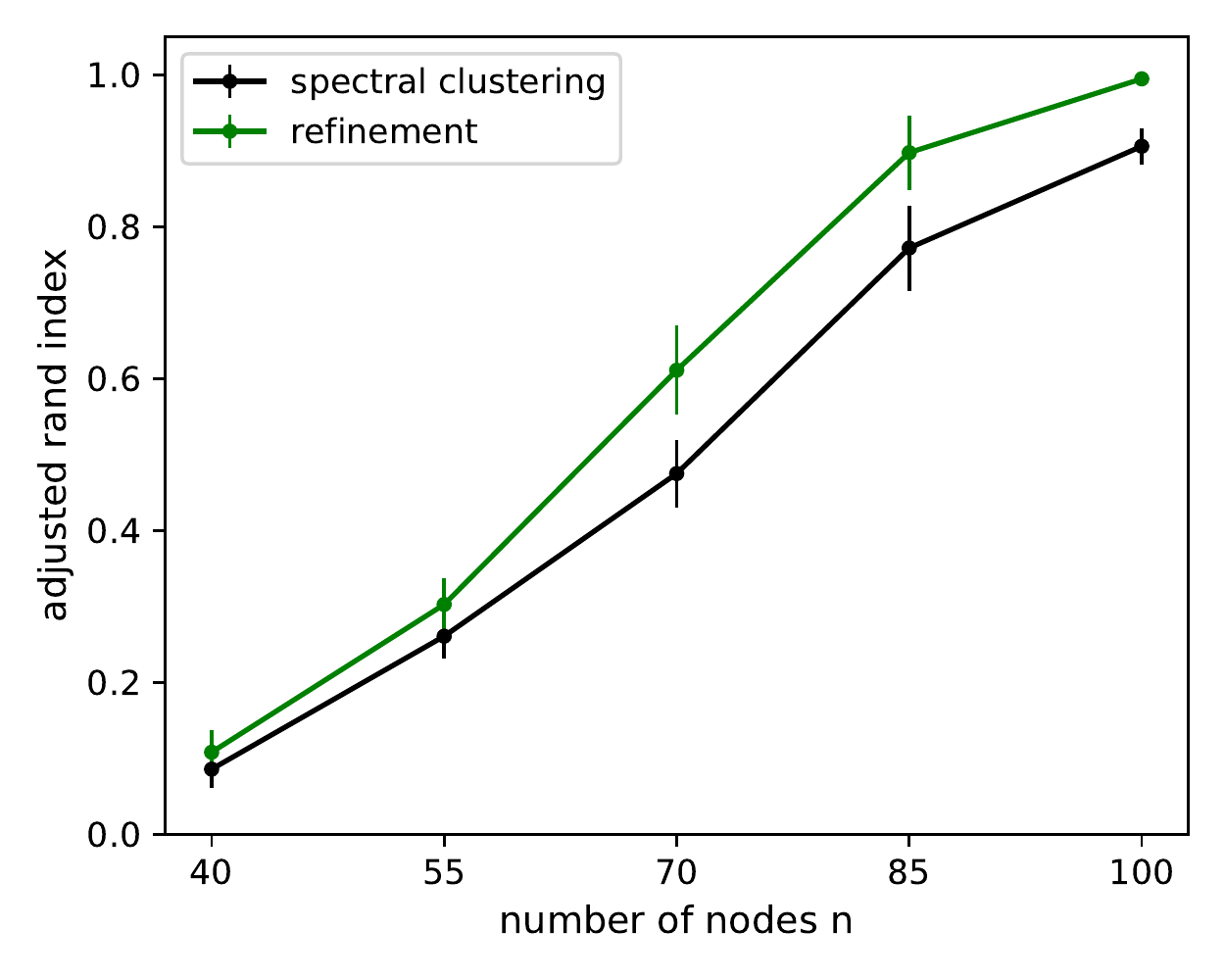}
        \caption{}
        \label{fig:RI_ref_sp_T105}
    \end{subfigure}
    \hfill
    \caption{Block estimation accuracy averaged over 10 simulated networks. 
    \subref{fig:RI_matrix} Heat map of adjusted Rand index for spectral clustering at $K=4$ while varying $T, n$.
    Comparison between adjusted Rand index achieved by spectral clustering and refinement algorithm for \subref{fig:RI_ref_sp_N70} fixed $n=70$, varying $T$ and \subref{fig:RI_ref_sp_T105} fixed $T=3.5$ months, varying $n$ ($\pm$ standard error over 10 runs).}
\end{figure*}

\section{Experiments}

\subsection{Simulated Networks}
We first test the ability of both spectral clustering and our likelihood refinement procedure to recover true node memberships on networks simulated from MULCH. In Appendix \ref{sec:supp_exp_param}, we present an additional experiment evaluating parameter estimation accuracy.

\paragraph{Spectral Clustering Accuracy}

We simulate networks at $K = 4$ while varying both $n$ and $T$. 
For each $(n, T)$ value, we simulate a network from the MULCH model, run spectral clustering on the count matrix, and calculate the adjusted Rand index \citep{hubert1985comparing} between the true and estimated node memberships, where a score of $1$ indicates perfect clustering and $0$ is the expected score for a random estimate. As shown in Figure \ref{fig:RI_matrix}, the accuracy of estimated node memberships improves as both $n, T$ increase. The average score over $10$ simulations is shown and indicates that spectral clustering can recover true node memberships for large $n$ and $T$.
Parameter values used to simulate the networks are listed in Appendix \ref{sec:supp_exp_rand}.

\paragraph{Likelihood Refinement Accuracy}

Figure \ref{fig:RI_ref_sp_N70} shows the adjusted Rand index from spectral clustering and from applying our refinement algorithm to networks simulated at $n=70$ and varying $T$. Similarly, Figure \ref{fig:RI_ref_sp_T105} shows the adjusted Rand index at $T=3.5$ months and varying $n$. Each point is averaged over $10$ simulations.
Notice that the adjusted Rand index always improves after applying our refinement algorithm, and in some cases, allows for perfect estimation of node memberships for cases where spectral clustering still makes errors.

\subsection{Real Networks}

\begin{table}[t]
    \centering
    \caption{Summary statistics of real network datasets}
    \label{tab:dataStats}
    \begin{tabular}{cccc}
    \toprule
    Dataset  & Nodes    & Total Events & Test Events \\
    \midrule
    Reality  & $70$     & $2,161$      & $661$ \\
    Enron    & $142$    & $4,000$      & $1,000$ \\
    MID      & $147$    & $5,117$      & $1,078$\\
    Facebook & $43,953$ & $852,833$    & $170,567$ \\
    \bottomrule
    \end{tabular}
\end{table}

We perform benchmark experiments on 4 real network datasets to evaluate the predictive and generative accuracy of our proposed MULCH model against several other models\footnote{Python code is available at \url{https://github.com/IdeasLabUT/Multivariate-Community-Hawkes}}. 
Summary statistics for the datasets are shown in Table \ref{tab:dataStats}, with additional details in Appendix \ref{sec:supp_datasets}.  
Each dataset consists of a set of events where each event is denoted by a sender, a receiver, and a timestamp.

\paragraph{Models for Comparison}
We compare against several other TPP models for continuous-time networks:
REM \citep{Dubois2013}, BHM \citep{junuthula2019block}, CHIP \citep{arastuie2020chip}, DLS \citep{yang2017decoupling}, and ADM4 \citep{zhou2013learning}. 
Each of these models can be fit to a network and used to evaluate test log-likelihood on future events and to simulate networks from the fit.
REM, BHM, CHIP, and DLS are continuous-time network models, while ADM4 is a sparse and low-rank regularized model for general multivariate Hawkes processes. 
Additional details on these models is provided in Appendix \ref{sec:supp_other_models}.
We present also experiments on scalability and on other parameterizations for MULCH in Appendices \ref{sec:scalibility} and \ref{sec:supp_exp_alphas}. 

\subsubsection{Predictive Accuracy}
\label{sec:exp_pred}
\paragraph{Experiment Set-up}
We first evaluate the ability of our proposed MULCH model to predict future events between nodes. 
To do this, we split the data into training and test sets, with the first $l_{\text{train}}$ events being used to fit the model and the remaining $l_{\text{test}}$ events (shown in Table \ref{tab:dataStats}) being used to evaluate the model's predictive ability. 
We assign all new nodes present in the test set but not the training set to the largest block in the training set, consistent with \citet{arastuie2020chip}. 
For the DLS model, we randomly sample latent positions for new nodes from a multivariate Gaussian.

We consider two evaluation metrics previously established in the literature.
The first is the mean test data log-likelihood per event \citep{Dubois2013, arastuie2020chip}. 
We use the same train and test splits as in \citet{Dubois2013} for the Reality and Enron datasets, which allows us to compare against their reported results. 
The second evaluation metric is the area under the receiver operating characteristic curve (AUC) for dynamic link prediction. 
Specifically, we adopt the dynamic link prediction setting proposed by  \citet{yang2017decoupling}. 
We divide the test set into $100$ random short time windows and compute the mean and standard deviation of the link prediction AUC over the $100$ windows.

\begin{table}[t]
	\centering
    \caption{Mean test log-likelihood per event for each real network dataset across all models. Larger (less negative) values indicate better predictive ability. 
    Bold entry denotes highest accuracy for a dataset. 
    Results for REM are reported values from \citet{Dubois2013}, so results on MID and Facebook are not available.
    DLS does not scale to Facebook; ADM4 does not scale beyond Reality.}
    \label{tab:testLogLik}
    \begin{tabular}{ccccc}
    \toprule
    Model & Reality      & Enron        & MID          & Facebook \\
    \midrule
    MULCH & $\bm{-3.82}$ & $\bm{-5.13}$ & $\bm{-3.53}$ & $\bm{-6.82}$ \\
    CHIP  & $-4.83$      & $-5.61$      & $-3.67$      & $-9.46$ \\
    BHM   & $-5.37$      & $-7.49$      & $-5.33$      & $-14.4$ \\
    DLS   & $-5.74$      & $-7.75$      & $-5.52$      &  \\
    REM   & $-6.11$      & $-6.84$      &              &  \\
    ADM4  & $-8.52$      &              &              &  \\
    \bottomrule
    \end{tabular}

\end{table}

\begin{table}[t]
    \setlength{\tabcolsep}{3pt}
	\centering
    \caption{Dynamic link prediction AUC for each real network dataset across all models. Mean (standard deviation) of AUC over 100 random short time windows is shown. 
    Bold entry denotes highest mean link prediction AUC for a dataset.}
    \label{tab:dlpAUC}
    \begin{tabular}{cccc}
    \toprule
    Model & Reality             & Enron               & MID \\
    \midrule
    MULCH & $\bm{0.954 (.036)}$ & $0.852 (.006)$      & $0.968 (.023)$ \\
    CHIP  & $0.931 (.033)$      & $0.792 (.005)$      & $0.966 (.030)$ \\
    BHM   & $0.951 (.035)$      & $0.846 (.005)$      & $0.973 (.022)$ \\
    DLS   & $0.935 (.034)$      & $\bm{0.872 (.001)}$ & $\bm{0.981 (.013)}$ \\
    \bottomrule
    \end{tabular}
\end{table}

\paragraph{Results and Discussion}
The predictive abilities of the different models are summarized in Table \ref{tab:testLogLik}. 
Notice that MULCH achieves the highest test log-likelihood on all 4 datasets, and by a large margin on the Reality and Facebook data. 
CHIP performed second best on all of the datasets, indicating the importance of self excitation. 

MULCH also includes reciprocal and other excitations, which is partially responsible for the improved predictive ability. 
We find that using a sum of exponential kernels in MULCH also helps to improve predictive log-likelihood compared to a single exponential kernel in CHIP and BHM.
The importance of the SBM-structured excitation matrix $\alphaMat$ in MULCH is clearly visible when comparing it to ADM4, which estimates a sparse and low rank $\alphaMat$ without additional network structure and is not competitive.

Table \ref{tab:dlpAUC} shows the dynamic link prediction AUC values\footnote{We exclude the Facebook data because the dynamic link prediction experiment does not scale to a network of its size.}. 
MULCH performs the best on Reality and is competitive on the other two datasets. 
Notice that the BHM performs better than CHIP in dynamic link prediction AUC, while it was substantially worse than CHIP in test log-likelihood in Table \ref{tab:testLogLik}. 
This is due to a difference in the evaluation metrics---test log-likelihood is evaluated on all events, so that repeated events between node pairs are counted multiple times. Conversely, dynamic link prediction considers only whether a pair of nodes had at least a single event within a short time window, so each node pair is only counted once.

\subsubsection{Generative Accuracy}
\label{sec:exp_gen}

\begin{figure}[t]
\centering
\includegraphics[width=2.6in]{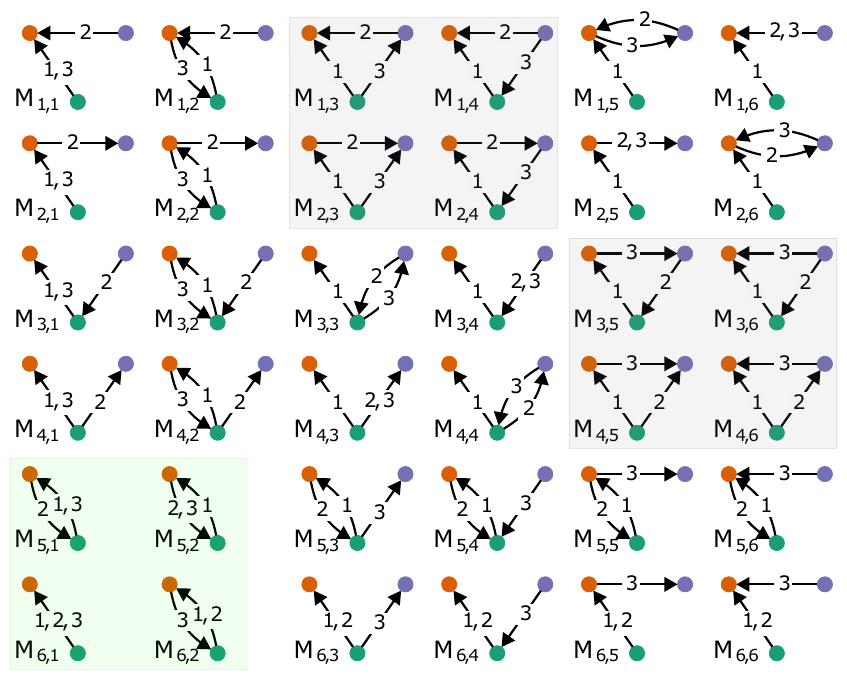}
\caption{All possible 3-edge temporal motifs. Green and grey shaded boxes denotes 2-node and triangle motifs, respectively. 
All other motifs are stars.
Figure credit: \citet{paranjape2017motifs}. }
\label{fig:Temporal_motifs}
\end{figure}

\begin{figure*}[t]
    \newcommand{\figwidth}{0.245\textwidth}
    \centering
    \hfill
    \begin{subfigure}[c]{\figwidth}
        \centering
        \includegraphics[width=\textwidth]{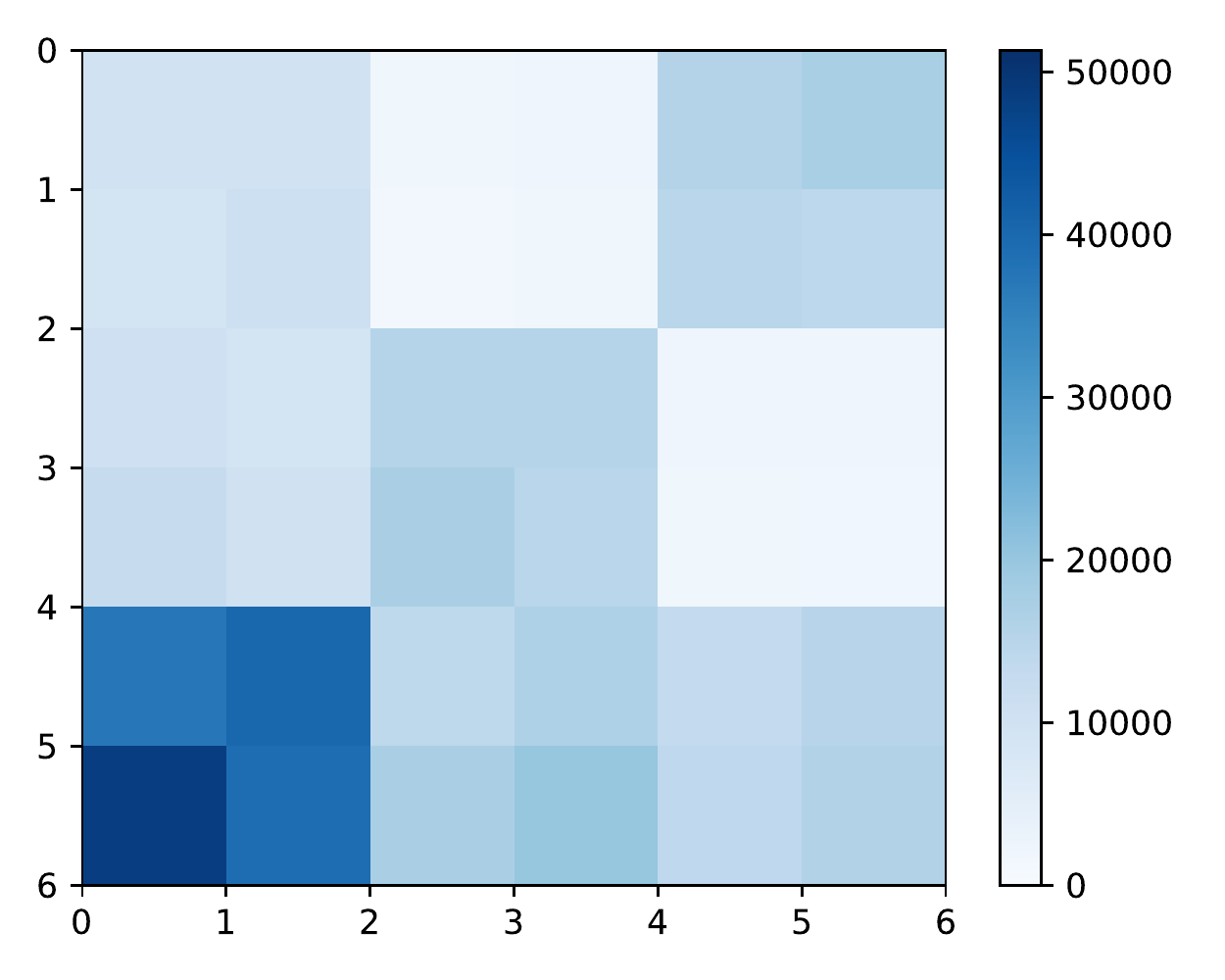}
        \caption{Reality: Actual network}
        \label{fig:MotifCountsR}
    \end{subfigure}
    \hfill
    \begin{subfigure}[c]{\figwidth}
        \centering
        \includegraphics[width=\textwidth]{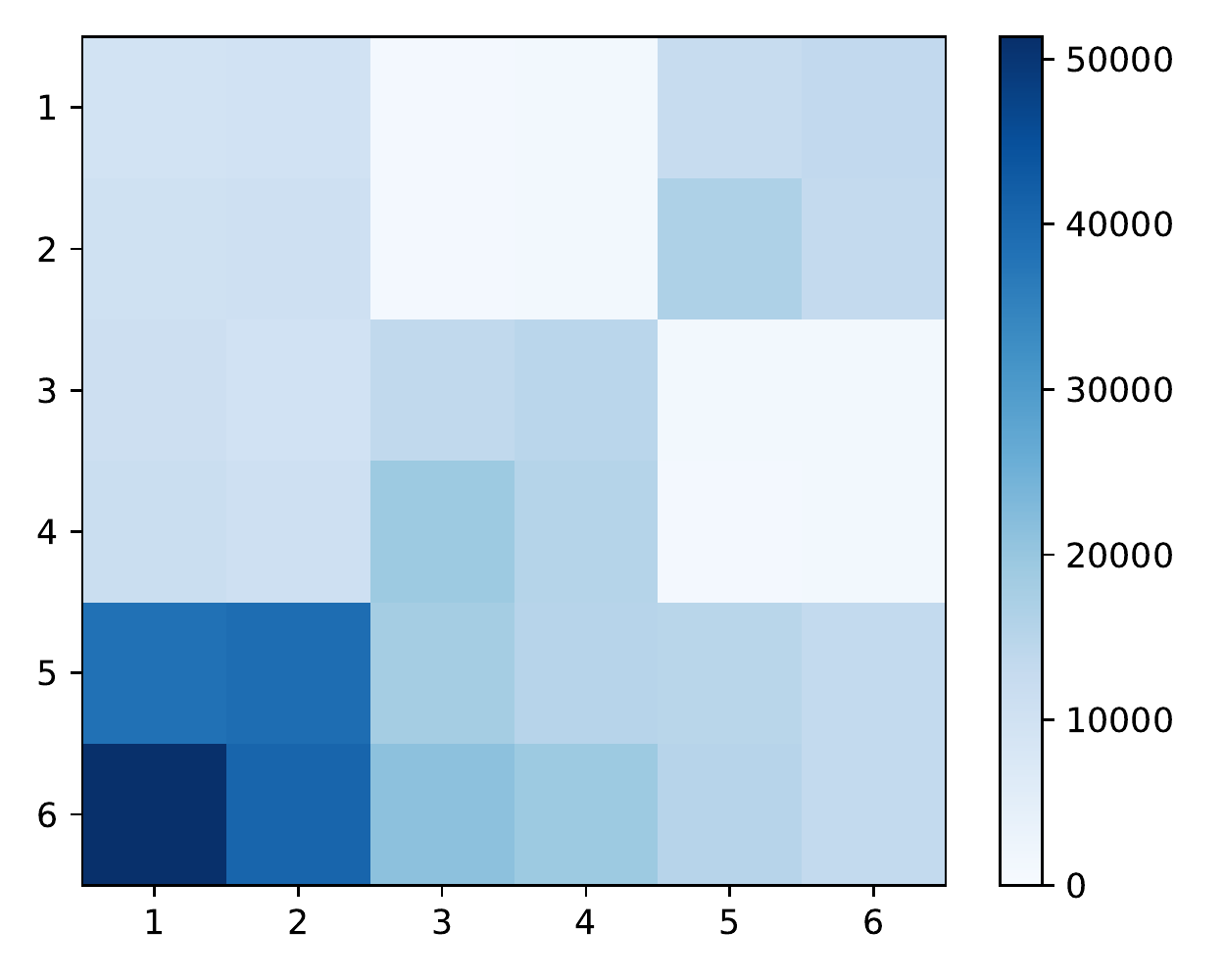}
        \caption{Reality: MULCH simulated}
        \label{fig:ModelMotifCountsR}
    \end{subfigure}
    \hfill
    \begin{subfigure}[c]{\figwidth}
        \centering
        \includegraphics[width=\textwidth]{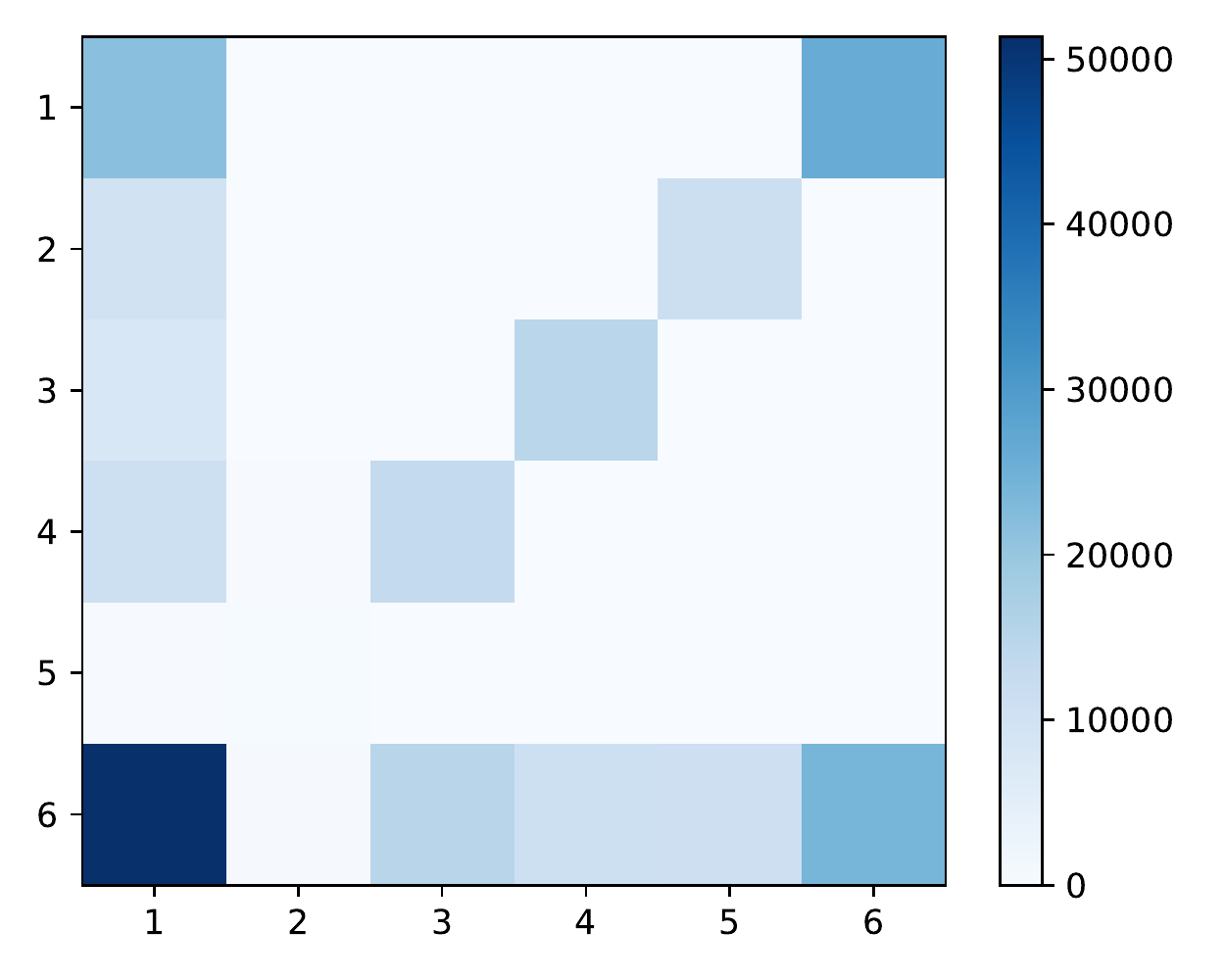}
        \caption{Reality: CHIP simulated}
        \label{fig:CHIPMotifCountsR}
    \end{subfigure}
    \hfill
    \begin{subfigure}[c]{\figwidth}
        \centering
        \includegraphics[width=\textwidth]{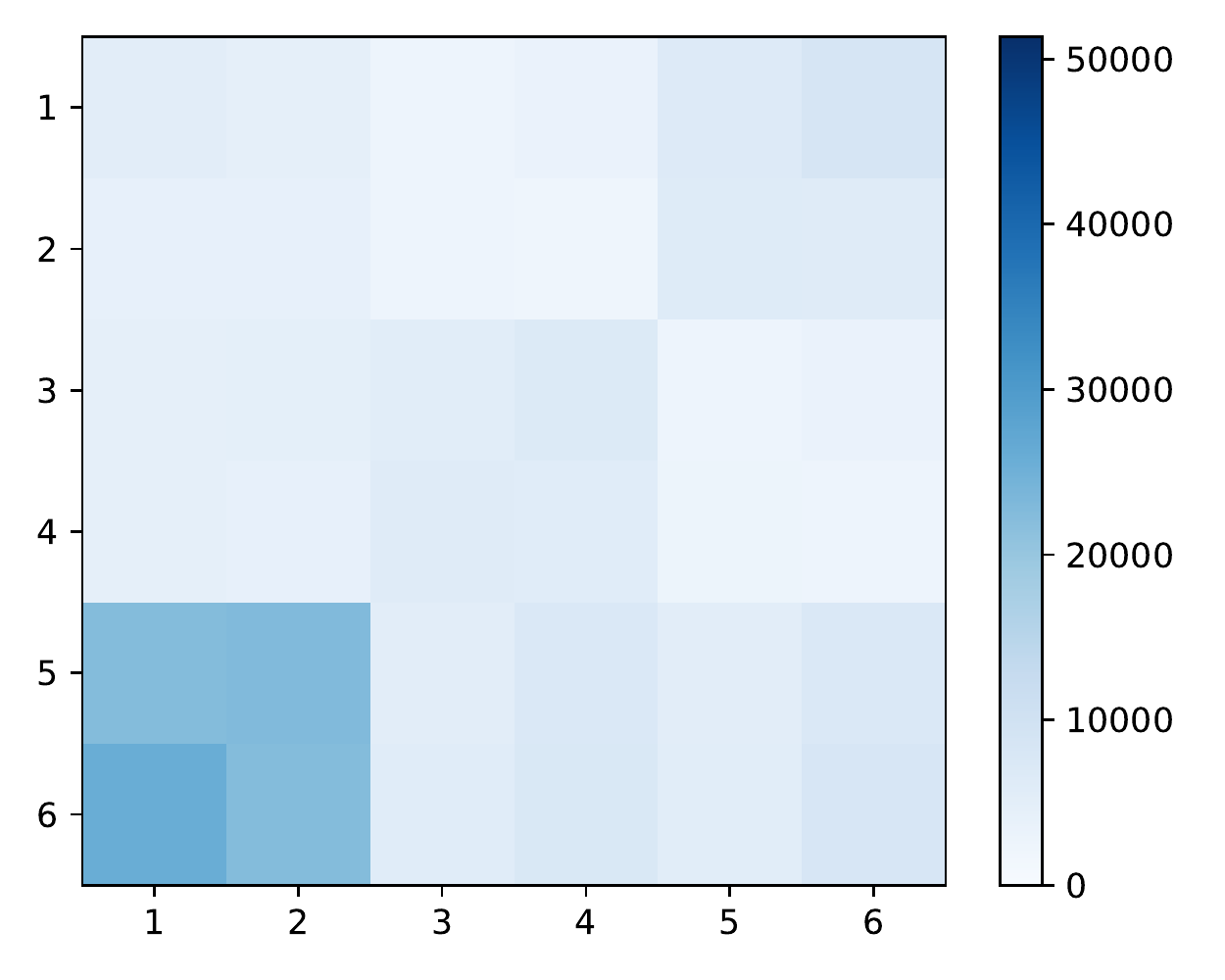}
        \caption{Reality: BHM simulated}
        \label{fig:BHMMotifCountsR}
    \end{subfigure}
    \hfill
    \\[6pt]
    \hfill
    \begin{subfigure}[c]{\figwidth}
        \centering
        \includegraphics[width=\textwidth]{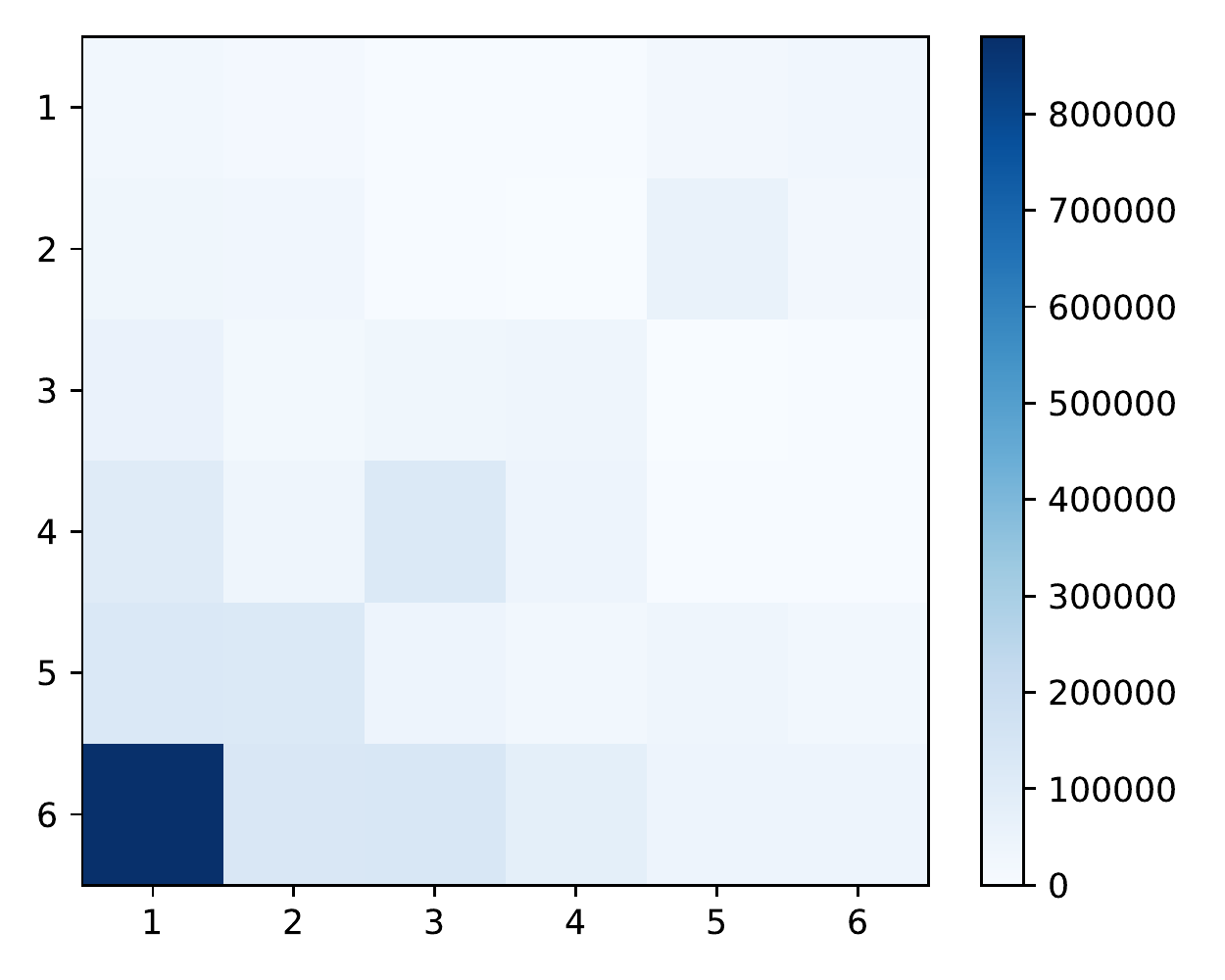}
        \caption{Enron: Actual network}
        \label{fig:MotifCountsE}
    \end{subfigure}
    \hfill
    \begin{subfigure}[c]{\figwidth}
        \centering
        \includegraphics[width=\textwidth]{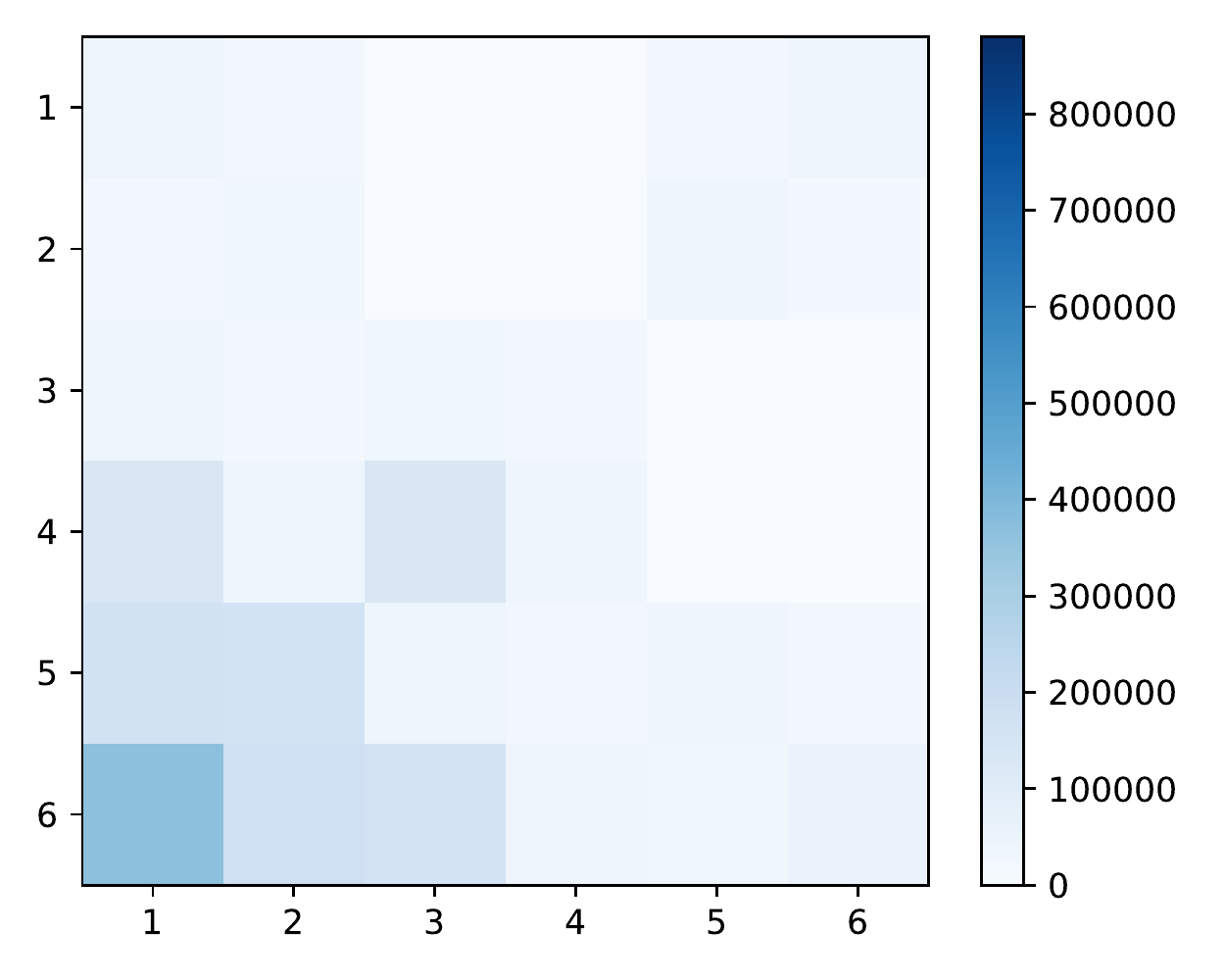}
        \caption{Enron: MULCH simulated}
        \label{fig:ModelMotifCountsE}
    \end{subfigure}
    \hfill
    \begin{subfigure}[c]{\figwidth}
        \centering
        \includegraphics[width=\textwidth]{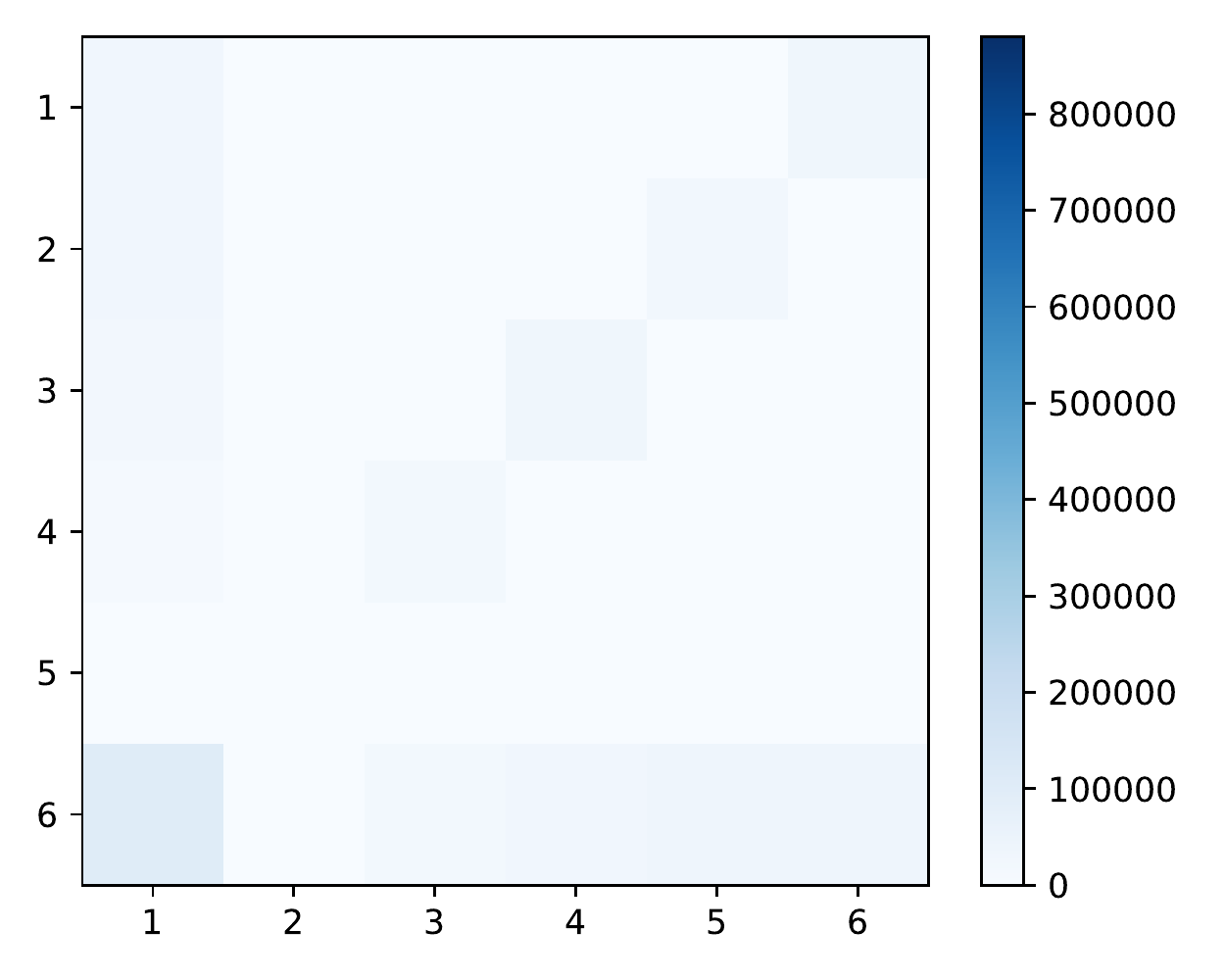}
        \caption{Enron: CHIP simulated}
        \label{fig:CHIPMotifCountsE}
    \end{subfigure}
    \hfill
    \begin{subfigure}[c]{\figwidth}
        \centering
        \parbox{\textwidth}{}
    \end{subfigure}
    \hfill
    \caption{Average temporal motif counts for time window $\delta=1$ week on $10$ simulated networks from MULCH, CHIP, and BHM models fitted on the \subref{fig:MotifCountsR}-\subref{fig:BHMMotifCountsR} Reality and \subref{fig:MotifCountsE}-\subref{fig:CHIPMotifCountsE} Enron datasets. 
    The simulations are unstable for the BHM model fit to Enron.
    Our proposed MULCH model best replicates the temporal motif counts from the actual networks.}
    \label{fig:AllMotifCounts}
\end{figure*}

\paragraph{Experiment Set-up}
We evaluate the generative accuracy of our proposed MULCH model by simulating networks from the fitted model and comparing the counts of temporal motifs in the simulated networks to those in the actual network. 
We consider all 36 possible temporal motifs with 2 or 3 nodes and 3 edges arranged in the same $6 \times 6$ matrix as defined by \citet{paranjape2017motifs}. 
The matrix of different motifs is shown in Figure \ref{fig:Temporal_motifs}. 
One would expect a good generative model to replicate the number of temporal motifs observed in the actual network. 
We consider temporal motifs over 1 week for the Reality and Enron datasets and 1 month for the MID data\footnote{We do not evaluate generative accuracy on the Facebook data due to its size. 
We also do not include DLS in this comparison because it resulted in unstable models that cannot be used to generate new networks, also been noted by \citet{huang2022mutually}.}.

\begin{table}[t]
	\centering
    \caption{Mean absolute percentage error (MAPE) on temporal motif counts for each real network dataset across all models. Smaller values indicate better generative ability. 
    Bold entry denotes best fit for a dataset. 
    The BHM fit to the Enron data results in an unstable Hawkes process that cannot simulate networks.}
    \label{tab:mapeGenerative}
    \begin{tabular}{cccc}
    \toprule
    Model & Reality      & Enron        & MID \\
    \midrule
    MULCH & $\bm{16.5}$  & $\bm{32.0}$  & $92.3$ \\
    CHIP  & $79.3$       & $74.5$       & $\bm{91.0}$\\
    BHM   & $51.1$       &              & $97.6$ \\
    \bottomrule
    \end{tabular}

\end{table}

To provide a quantitative assessment of generative ability, we compute the mean absolute percentage error (MAPE) on the temporal motif count matrix. 
The MAPE is defined by 
\begin{equation*}
\text{MAPE} = \frac{100}{36} \sum_{i=1}^6 \sum_{j=1}^6 \left| \frac{M_{i,j}^A - M_{i,j}^S}{M_{i,j}^A} \right|,
\end{equation*}
where $M_{i,j}^A$ denotes the number of occurrences of motif $M_{i,j}$ in the actual network, and $M_{i,j}^S$ denotes the mean number of occurrences of the motif over $10$ simulated networks.

\paragraph{Results and Discussion}
The MAPE values for the different models and datasets are shown in Table \ref{tab:mapeGenerative}. 
MULCH is by far the best at replicating temporal motif counts on Reality and Enron, while all of the models struggle on the MID data.
\citet{do2021analyzing} found that the majority of motifs occurred during several major international conflicts in the years 1999 and 2000, which are not accurately replicated by any of the models. 
More sophisticated models that incorporate change points may be required to capture these dynamics.

\begin{figure*}[t]
    \newcommand{\figwidth}{0.245\textwidth}
    \centering
    \hfill
    \begin{subfigure}[c]{\figwidth}
        \centering
        \includegraphics[width=\textwidth]{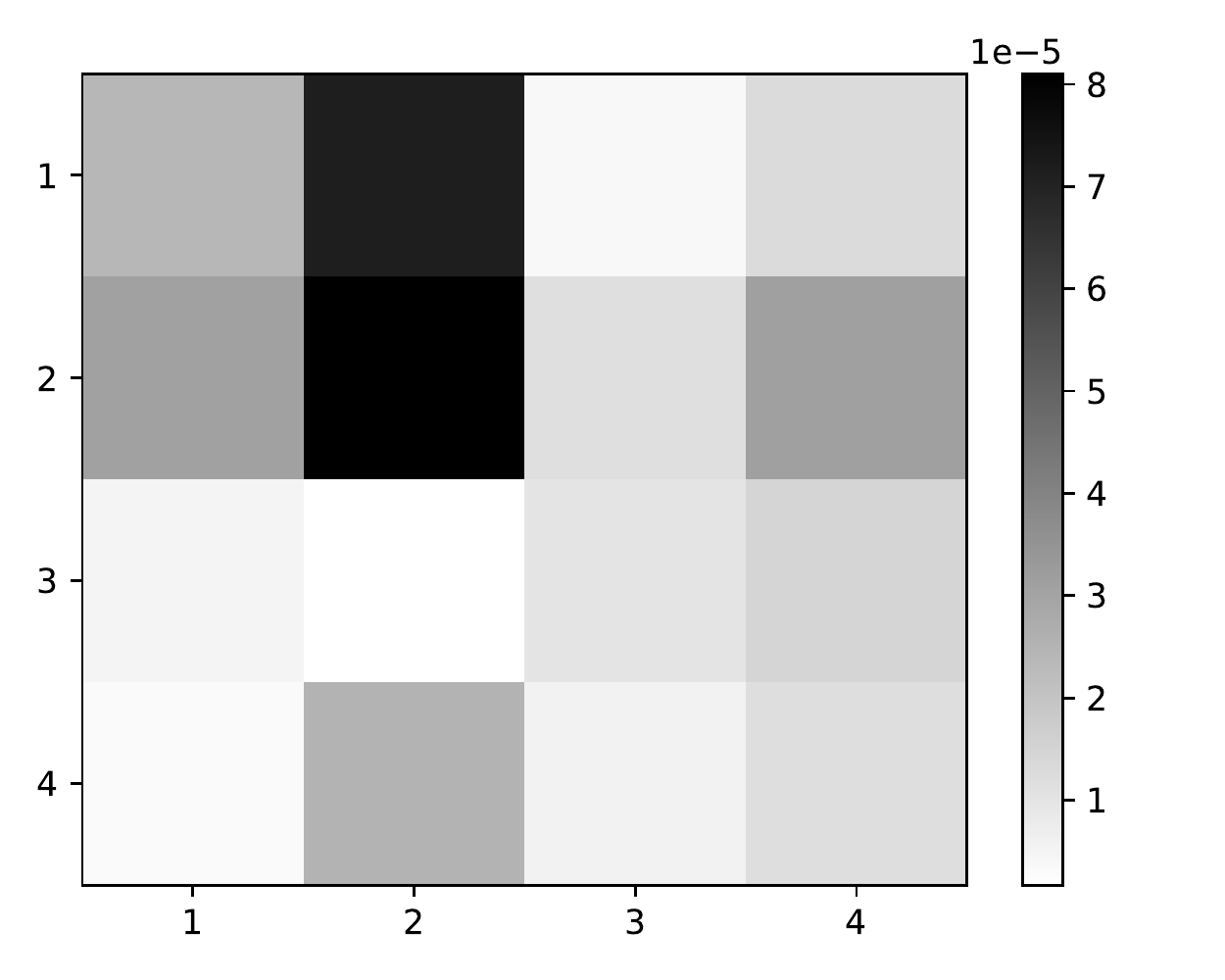}
        \caption{Base intensity $\mu_{ab}$}
        \label{fig:MID_mu_ht}
    \end{subfigure}
    \hfill
    \begin{subfigure}[c]{\figwidth}
        \centering
        \includegraphics[width=\textwidth]{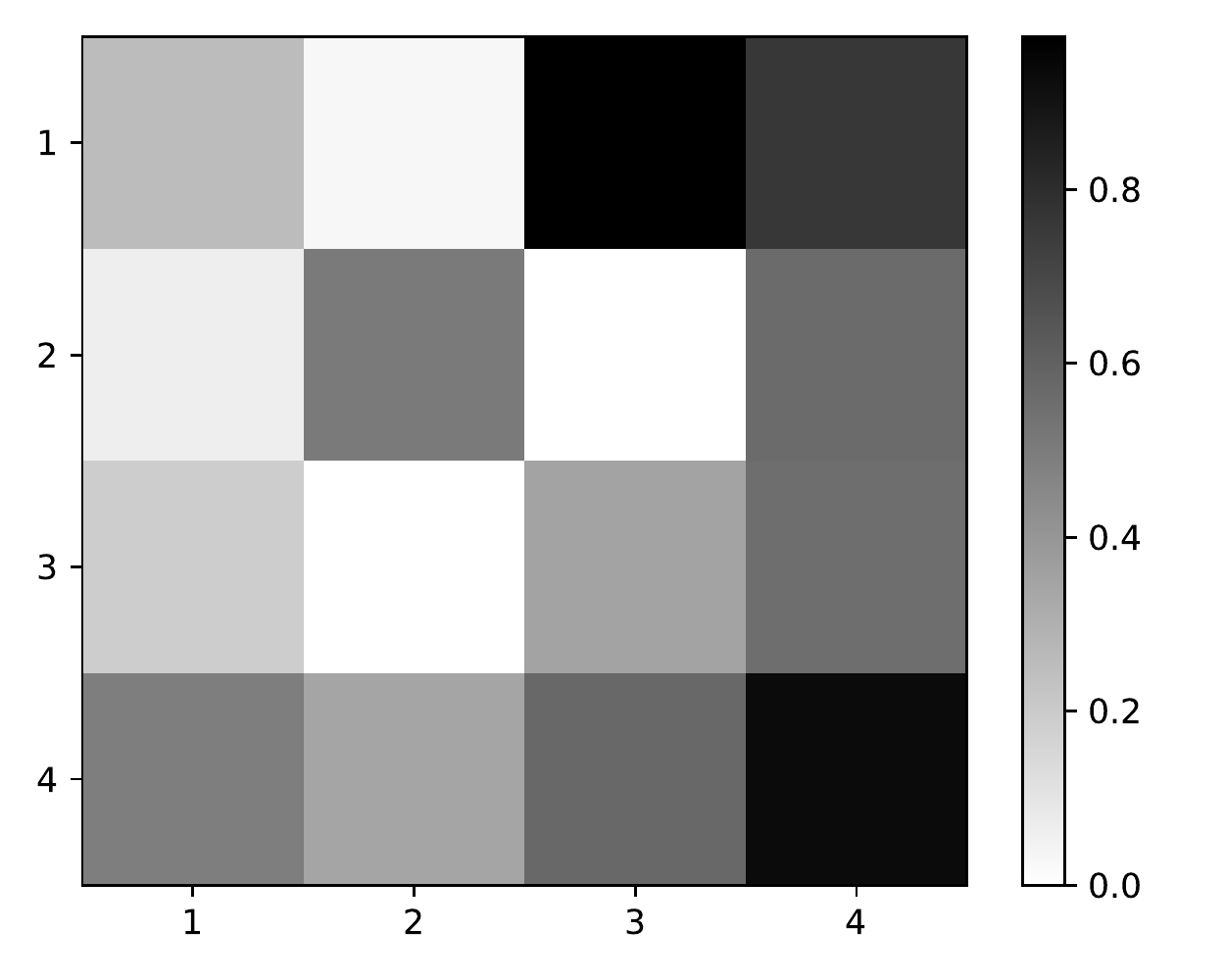}
        \caption{Self excitation $\alpha_{ab}^{xy \rightarrow xy}$}
        \label{fig:MID_s_ht}
    \end{subfigure}
    \hfill
    \begin{subfigure}[c]{\figwidth}
        \centering
        \includegraphics[width=\textwidth]{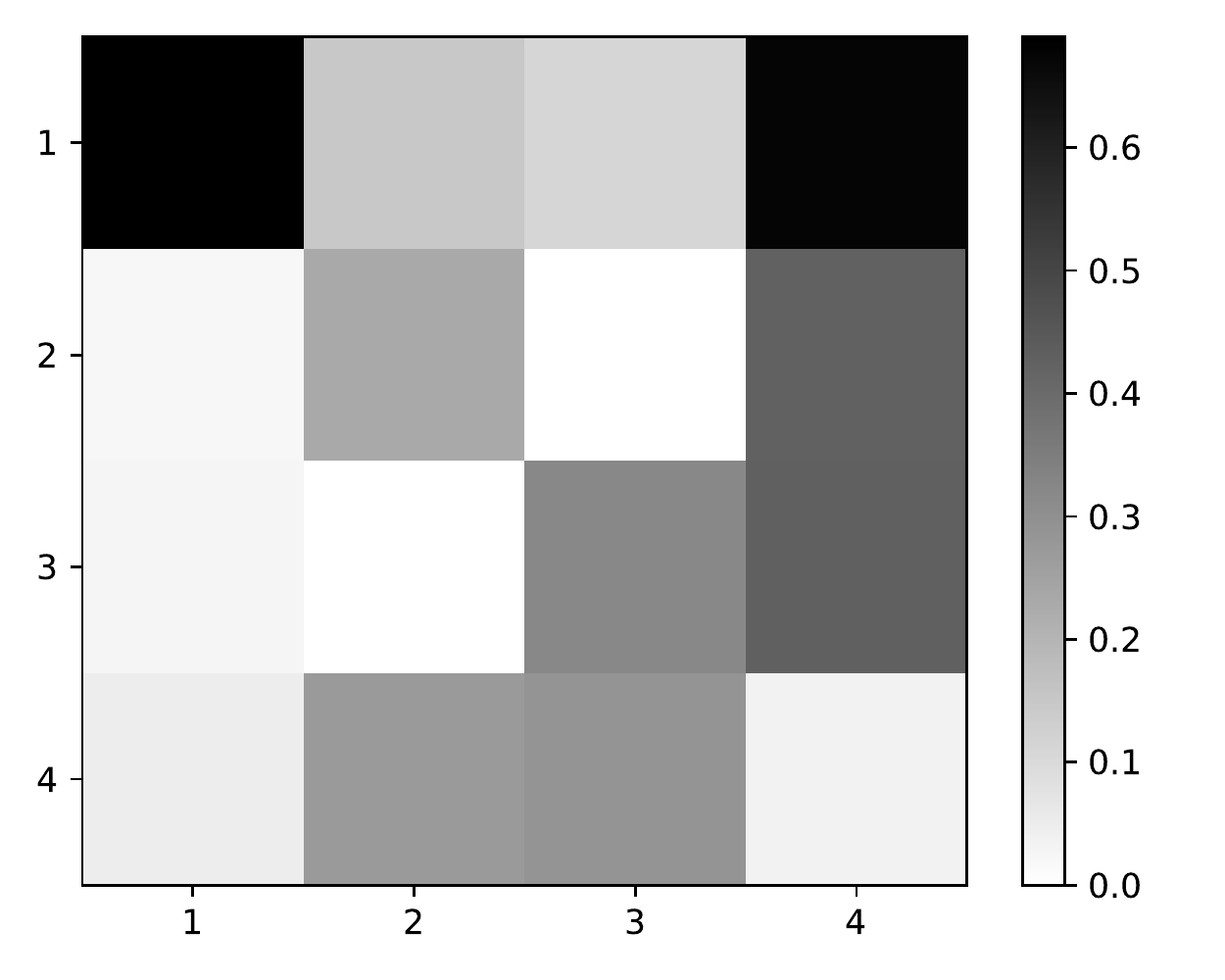}
        \caption{Reciprocal excitation $\alpha_{ab}^{xy \rightarrow yx}$}
        \label{fig:MID_r_ht}
    \end{subfigure}
    \hfill
    \begin{subfigure}[c]{\figwidth}
        \centering
        \includegraphics[width=\textwidth]{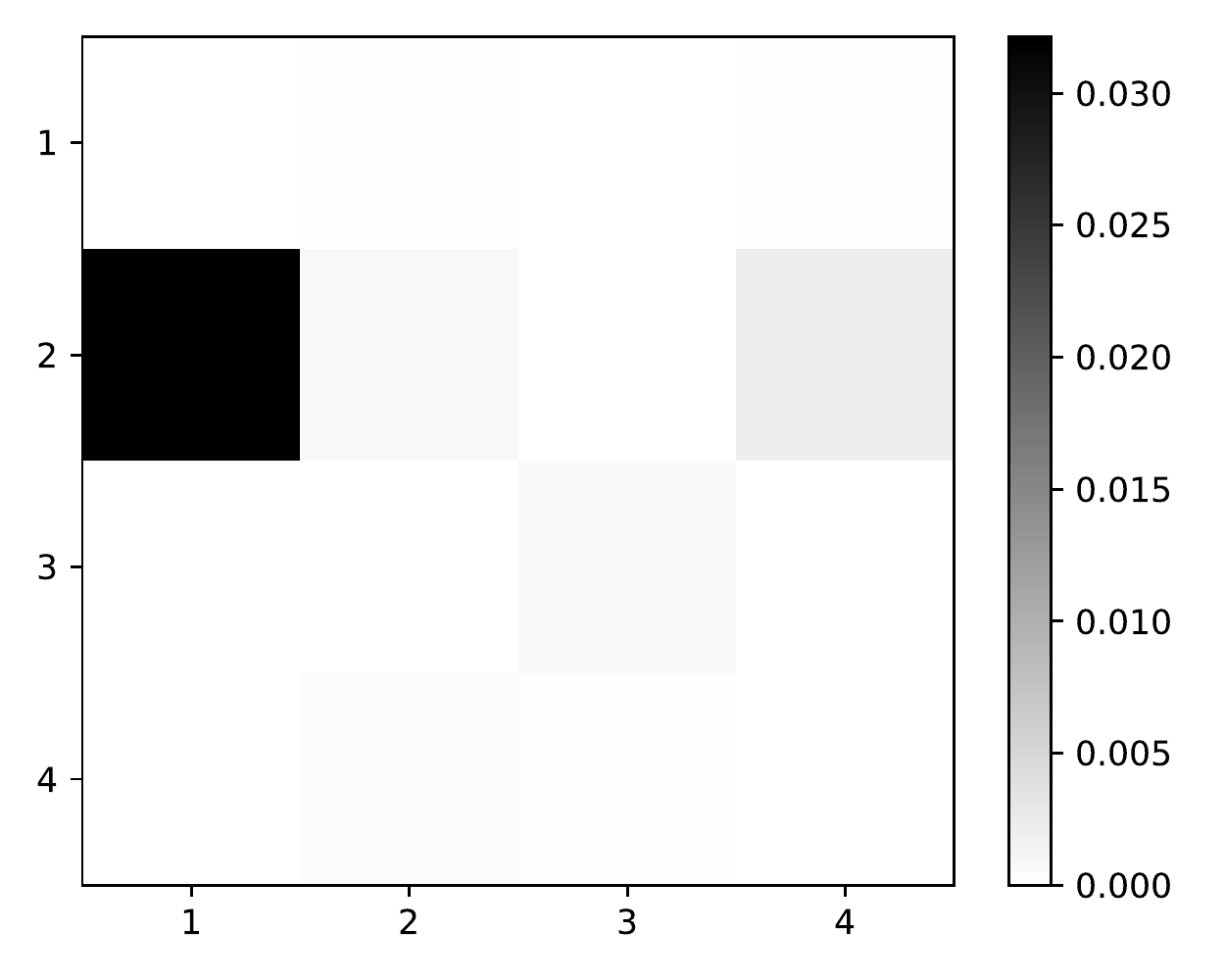}
        \caption{Turn continuation $\alpha_{ab}^{xy \rightarrow xb}$}
        \label{fig:MID_br_ht}
    \end{subfigure}
    \hfill
    \begin{subfigure}[c]{\figwidth}
        \centering
        \includegraphics[width=\textwidth]{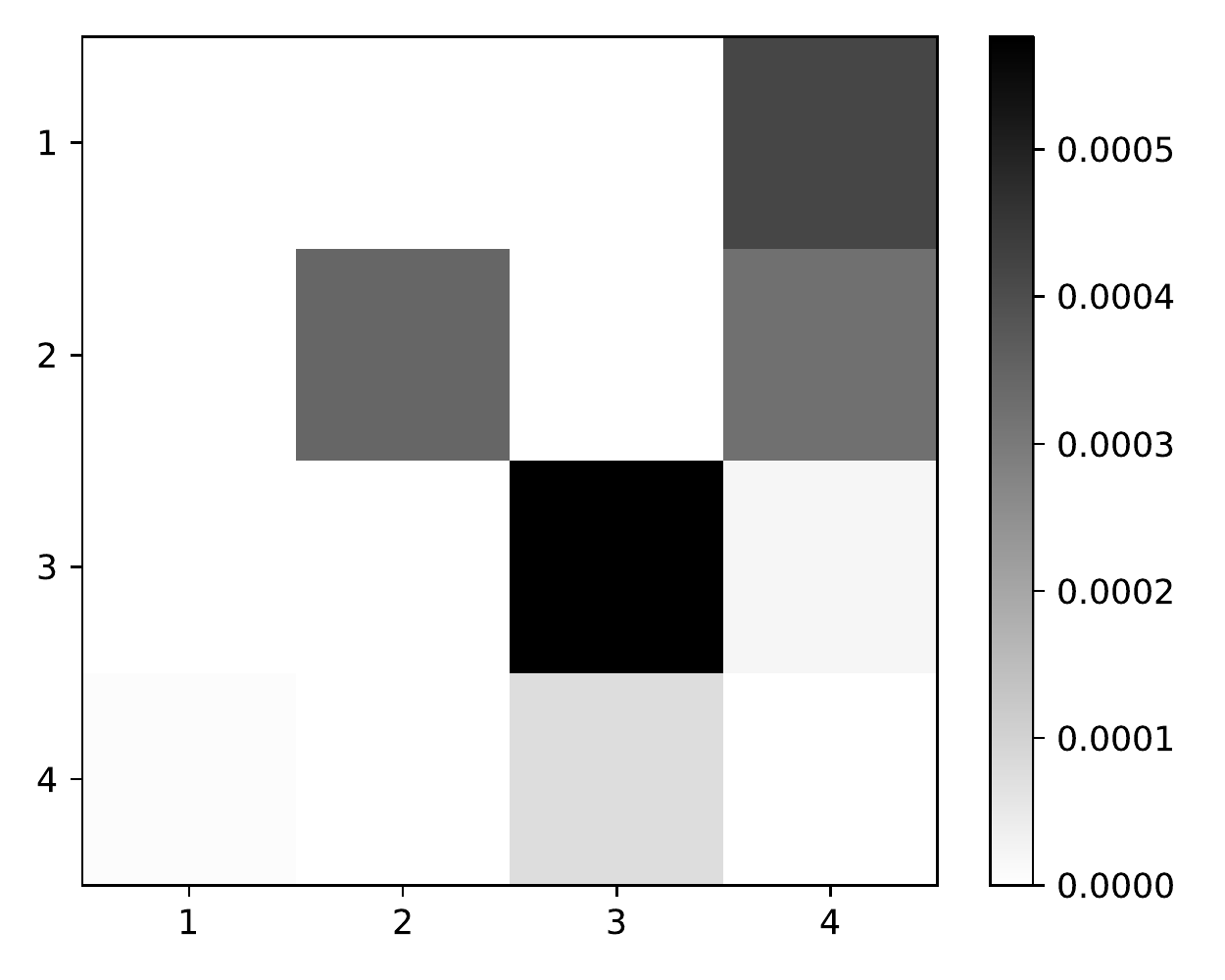}
        \caption{Generalized recip.~$\alpha_{ab}^{xy \rightarrow ya}$}
        \label{fig:MID_gr_ht}
    \end{subfigure}
    \hfill
    \begin{subfigure}[c]{\figwidth}
        \centering
        \includegraphics[width=\textwidth]{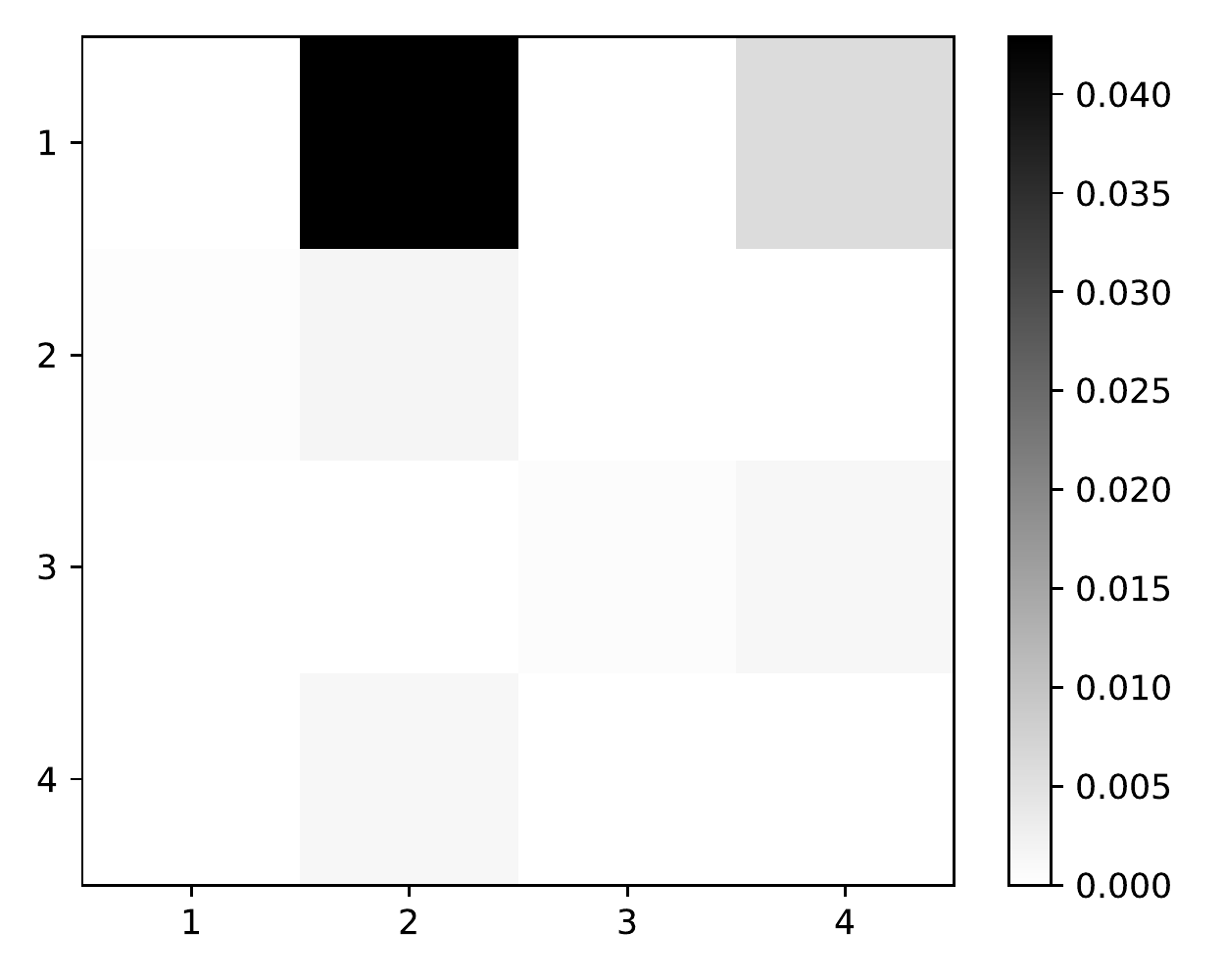}
        \caption{Allied continuation $\alpha_{ab}^{xy \rightarrow ay}$}
        \label{fig:MID_al_ht}
    \end{subfigure}
    \hfill
    \begin{subfigure}[c]{\figwidth}
        \centering
        \includegraphics[width=\textwidth]{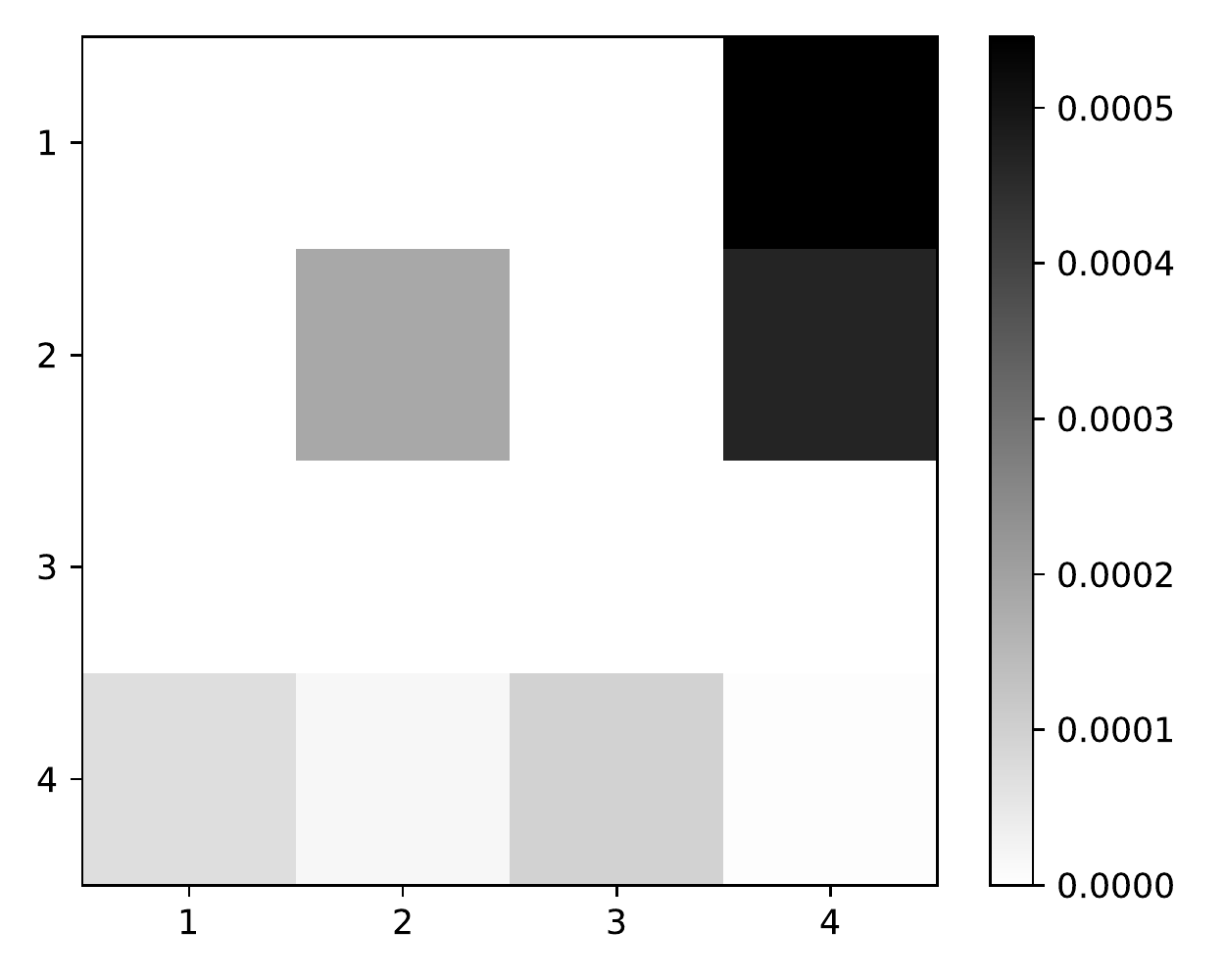}
        \caption{Allied reciprocity $\alpha_{ab}^{xy \rightarrow bx}$}
        \label{fig:MID_alr_ht}
    \end{subfigure}
    \hfill
    \begin{subfigure}[c]{\figwidth}
        \centering
        \includegraphics[width=\textwidth]{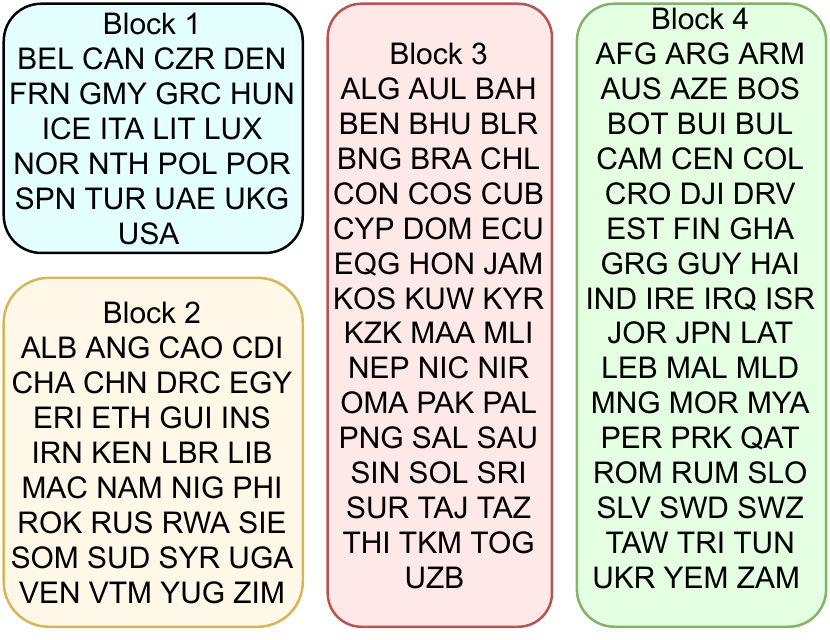}
        \caption{Node memberships}
        \label{fig:MID_blocks}
    \end{subfigure}
    \hfill
    \caption{Estimated values for MULCH parameter estimates and node memberships fit to the Military Interstate Disputes data. 
    Nodes are represented by their 3-letter country codes from \citet{cowcountrycodes}.}
    \label{fig:HT}
\end{figure*}

In Figure \ref{fig:AllMotifCounts}, we show a comparison of the average temporal motif counts on simulated networks from the MULCH, CHIP, and BHM model fits to the Reality and Enron datasets to the actual counts. 
Notice that CHIP, which uses only self excitation, does not generate motifs with reciprocated edges with any appreciable frequency. 
The BHM uses self excitation at the block level and then randomly assigns events to node pairs, which results in the uniform-like distribution of 3-node temporal motifs seen in Figure \ref{fig:BHMMotifCountsR}. 
It generates a wide variety of motifs, but not at frequencies similar to the actual network. 
On the other hand, our proposed MULCH model generates both a variety of different temporal motifs and at relative frequencies similar to the actual network, as shown by the 
similarities of the checkerboard patterns between Figures \ref{fig:MotifCountsR} and \ref{fig:ModelMotifCountsR} as well as Figures \ref{fig:MotifCountsE} and \ref{fig:ModelMotifCountsE} for the Enron data. 
This indicates that the added excitations in MULCH that create dependence between node pairs can indeed replicate higher-order motifs found in real networks.

\section{Case Study}
\label{sec:caseStudy}

We now present a case study on the Militarized Interstate Disputes (MID) dataset. 
A node denotes a (sovereign) state. 
An edge denotes a threat, display, or use of force one state directs towards another.
We apply MULCH to perform model-based exploratory analysis using the different excitations to reveal insights into behaviors of different states.

Our parameter and block estimates are shown in Figure \ref{fig:HT}. 
From examining the magnitudes of the different excitations, one can see that self and reciprocal excitation are the strongest, with highest $\alpha$ values around $0.9$ and $0.7$, respectively. 
Next strongest are allied continuation and turn continuation, with highest $\alpha$ values around $0.04$ and $0.03$, respectively, which is one order of magnitude weaker. 
Generalized reciprocity and allied reciprocity have much lower $\alpha$ values, with the highest around $0.0005$ for both, two orders of magnitude weaker than allied and turn continuation.

Examining the combination of node memberships and parameter estimates also reveals some interesting insights. 
Events in this network correspond to disputes between states, and we find that most events occur between blocks. 
Every member of block 1 is a member state of the North Atlantic Treaty Organization (NATO) except the UAE, which has partnered with NATO in several disputes, including the intervention in Libya in 2011.

We observe that block pair $(1,2)$ has by far the highest excitation for allied continuation, indicating that nodes in block 1 (NATO members) tend to jointly engage states in block 2. 
By consulting the narratives that accompany the MID dataset, indeed we find that two of the most prominent disputes involving NATO correspond to the NATO bombing of Yugoslavia in 1999 and the Libya intervention in 2011. 
Both Yugoslavia (YUG) and Libya (LIB) are in block 2, so the high allied continuation $\alpha_{1,2}^{xy \rightarrow ay}$ is  correctly modeling these incidents.
Note also that block pair $(2,1)$ has the highest excitation for turn continuation. 
This is dominated again by the dispute between NATO and Yugoslavia, with Yugoslavia threatening multiple NATO states in rapid succession.

We present this case study to illustrate the type of analysis that MULCH can be used for. 
Our analysis is exploratory rather than confirmatory, and we caution against jumping to conclusions about the behaviors of states from our results.

\section{Conclusion}
We proposed the multivariate community Hawkes (MULCH) model for continuous-time dynamic networks and demonstrated that it is superior to existing models both in terms of predictive and generative abilities on several real network datasets. 
The main innovation in our model is introducing dependence between node pairs in a tractable manner by using multivariate Hawkes processes with a structured excitation matrix $\alphaMat$ inspired by the SBM. 
In addition to self and reciprocal excitation, 
we also incorporated excitations motivated by sociological concepts of turn continuing and generalized reciprocity, which can replicate higher-order temporal motifs. 
We emphasize that these are not the only types of excitations that can be incorporated into our modeling framework---an investigation of other potential excitations would be a useful avenue for future research.

\section*{Acknowledgements}
This material is based upon work supported by the National Science Foundation grants IIS-1755824, DMS-1830412, IIS-2047955, and DMS-1830547. 

% In the unusual situation where you want a paper to appear in the
% references without citing it in the main text, use \nocite

\bibliography{references}
\bibliographystyle{icml2022}

%%%%%%%%%%%%%%%%%%%%%%%%%%%%%%%%%%%%%%%%%%%%%%%%%%%%%%%%%%%%%%%%%%%%%%%%%%%%%%%
%%%%%%%%%%%%%%%%%%%%%%%%%%%%%%%%%%%%%%%%%%%%%%%%%%%%%%%%%%%%%%%%%%%%%%%%%%%%%%%
% APPENDIX
%%%%%%%%%%%%%%%%%%%%%%%%%%%%%%%%%%%%%%%%%%%%%%%%%%%%%%%%%%%%%%%%%%%%%%%%%%%%%%%
%%%%%%%%%%%%%%%%%%%%%%%%%%%%%%%%%%%%%%%%%%%%%%%%%%%%%%%%%%%%%%%%%%%%%%%%%%%%%%%
\newpage
\appendix
\onecolumn

\section{Additional Details on MULCH Model and Estimation Procedure}

\subsection{Sum of Kernels Hawkes Process Intensity}
\label{sec:supp_sum_of_kernels}
By allowing different block pairs to have different scaling parameters $C_{ab}^q$, their intensities can decay at different rates, as in the CHIP \citep{arastuie2020chip} and BHM \citep{junuthula2019block} models where a single $\beta_{ab}$ value was estimated for each block pair. 
As shown in Figure \ref{fig:Sum_kernels}, $\bm{C}_{ab}^q$ controls decay rate, and also, allows every block pair to have distinct kernel shape. The intensity of $(i,j)\in \bp(a,b)$ becomes
\begin{equation}
    \lambda_{ij}(t) = \mu_{ab} + \sum_{\substack{(x,y) \in \bp(a,b) \\ (x,y) \in \bp(b,a)}} \alpha_{ab}^{xy \rightarrow ij} \sum_{t_s \in T_{xy}} \sum_{q=1}^{Q} C_{ab}^{q} \beta_{q} e^{-\beta_{q}\left(t-t_{s}\right)}.
    \label{MBHM_sum_kernels}
\end{equation}

\begin{figure}[t]
\centering
\includegraphics[width=3.5in]{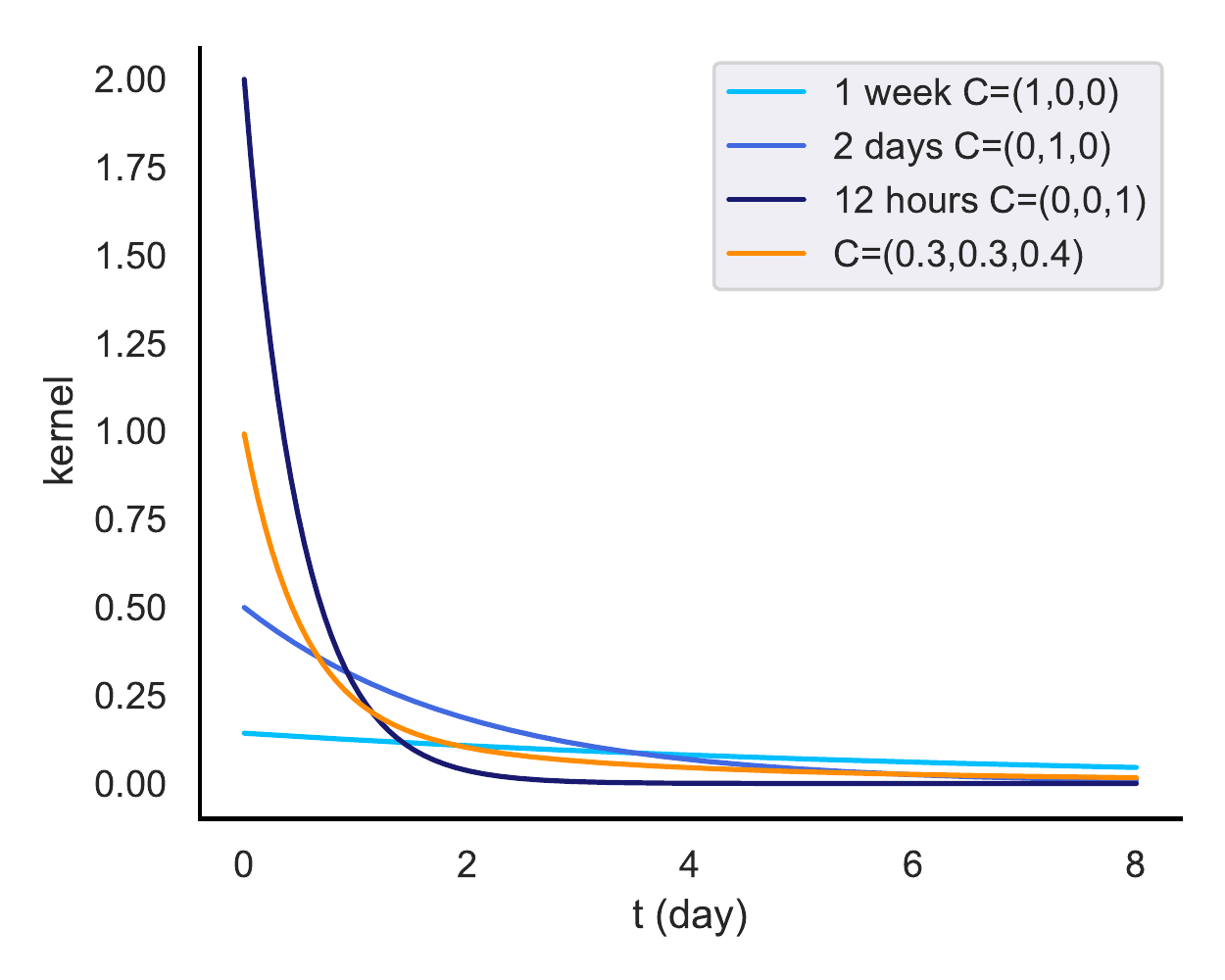}
\caption{Sum of exponential kernels shapes at $\bm{\beta = (}$1 week, 2 days, 12 hours) and different $\bm{C}_{ab}^q$ values. }
\label{fig:Sum_kernels}
\end{figure}

\subsection{Full Log-likelihood Function}
\label{sec:full_log_lik}
Based on our intensity function (\ref{MBHM_sum_kernels}) and the log-likelihood function for exponential Hawkes processes, we can derive our block pair wise log-likelihood function to be
\begin{align*}
    \ell_{ab}(\vec\theta_{ab}|\vec Z,\mathcal{H}_t) &= \sum_{Z_i=a, Z_j=b, i\neq j} \left\{ -\mu_{ab}T 
    - \sum_{\substack{(x,y) \in \bp(a,b) \\ (x,y) \in \bp(b,a)}}\alpha_{ab}^{xy\rightarrow ij}\sum_{t_s \in T_{xy}}\sum_{q=1}^QC_{ab}^q\left[\left(1-e^{-\beta_q(T-t_s)}\right)\right] \right.\\
    &\left. +\sum_{t_s\in T_{ij}}\ln \left[\mu_{ab} + \sum_{\substack{(x,y) \in \bp(a,b) \\ (x,y) \in \bp(b,a)}}\alpha_{ab}^{xy\rightarrow ij}\sum_{q=1}^QC_{ab}^q\beta_qR^q_{xy\rightarrow ij}(t_s)\right]\right\}
\end{align*}
where 
\begin{equation*}
R^q_{xy\rightarrow ij}(t_s) = \sum_{\substack{t_r \in T_{xy}\\t_r <t_s}}e^{-\beta_q(t_s-t_r)}
\end{equation*}
can be computed recursively \cite{arastuie2020chip}.

\subsection{Block Structure in Expected Count Matrix}
\label{sec:supp_block_structure}
As shown in Figure \ref{fig:Block_structure_count}, the expected count matrix $E[\bm{N}(T)]$ follows a block structure. 
This block structure is stated explicitly in Theorem \ref{thm:block_structure}, and the proof is below.
\begin{proof}[Proof of Theorem \ref{thm:block_structure}]
When the process is stationary, we can derive the vectorized version of the expected intensity as $vec(\boldsymbol{\lambda})=(\boldsymbol{I}-\boldsymbol{\Gamma})^{-1} \boldsymbol{\mu}$, where all node pairs in a block pair $(a,b)$ are kept together and they are ordered such that block pairs $(a,b)$, $(b,a)$ occupy consecutive positions. Then $\boldsymbol{\Gamma}$ is a $n(n-1) \times n(n-1)$ matrix whose elements are
$\Gamma_{(i,j)\rightarrow(x,y)} = {\alpha_{a b}^{ij \rightarrow xy}}$ \cite{hawkes1971spectra} and it has a diagonal block structure, i.e.,
$$\boldsymbol{\Gamma} = \left(\begin{array}{cccc}
\boldsymbol\Gamma_{(a,b)}  & & \\
& \boldsymbol\Gamma_{(a',b')}  & \\
& &  \dots& 
\end{array}\right)$$
where $\boldsymbol\Gamma_{(a,b)}$ contains the rows  $(i,j)\rightarrow (x,y)$ such that $(i,j) \in (a,b)$ or $(b,a)$ and $(x,y) \in (a,b)$ or $(b,a)$. This is because $\Gamma_{(i,j)\rightarrow(x,y)} = 0$ for any $(i,j) \in (a,b)$ and $(x,y) \notin \{(a,b),(b,a)\}$.
Denote $\boldsymbol G_{(a,b)}=\boldsymbol I- \boldsymbol \Gamma_{(a,b)}$, where $\boldsymbol I$ is the identity matrix. Then, 
$$\boldsymbol{I-\Gamma} = \left(\begin{array}{cccc}
\boldsymbol G_{(a,b)}  & & \\
& \boldsymbol G_{(a',b')}  & \\
& &  \dots& 
\end{array}\right)$$
also have the diagonal block structure, and so does
$$({ \boldsymbol I-\boldsymbol \Gamma)^{-1}} = \left(\begin{array}{cccc}
\boldsymbol G_{(a,b)}^{-1}  & & \\
& \boldsymbol G_{(a',b')}^{-1}  & \\
& &  \dots& 
\end{array}\right)$$

By our construction of the $\boldsymbol \Gamma$, we can write the $\boldsymbol G_{(a,b)}$ as a block matrix:
$$\boldsymbol G_{(a,b)} = \left(\begin{array}{cc}
\boldsymbol G_{ab\rightarrow ab}  & \boldsymbol G_{ab\rightarrow ba} \\
\boldsymbol G_{ba\rightarrow ab}  & \boldsymbol G_{ba\rightarrow ba}\\
\end{array}\right)$$
and in the same block matrix, the rows have the same row sum.

Note for a block matrix, the inversion is 
$$\left(\begin{array}{ll}
\mathbf{A} & \mathbf{B} \\
\mathbf{C} & \mathbf{D}
\end{array}\right)^{-1}= \left(\begin{array}{cc}
\left(\mathbf{A}-\mathbf{B D}^{-1} \mathbf{C}\right)^{-1} & -\left(\mathbf{A}-\mathbf{B D}^{-1} \mathbf{C}\right)^{-1} \mathbf{B D}^{-1} \\
-\mathbf{D}^{-1} \mathbf{C}\left(\mathbf{A}-\mathbf{B D}^{-1} \mathbf{C}\right)^{-1} & \mathbf{D}^{-1}+\mathbf{D}^{-1} \mathbf{C}\left(\mathbf{A}-\mathbf{B D}^{-1} \mathbf{C}\right)^{-1} \mathbf{B D}^{-1}
\end{array}\right)$$

We summarize a few simple observations in the following proposition.
\begin{proposition} For any square matrix $\boldsymbol A$ the following holds.
\begin{enumerate}
    \item If $\boldsymbol A\boldsymbol{1} = a \boldsymbol 1$, i.e., if the row sum of $\boldsymbol{A}$ are identical, and if $\boldsymbol A^{-1}$ exists, then $\boldsymbol A^{-1}\boldsymbol{1} = a^{-1} \boldsymbol{1}$, i.e., the row sum of $\boldsymbol A^{-1}$ are also identical.
    \item If $\boldsymbol A \boldsymbol 1 = a \boldsymbol 1$ and $\boldsymbol B \boldsymbol 1 = b \boldsymbol 1$, then $\boldsymbol A \boldsymbol B \boldsymbol 1 = ab \boldsymbol 1$
    \item If $\boldsymbol A \boldsymbol 1 = a \boldsymbol 1$ and $\boldsymbol B \boldsymbol 1 = b\boldsymbol 1$, then $(\boldsymbol A - \boldsymbol B)\boldsymbol 1  = (a-b)\boldsymbol 1$
\end{enumerate}
\end{proposition}

Thus, since $\boldsymbol G_{ab\rightarrow ab},\boldsymbol G_{ab\rightarrow ba}, \boldsymbol G_{ba\rightarrow ab}, \boldsymbol G_{ba\rightarrow ba}$ have the same row sum, and if all of them and $\boldsymbol G_{(a,b)}$ are invertible, then using the proposition above, we have 
\begin{align*}\boldsymbol G_{(a,b)}^{-1} \left(\begin{array}{c}
\boldsymbol 1_{n_an_b} \\
\boldsymbol 0_{n_an_b} \\
\end{array}\right) & = \left(\begin{array}{c}
\left(\mathbf{G}_{ab\rightarrow ab}-\mathbf{G}_{ab\rightarrow ba} {\mathbf{G}_{ba\rightarrow ba}}^{-1} \mathbf{G}_{ba\rightarrow ab}\right)^{-1}\boldsymbol 1_{n_an_b} \\
-\mathbf{G}_{ba\rightarrow ba}^{-1} \mathbf{G}_{ba\rightarrow ab}\left(\mathbf{G}_{ab\rightarrow ab}-\mathbf{G}_{ab\rightarrow ba} {\mathbf{G}_{ba\rightarrow ba}}^{-1} \mathbf{G}_{ba\rightarrow ab}\right)^{-1} \boldsymbol 1_{n_an_b} \\ 
\end{array}\right) \\
& = \left(\begin{array}{c}
g^1_{ab}\boldsymbol 1_{n_an_b} \\
g^2_{ab}\boldsymbol 1_{n_an_b} \\
\end{array}\right)
\end{align*}

where $g^1_{ab}, g^2_{ab}$ are some real valued functions of the parameters $\boldsymbol\alpha_{ab}, \boldsymbol\alpha_{ba}$, and $n_a$ denotes the number of nodes in block $a$.

Similarly, 
$$\boldsymbol G_{(a,b)}^{-1} \left(\begin{array}{c}
\boldsymbol 0_{n_an_b} \\
\boldsymbol 1_{n_an_b} \\
\end{array}\right) = \left(\begin{array}{c}
g^1_{ba}\boldsymbol 1_{n_an_b} \\
g^2_{ba}\boldsymbol 1_{n_an_b} \\
\end{array}\right)$$

Thus, by our construction, 
\begin{align}
\text{vec}(\boldsymbol \lambda_{{(i,j)|i\in a, j \in b \text{ or } i\in b, j \in a}})&= \boldsymbol G_{(a,b)}^{-1}\left(\begin{array}{c}
\mu_{ab} \\
\vdots \\
\mu_{ab} \\
\mu_{ba} \\
\vdots \\
\mu_{ba}
\end{array}\right)\\
&=\boldsymbol G_{(a,b)}^{-1}\left(\left(\begin{array}{c}
\boldsymbol 1_{n_an_b} \\
\boldsymbol 0_{n_an_b} \\
\end{array}\right)\mu_{ab}+\left(\begin{array}{c}
\boldsymbol 0_{n_an_b} \\
\boldsymbol 1_{n_an_b} \\
\end{array}\right)\mu_{ba}\right)\\
&= \left(\begin{array}{c}
(g^1_{ab} \mu_{ab} + g^1_{ba} \mu_{ba})\boldsymbol 1_{n_an_b} \\
(g^2_{ab} \mu_{ab} + g^2_{ba} \mu_{ba})\boldsymbol 1_{n_an_b} \\
\end{array}\right)
\end{align}
which means for any $(i,j) \in (a,b)$, $\lambda_{(i,j)} =(g^1_{ab} \mu_{ab} + g^1_{ba} \mu_{ba})$, and for any $(i',j') \in (b,a)$, $\lambda_{(i',j')} = (g^2_{ab} \mu_{ab} + g^2_{ba} \mu_{ba})\boldsymbol 1_{n_an_b}$.
\end{proof}

\section{Additional Experiment Details and Results}
\label{sec:supp_exp}

\subsection{Spectral Clustering and Refinement Accuracy}
\label{sec:supp_exp_rand}
For these experiments, we generate data from the MULCH model with $K=4$ and assume parameters of the four diagonal block pairs are equal, and similarly, parameters of the off-diagonal block pairs are equal.

We denote the set of Hawkes process parameters for a block pair $(a,b)$ by
\begin{equation*}
\bm{\theta}_{ab} = 
\left(\mu_{ab}, \alpha_{ab}^{xy \rightarrow xy}, \alpha_{ab}^{xy \rightarrow yx}, \alpha_{ab}^{xy \rightarrow xb}, \alpha_{ab}^{xy \rightarrow ya}, \alpha_{ab}^{xy \rightarrow ay}, \alpha_{ab}^{xy \rightarrow bx} \right).
\end{equation*}

In these experiments, we use the following parameters:
\begin{align*}
\bm{\theta}_{aa}&=\bm{\theta}_{bb}=(0.008, 0.3, 0.3, 0.002, 0.0005, 0.001, 0.0005) \\
\bm{\theta}_{ab}&=\bm{\theta}_{ba}=(0.008, 0.1, 0.1, 0.001, 0.0001, 0.001, 0.0001) \\
\bm{C}_{aa}&=\bm{C}_{ab}=\bm{C}_{ba}=\bm{C}_{bb}=(0.33, 0.33, 0.34) \\
\bm{\beta}&=(2 \text{ weeks}, 1 \text{ day}, 2 \text{ hours}) = (1/14 , 1, 24/2)
\end{align*}

Note that, simulated networks have assortative mixing, and we assumed timestamps are in unit of days.

\subsection{Parameter Estimation Accuracy}
\label{sec:supp_exp_param}

We test the accuracy of our MLE in this experiment. 
We generate data from the MULCH model with $K=2$ and assume same structured parameters as in \ref{sec:supp_exp_rand}. However, we switch diagonal and off-diagonal block pairs parameters so simulated networks are disassortative:
\begin{align*}
\bm{\theta}_{aa}&=\bm{\theta}_{bb}=(0.008, 0.1, 0.1, 0.001, 0.0001, 0.001, 0.0001) \\
\bm{\theta}_{ab}&=\bm{\theta}_{ba}=(0.008, 0.3, 0.3, 0.002, 0.0005, 0.001, 0.0005)
\end{align*}

\begin{figure*}[tp]
    \newcommand{\figwidth}{0.245\textwidth}
    \centering
    \hfill
    \begin{subfigure}[c]{\figwidth}
        \centering
        \includegraphics[width=\textwidth]{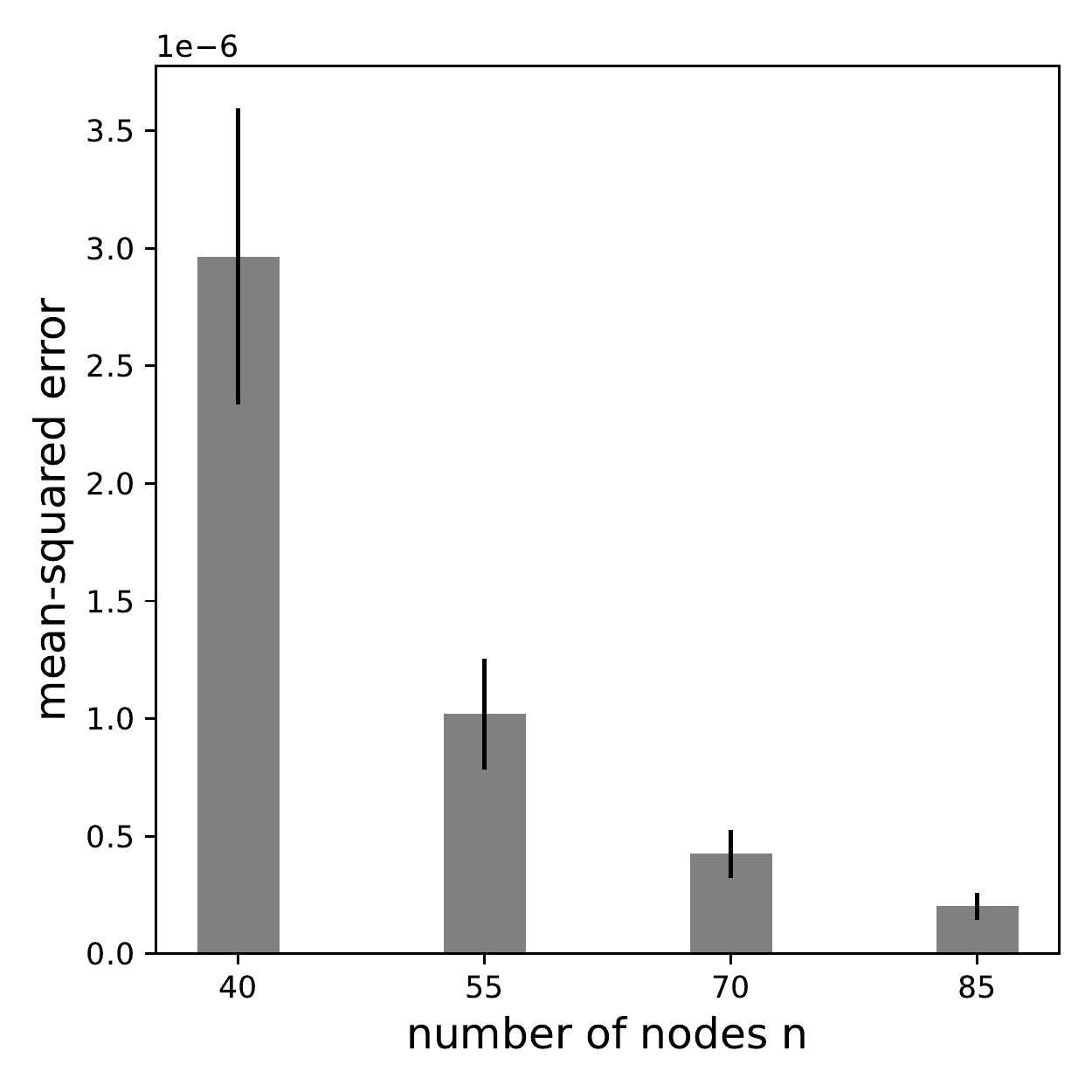}
        \caption{Base intensity $\mu_{ab}$}
    \end{subfigure}
    \hfill
    \centering
    \begin{subfigure}[c]{\figwidth}
        \centering
        \includegraphics[width=\textwidth]{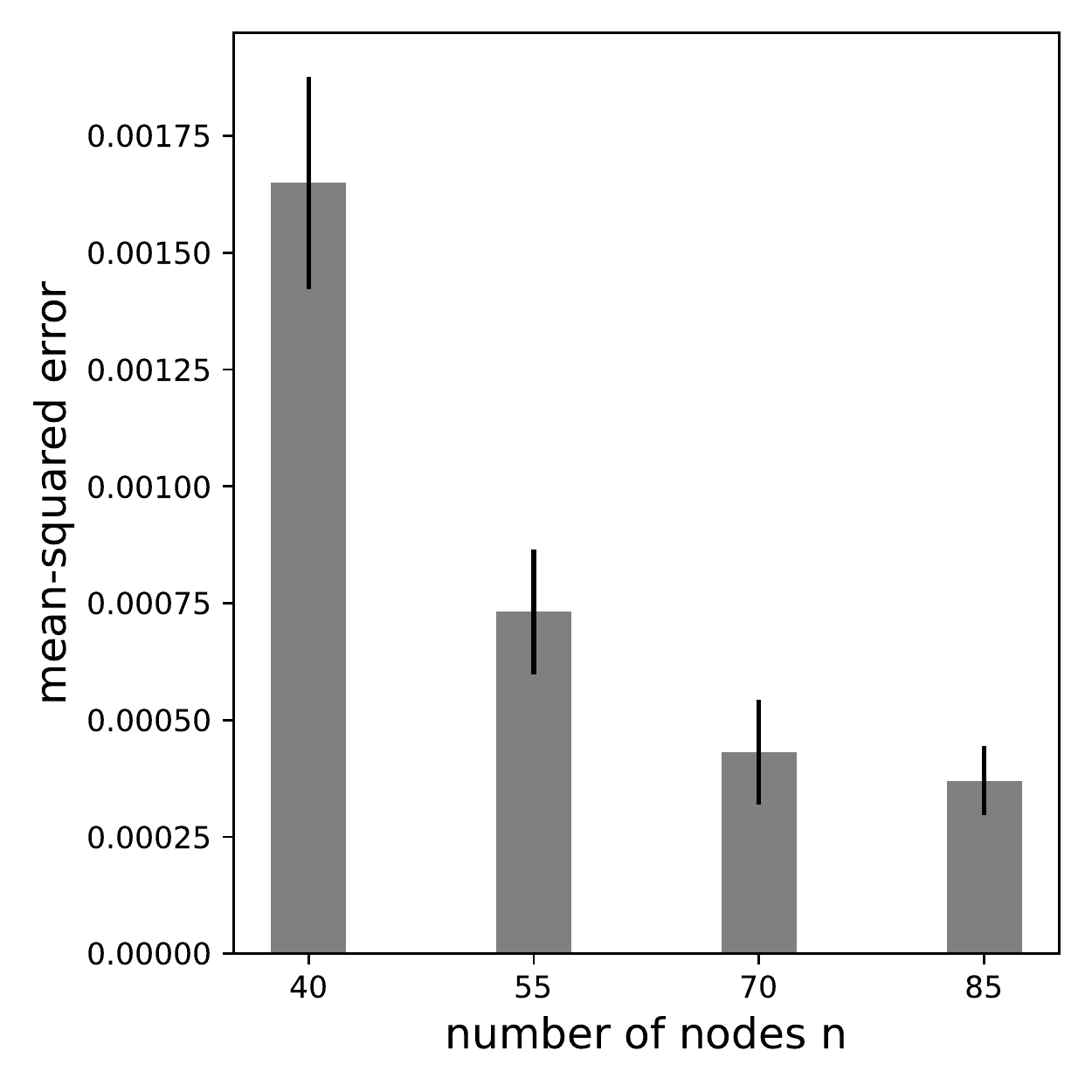}
        \caption{Self excitation $\alpha_{ab}^{xy \rightarrow xy}$}
    \end{subfigure}
    \hfill
    \centering
    \begin{subfigure}[c]{\figwidth}
        \centering
        \includegraphics[width=\textwidth]{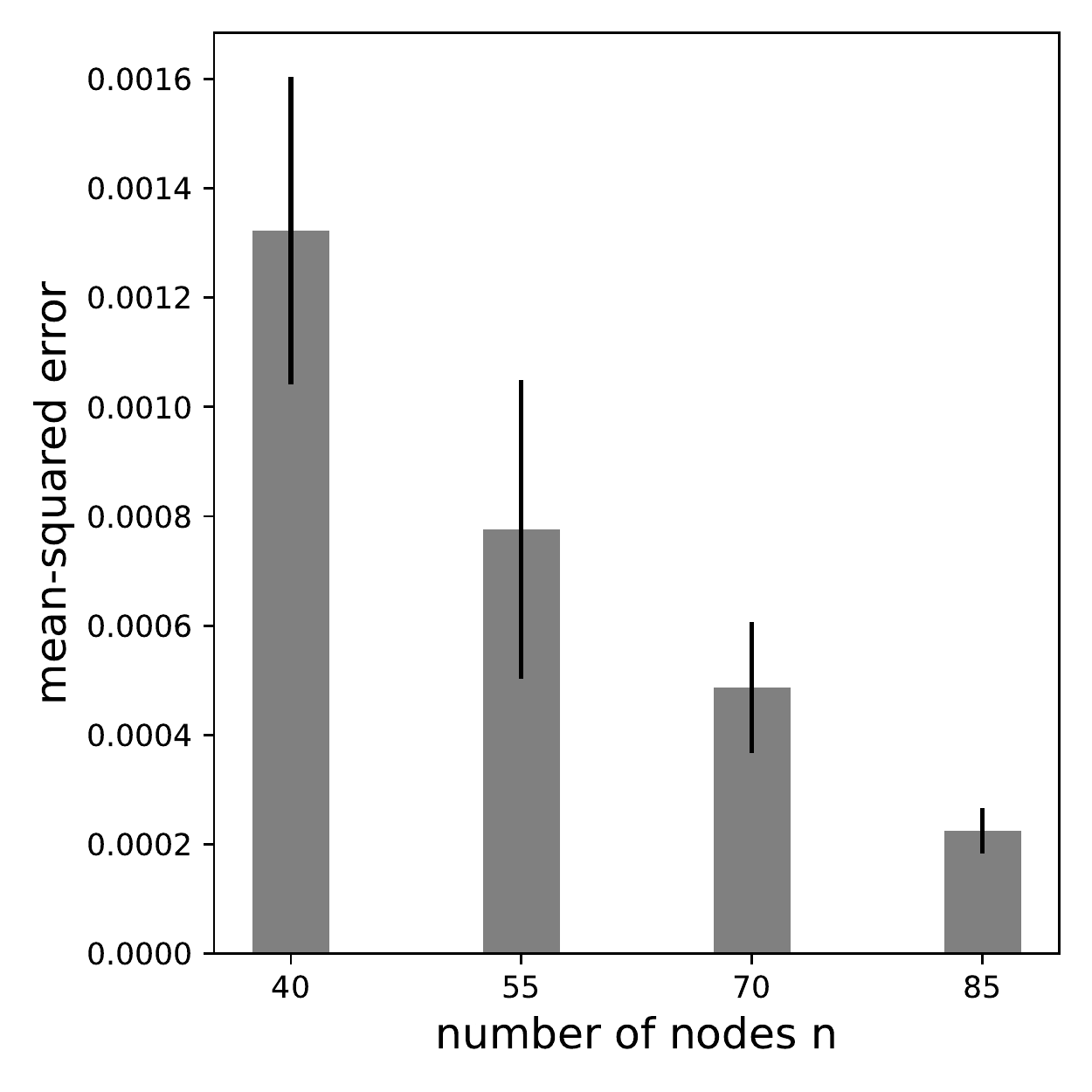}
        \caption{Reciprocal excitation $\alpha_{ab}^{xy \rightarrow yx}$}
    \end{subfigure}
    \hfill
    \centering
    \begin{subfigure}[c]{\figwidth}
        \centering
        \includegraphics[width=\textwidth]{figures/sim_tests/mse_T_150/turn_conti.pdf}
        \caption{Turn continuation $\alpha_{ab}^{xy \rightarrow xb}$}
    \end{subfigure}
    \hfill
    \\
    % \centering
    \begin{subfigure}[c]{\figwidth}
        \centering
        \includegraphics[width=\textwidth]{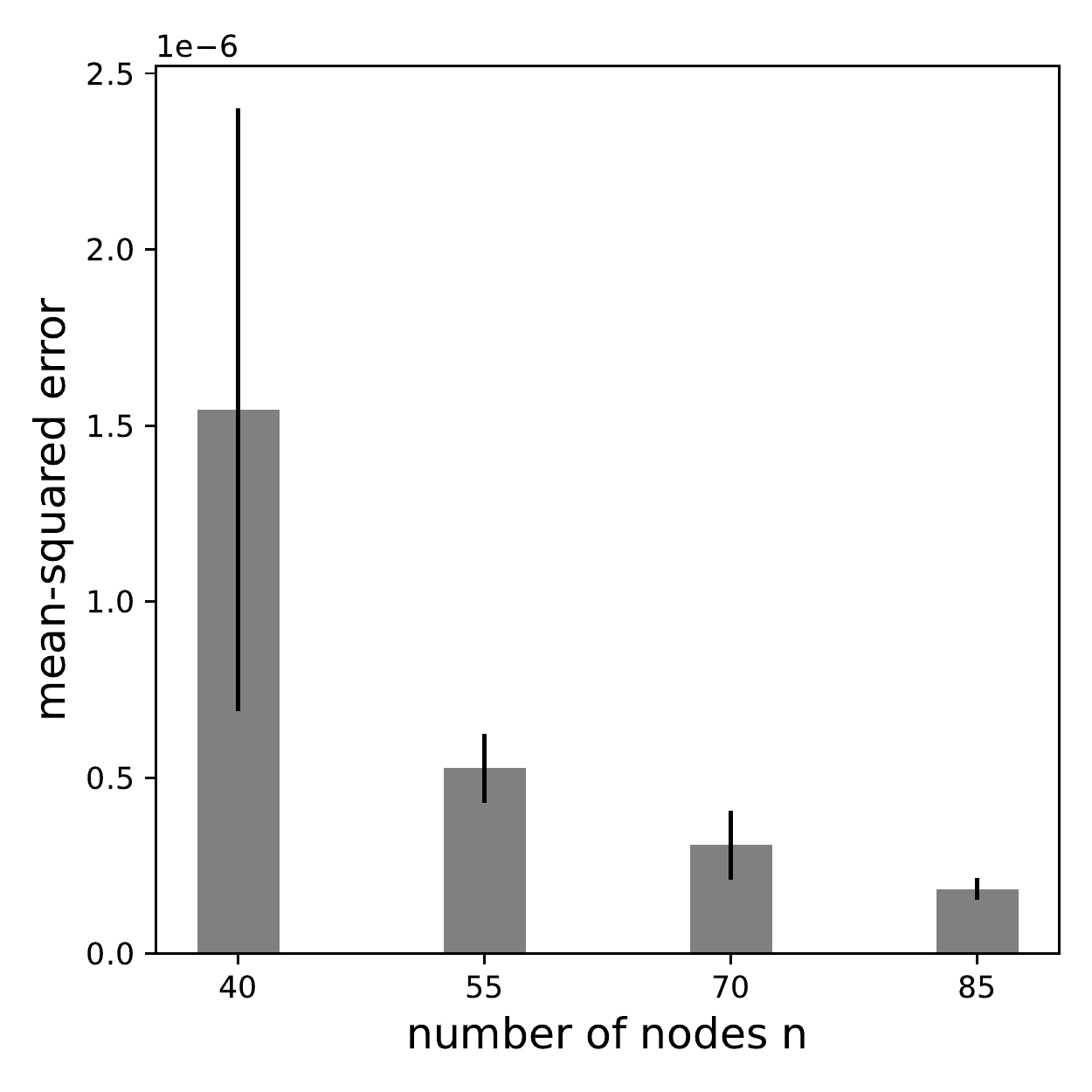}
        \caption{Generalized recip.~$\alpha_{ab}^{xy \rightarrow ya}$}
    \end{subfigure}
        \hfill
    \centering
    \begin{subfigure}[c]{\figwidth}
        \centering
        \includegraphics[width=\textwidth]{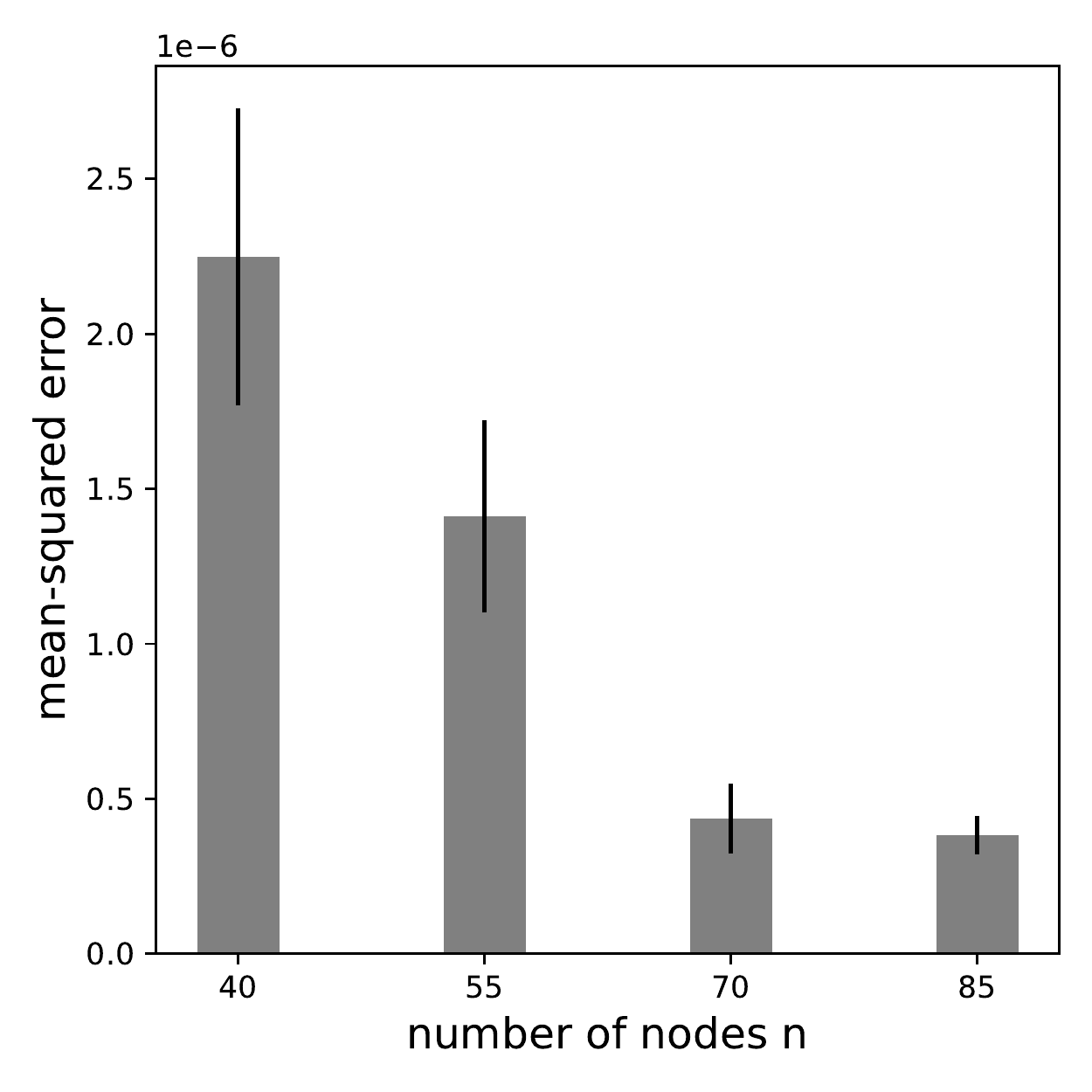}
        \caption{Allied continuation $\alpha_{ab}^{xy \rightarrow ay}$}
    \end{subfigure}
    \hfill
    \centering
    \begin{subfigure}[c]{\figwidth}
        \centering
        \includegraphics[width=\textwidth]{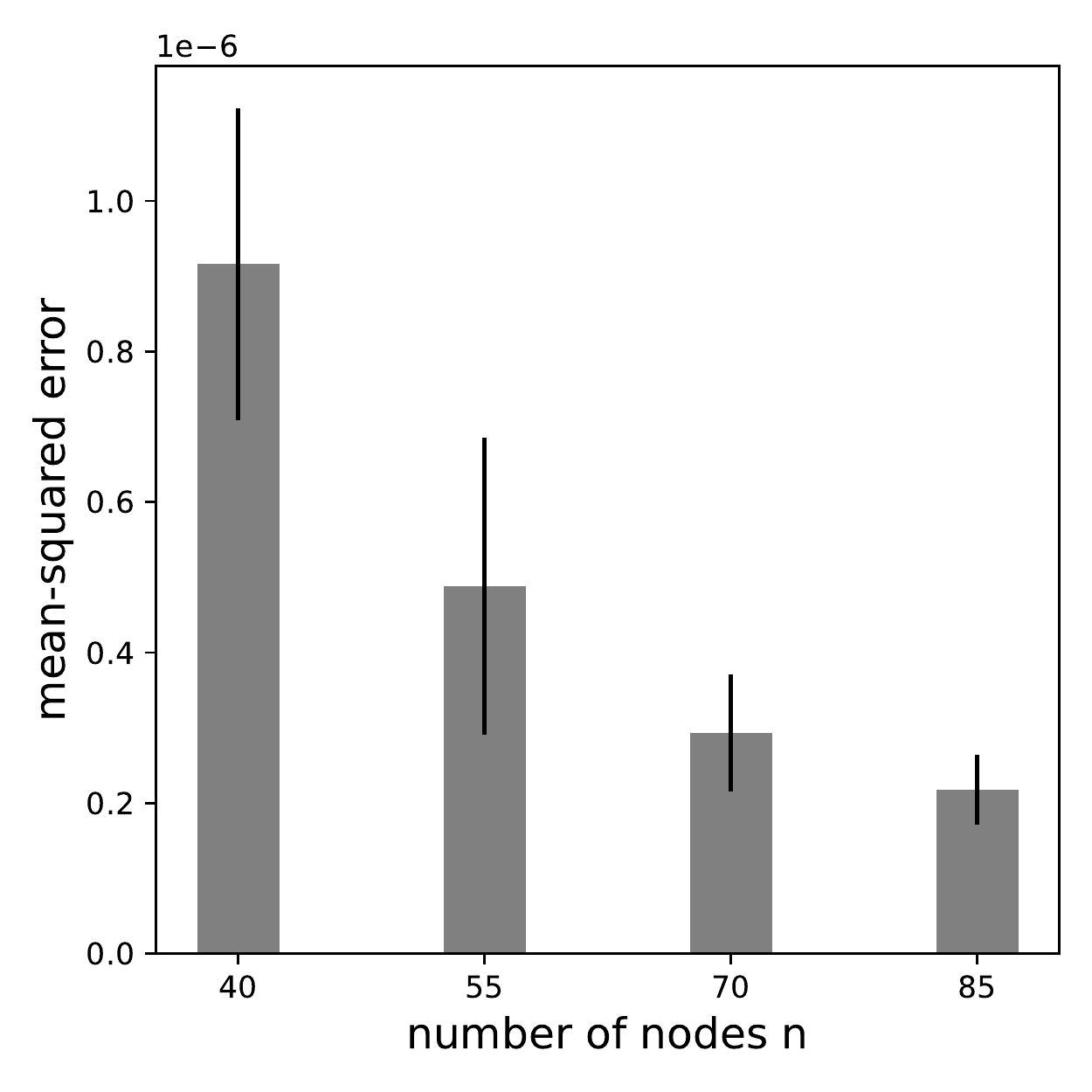}
        \caption{Allied reciprocity $\alpha_{ab}^{xy \rightarrow bx}$}
    \end{subfigure}
    \hfill
    \centering
    \begin{subfigure}[c]{\figwidth}
        \centering
        \includegraphics[width=\textwidth]{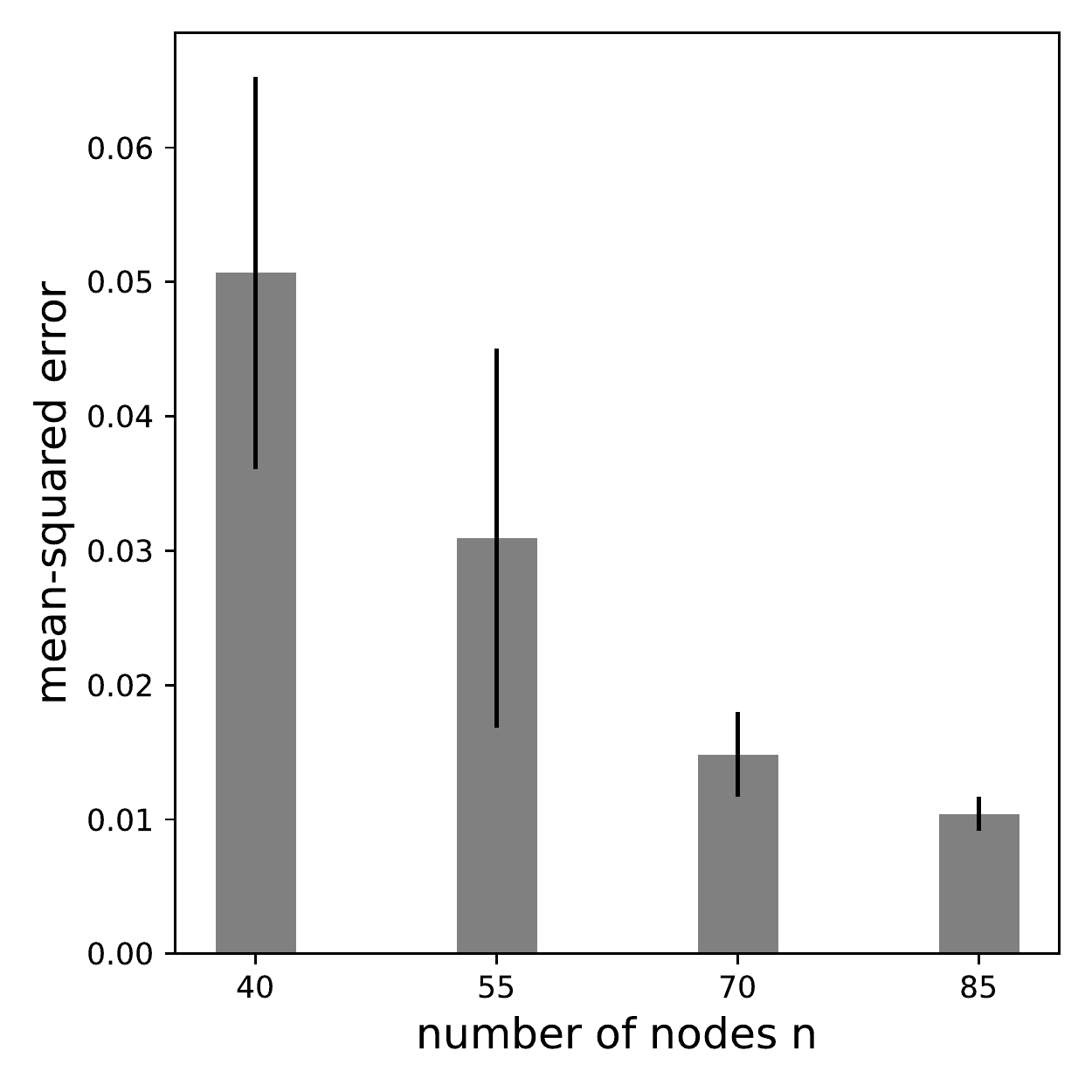}
        \caption{Kernel scaling $C_{ab}$}
    \end{subfigure}
    \hfill
    \caption{Mean squared error for Hawkes process parameters on simulated networks at $T=5$ months ($\pm$ standard error over 10 runs).}
    \label{fig:param_estimation_T}
\end{figure*}

\begin{figure*}[tp]
    \newcommand{\figwidth}{0.245\textwidth}
    \centering
    \hfill
    \begin{subfigure}[c]{\figwidth}
        \centering
        \includegraphics[width=\textwidth]{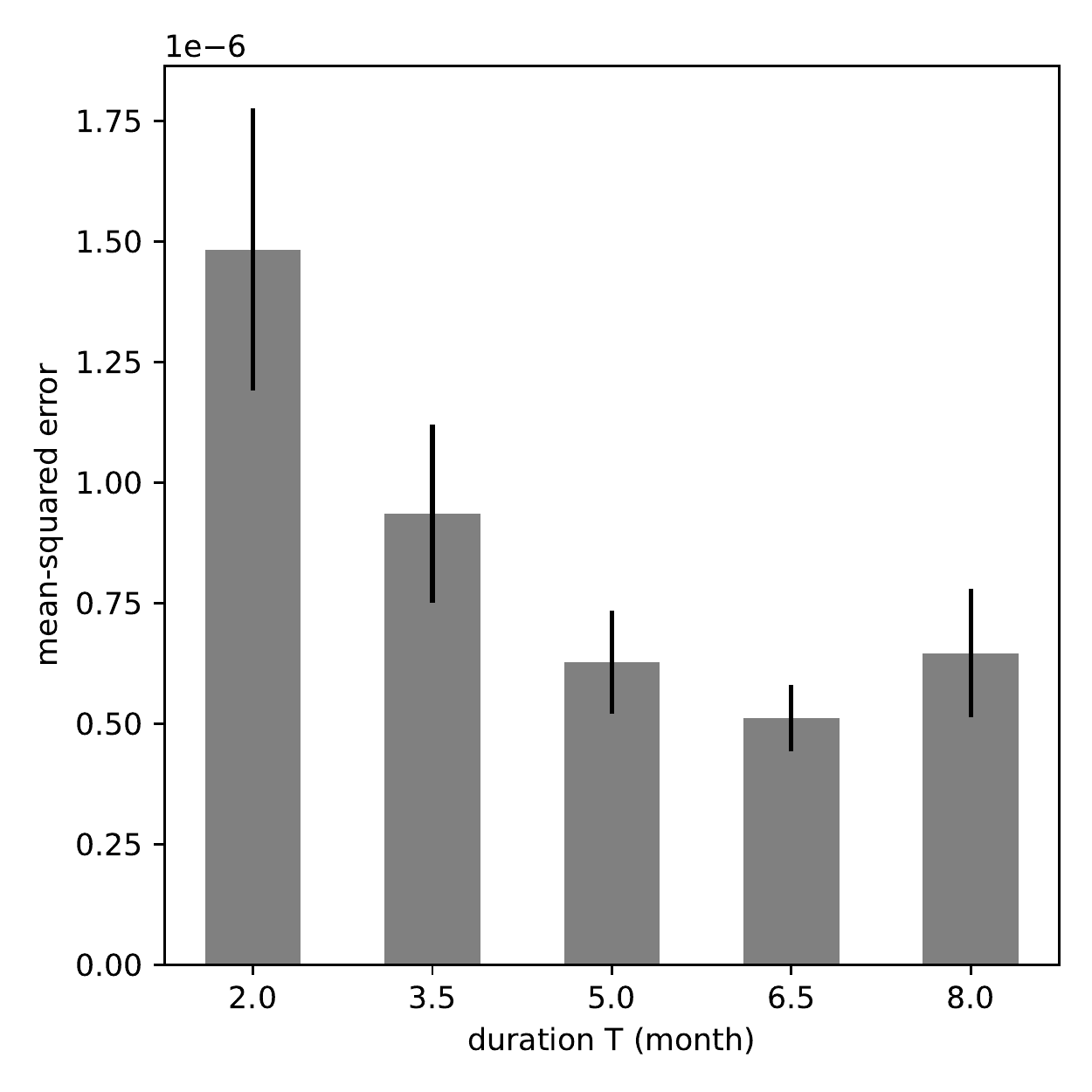}
        \caption{Base intensity $\mu_{ab}$}
    \end{subfigure}
    \hfill
    \centering
    \begin{subfigure}[c]{\figwidth}
        \centering
        \includegraphics[width=\textwidth]{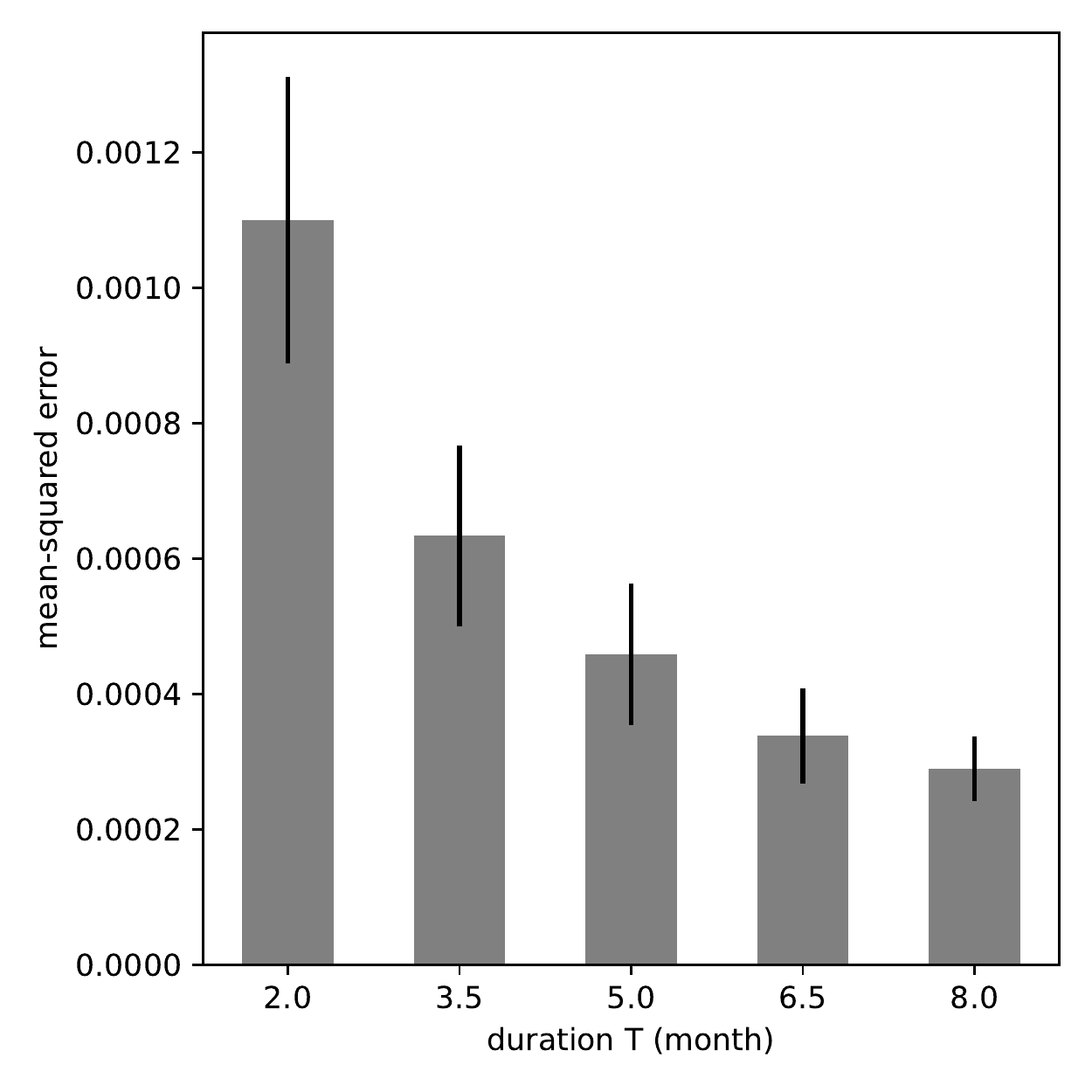}
        \caption{Self excitation $\alpha_{ab}^{xy \rightarrow xy}$}
    \end{subfigure}
    \hfill
    \centering
    \begin{subfigure}[c]{\figwidth}
        \centering
        \includegraphics[width=\textwidth]{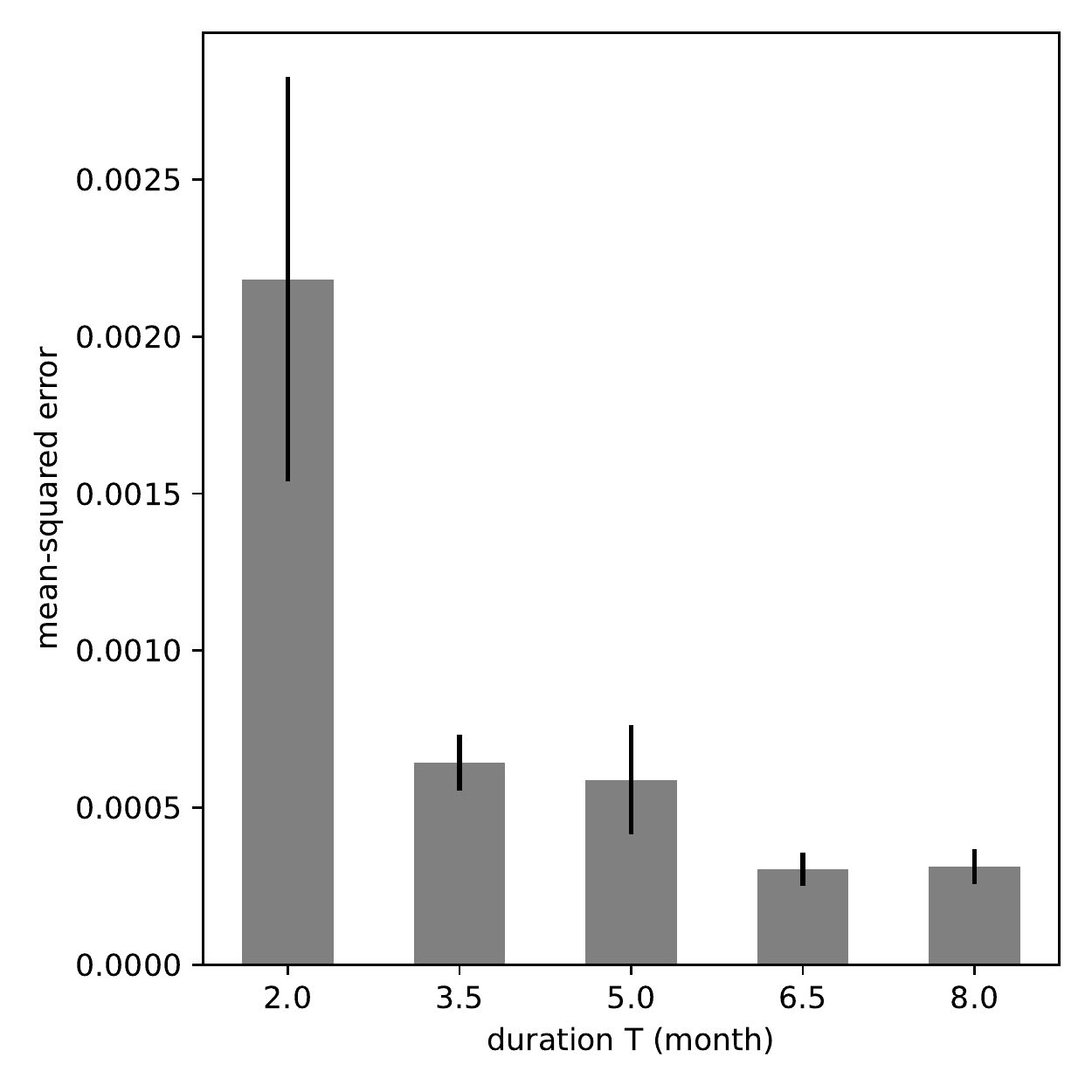}
        \caption{Reciprocal excitation $\alpha_{ab}^{xy \rightarrow yx}$}
    \end{subfigure}
    \hfill
    \centering
    \begin{subfigure}[c]{\figwidth}
        \centering
        \includegraphics[width=\textwidth]{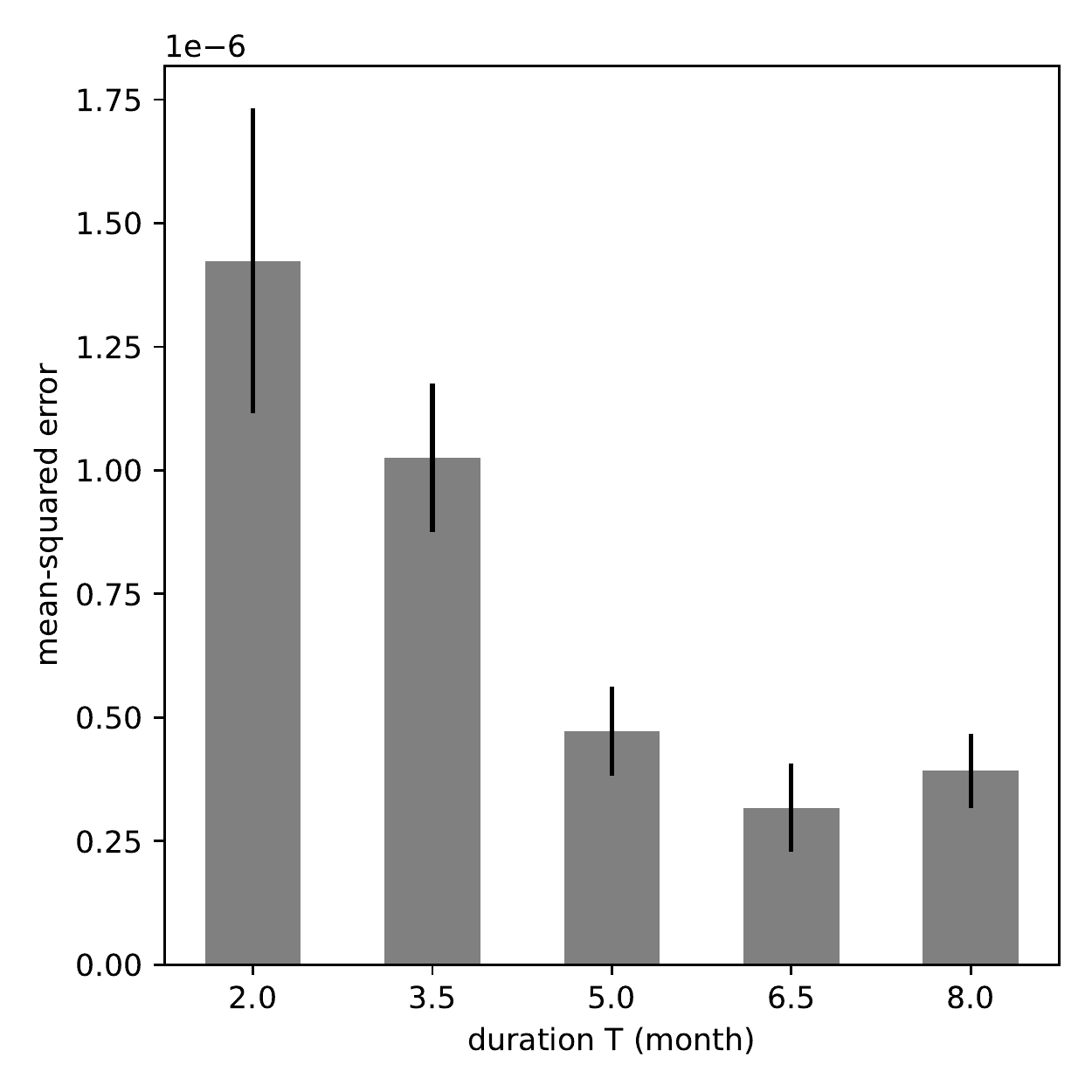}
        \caption{Turn continuation $\alpha_{ab}^{xy \rightarrow xb}$}
    \end{subfigure}
    \hfill
    \\
    % \centering
    \begin{subfigure}[c]{\figwidth}
        \centering
        \includegraphics[width=\textwidth]{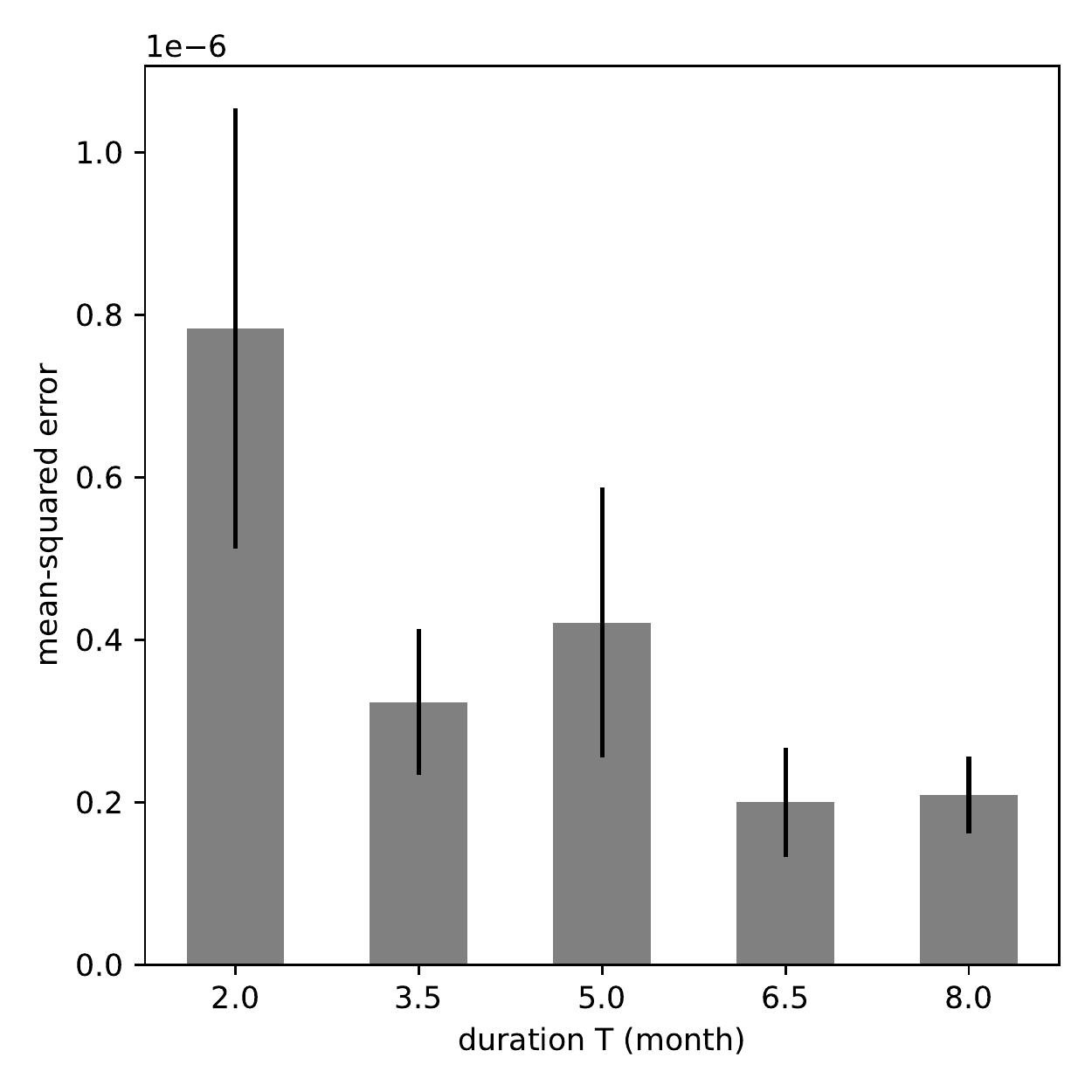}
        \caption{Generalized recip.~$\alpha_{ab}^{xy \rightarrow ya}$}
    \end{subfigure}
        \hfill
    \centering
    \begin{subfigure}[c]{\figwidth}
        \centering
        \includegraphics[width=\textwidth]{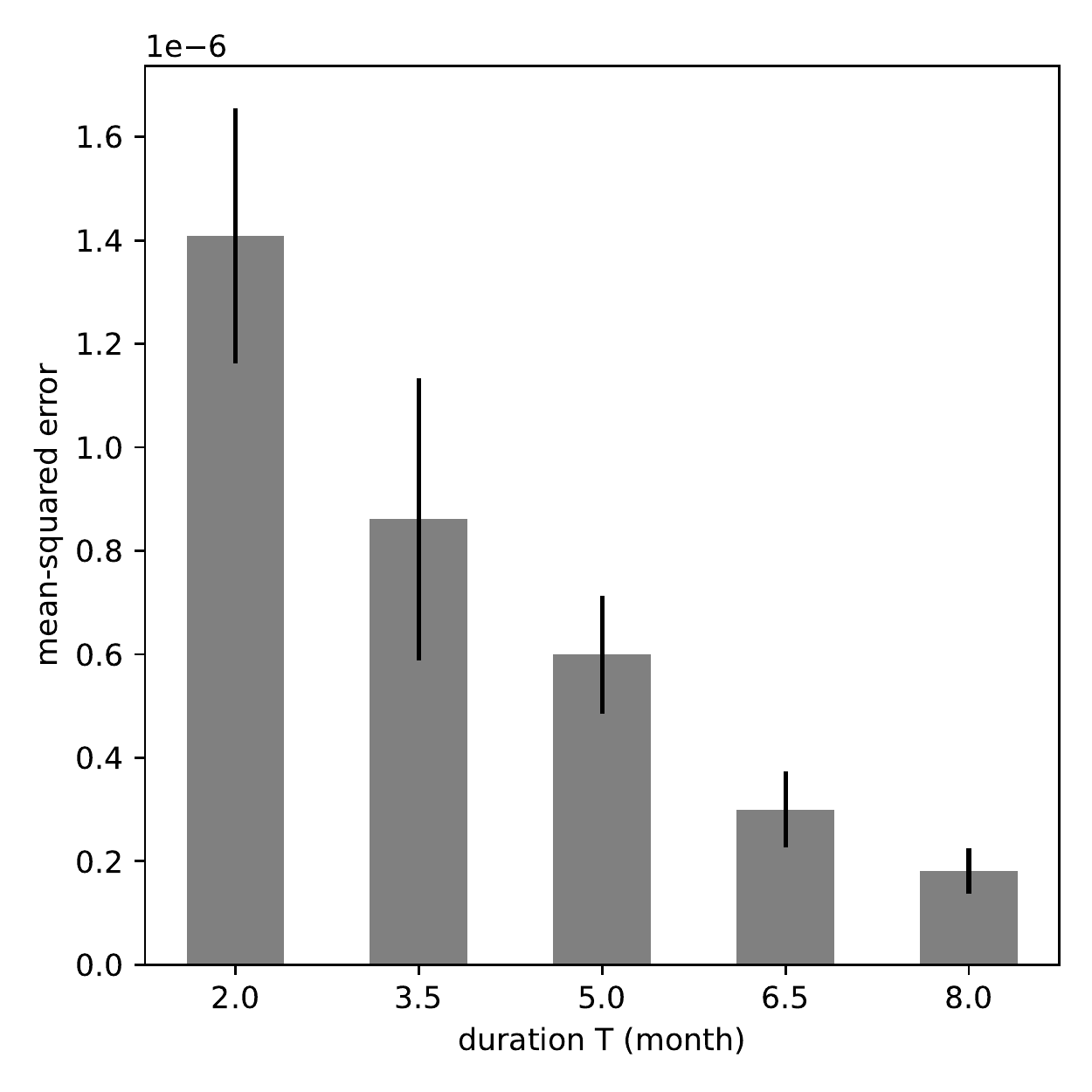}
        \caption{Allied continuation $\alpha_{ab}^{xy \rightarrow ay}$}
    \end{subfigure}
    \hfill
    \centering
    \begin{subfigure}[c]{\figwidth}
        \centering
        \includegraphics[width=\textwidth]{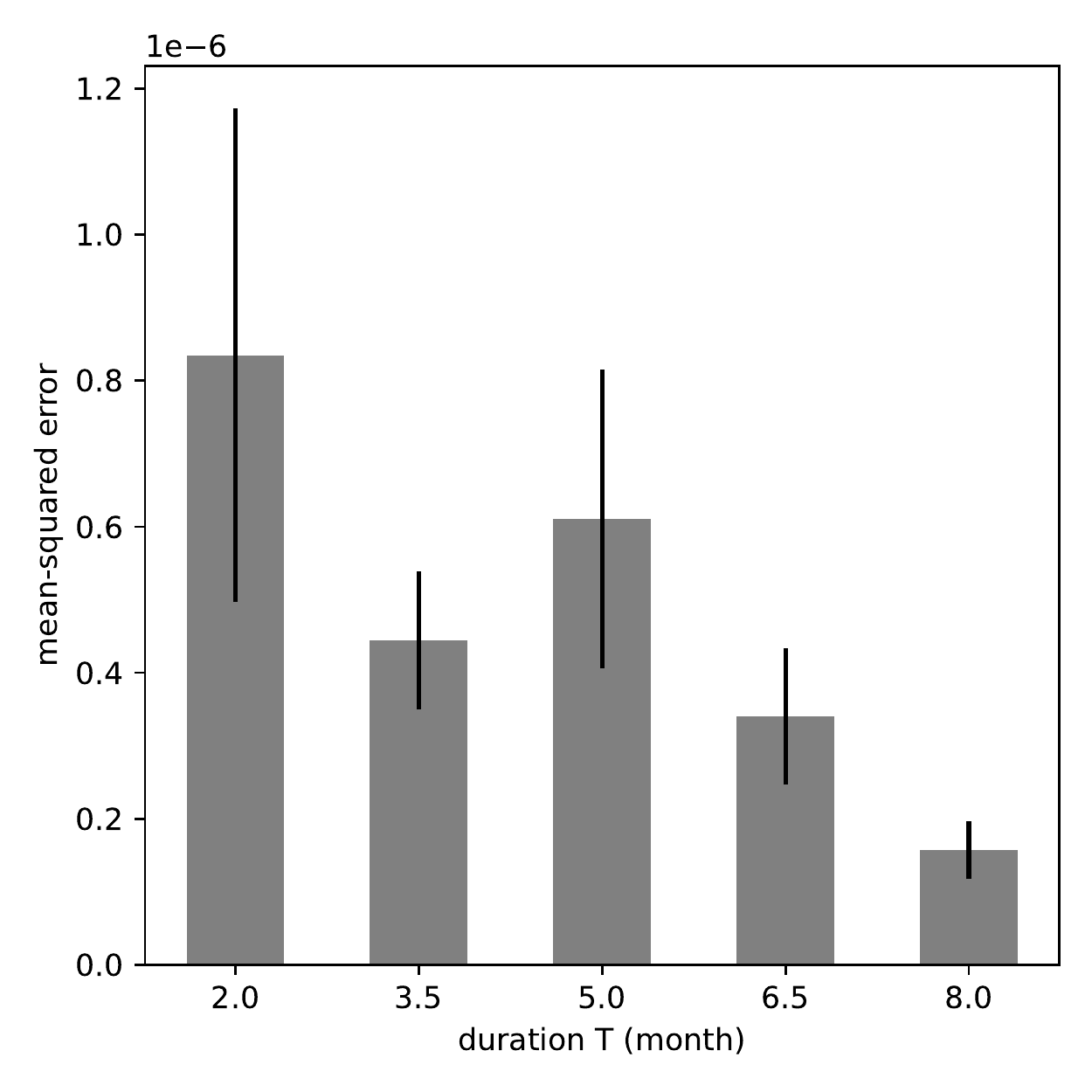}
        \caption{Allied reciprocity $\alpha_{ab}^{xy \rightarrow bx}$}
    \end{subfigure}
    \hfill
    \centering
    \begin{subfigure}[c]{\figwidth}
        \centering
        \includegraphics[width=\textwidth]{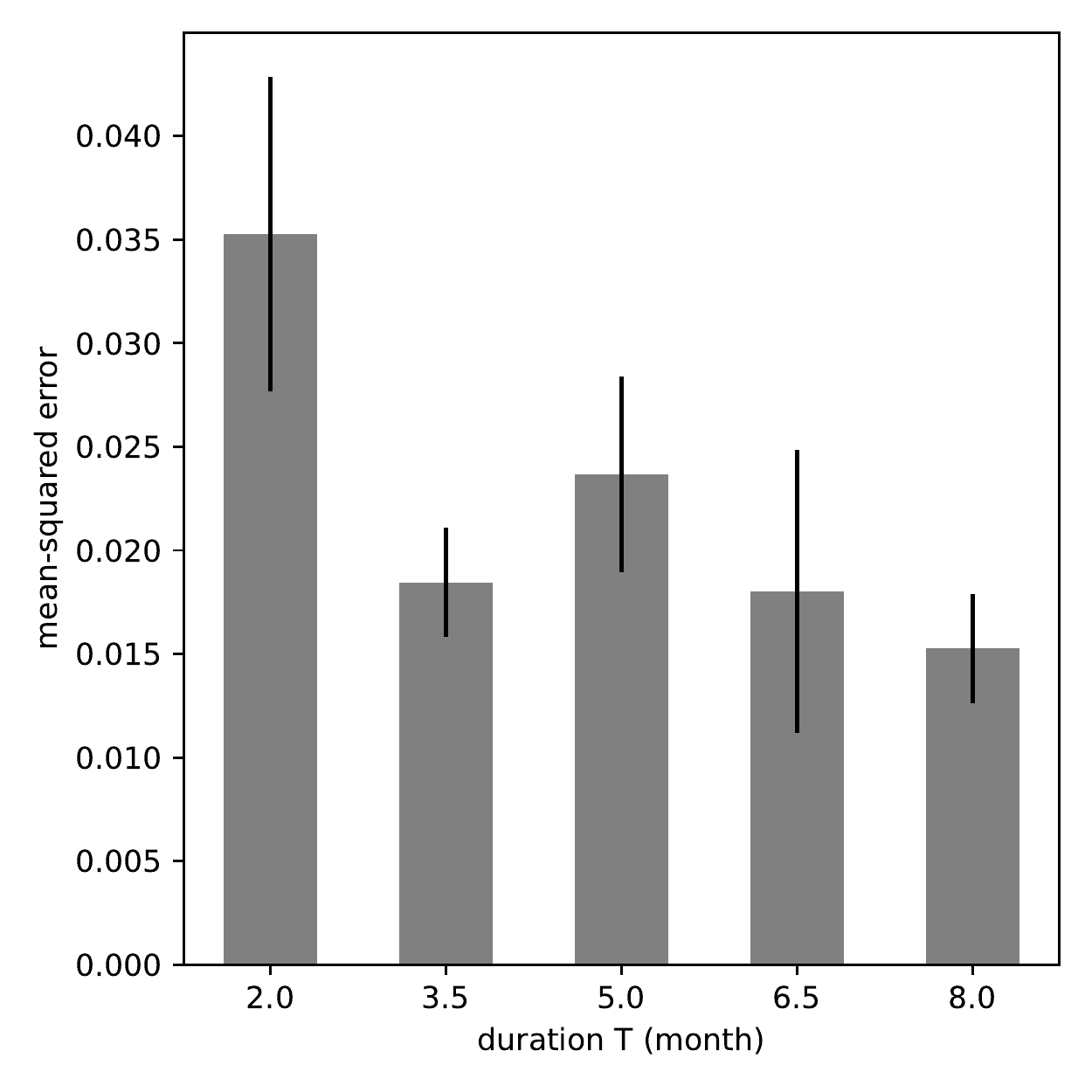}
        \caption{Kernel scaling $C_{ab}$}
    \end{subfigure}
    \hfill
    \caption{Mean squared error for Hawkes process parameters on simulated networks at $n=70$ ($\pm$ standard error over 10 runs).}
    \label{fig:param_estimation_n}
\end{figure*}

At a fixed duration $T=5$ months, and while varying number of nodes $n$, we simulate and fit our model to evaluate the accuracy of the different Hawkes process parameters. 
As shown in Figure \ref{fig:param_estimation_T}, the mean-squared error (MSE) for each estimated parameter decreases with increasing $n$ as expected.
Similarly, for a fixed number of nodes $n=70$, we vary the duration $T$.
Figure \ref{fig:param_estimation_n} shows that the MSE decreases with increasing $T$ as expected. 
For each pair values of $(n, T)$, we run 10 simulations.

\subsection{Real Dataset Descriptions}
\label{sec:supp_datasets}
Each dataset consists of a set of events where each event is denoted by a sender, a receiver, and a timestamp:
\begin{itemize}
    \item MIT Reality Mining \citep{eagle2009inferring}: We use the phone call data, where the start time of each call was used as the event timestamp.
    We consider calls between pairs of the core 70 callers and recipients. 
    We use $\beta$ values of 1 week, 2 days, and 1 hour.

    \item Enron Emails \citep{klimt2004enron}: We use the same subset of the Enron email corpus as in \citet{Dubois2013}.
    We use $\beta$ values of 1 week, 2 days, and 6 hours.

    \item Militarized Interstate Disputes (MID) \citep{palmer2021mid5}: We use the MID 5.01 dataset compiled by the Correlates of War project. 
    Nodes denote (sovereign) states. 
    Each edge denotes an incident, which is a threat, display, or use of force one state directs towards another. 
    We remove 8 nodes that are disconnected from the largest connected component. 
    Unlike the other networks, this is a conflict network, so we expect to see \emph{disassortative} rather than assortative mixing between nodes.
    We use $\beta$ values of 2 months, 2 weeks, and 12 hours.
    
    \item Facebook Wall Posts \citep{viswanath2009evolution}: 
    We consider only posts from a user to another user's wall so that there are no self-edges.
    We analyze the largest connected component of the network excluding self loops.
    We use $\beta$ values of 2 months, 1 week, and 2 hours.
\end{itemize}
The MIT Reality Mining, Enron, and Facebook datasets were loaded and preprocessed identically to \citet{arastuie2020chip}. 
Timestamps in all datasets were rescaled to be in the range $[0, 1000]$, the same as in \citet{arastuie2020chip} and \citet{Dubois2013} so that log-likelihoods are comparable with their reported figures.

\subsection{Descriptions of Other Models for Comparison}
\label{sec:supp_other_models}
We compare against several other TPP models for continuous-time networks:
\begin{itemize}
    \item Community Hawkes Independent Pairs (CHIP) \citep{arastuie2020chip}: Univariate Hawkes process network model with block structure where each node pair is independent of all others. We use the implementation at \url{https://github.com/IdeasLabUT/CHIP-Network-Model}
    
    \item Block Hawkes Model (BHM) \citep{junuthula2019block}: Univariate Hawkes process network model with block structure where an event between a node pair equally excites all node pairs in the same block pair. We use the implementation at \url{https://github.com/IdeasLabUT/CHIP-Network-Model}
    
    \item Relational Event Model (REM) \citep{Dubois2013}: Inhomogeneous Poisson process network model with piecewise constant intensities and block structure. The instantaneous intensity for a node pair depends on several network summary statistics in a manner similar to an exponential random graph model \citep{goldenberg09}. 
    We were not able to locate a working implementation so we only compare against reported results.
    
    \item Dual Latent Space (DLS) \citep{yang2017decoupling}: Bivariate Hawkes process network model with separate continuous latent spaces for the base intensities and for reciprocal excitations. We use the implementation at \url{https://github.com/jiaseny/lspp}
    
    \item ADM4 \citep{zhou2013learning}: Multivariate Hawkes process network model with penalties to encourage sparsity and low rank for the excitation matrix $\alphaMat$. We use the implementation in the Python package \texttt{tick} \citep{bacry2017tick}.
\end{itemize}

\subsection{Scalability of MULCH}
\label{sec:scalibility}

\begin{table}[t]
\centering

\caption{Wall clock times to fit CHIP, BHM, and MULCH on the Facebook and MID datasets. Test $\ell$ denotes the mean test log-likelihood per event as defined in Section \ref{sec:exp_pred}.}
\label{tab:scalability}

\begin{tabular}{cccccccccc}
\toprule
 & \multicolumn{3}{c}{CHIP} & \multicolumn{3}{c}{BHM} & \multicolumn{3}{c}{MULCH} \\ 
\cmidrule{2-10} 
Dataset & Test $\ell$  & $K$  & Time  & Test $\ell$  & $K$  & Time & Test $\ell$   & $K$  & Time  \\ 
\midrule
Facebook & $-9.48$ & $9$ & $3.8$ minutes  & $-14.3$ & $15$ & $3.5$ minutes  & $-6.82$ & $1$ & $16$ hours   \\
MID      & $-3.67$ & $2$ & $0.48$ seconds & $-5.18$  & $91$ & $3.5$ seconds & $-3.53$ & $2$ & $31$ seconds \\ 
\bottomrule
\end{tabular}

\end{table}

We compare scalability of MULCH against the CHIP and BHM models. We fit all 3 models to the  Facebook ($n=43,953$) and MID ($n=147$) datasets (see Table \ref{tab:dataStats} for other dataset statistics). We run the models without likelihood refinement over a range of number of blocks $K$ and report the wall clock time required to fit to the model at $K$ corresponding to the best test log-likelihood score. 
The wall clock times and test log-likelihoods are both shown in Table \ref{tab:scalability}.
Both CHIP and BHM, which utilize univariate Hawkes processes, are much faster to fit than our proposed MULCH model that uses multivariate Hawkes processes. 
However, MULCH is able to achieve better fits, as indicated by the higher test log-likelihood.

\subsection{Ablation Experiment on Number of Excitation Types}
\label{sec:supp_exp_alphas}

The excitation parameters $\alpha_{ab}^{xy \rightarrow ij}$ control which node pairs can mutually excite each other. For the full MULCH model, we consider 6 types of excitations, listed in Table \ref{tab:excitationDescriptions}. In this experiment, we perform an ablation study by removing some of the excitation types (i.e. set to 0), then refit the model. We test out two additional variants of our MULCH model: two $\alpha$'s (self and reciprocal excitation only) and four $\alpha$'s (add turn continuation and generalized reciprocity). 
We use the Reality Mining dataset in this experiment.

We find that the full six-$\alpha$ model has slightly better predictive accuracy. Compared to the $0.954$ link prediction AUC of the full six-$\alpha$ model, the four-$\alpha$ and two-$\alpha$ models had AUCs of $0.950$ and $0.951$, respectively. All 3 models had similar test-log likelihood ($\pm 0.01$).

We find that the six-$\alpha$ and four-$\alpha$ models significantly improve generative accuracy compared to the two-$\alpha$ model, as shown in Figure \ref{fig:AlpahsMotifCounts}.  
The improvement is because the
two-$\alpha$ model has no mechanism to generate 3-node motifs, e.g.~triangles, and needs to generate way too many 2-node motifs (see Figure \ref{fig:2alphaMotifCountsR}, $\text{MAPE}=69.8$) to generate 3-node motif counts close to the actual network in Figure \ref{fig:MotifCountsR2}. The full six-$\alpha$ model (Figure \ref{fig:6alphaMotifCountsR}) can accurately generate all motif counts due to the additional excitations and achieves a much better $\text{MAPE}=16.5$ score. Note that the four-$\alpha$ model is able to perform well and generate motifs close to actual dataset at $\text{MAPE}=16.1$.

\begin{figure*}[t]
    \newcommand{\figwidth}{0.245\textwidth}
    \centering
    \hfill
    \begin{subfigure}[c]{\figwidth}
        \centering
        \includegraphics[width=\textwidth]{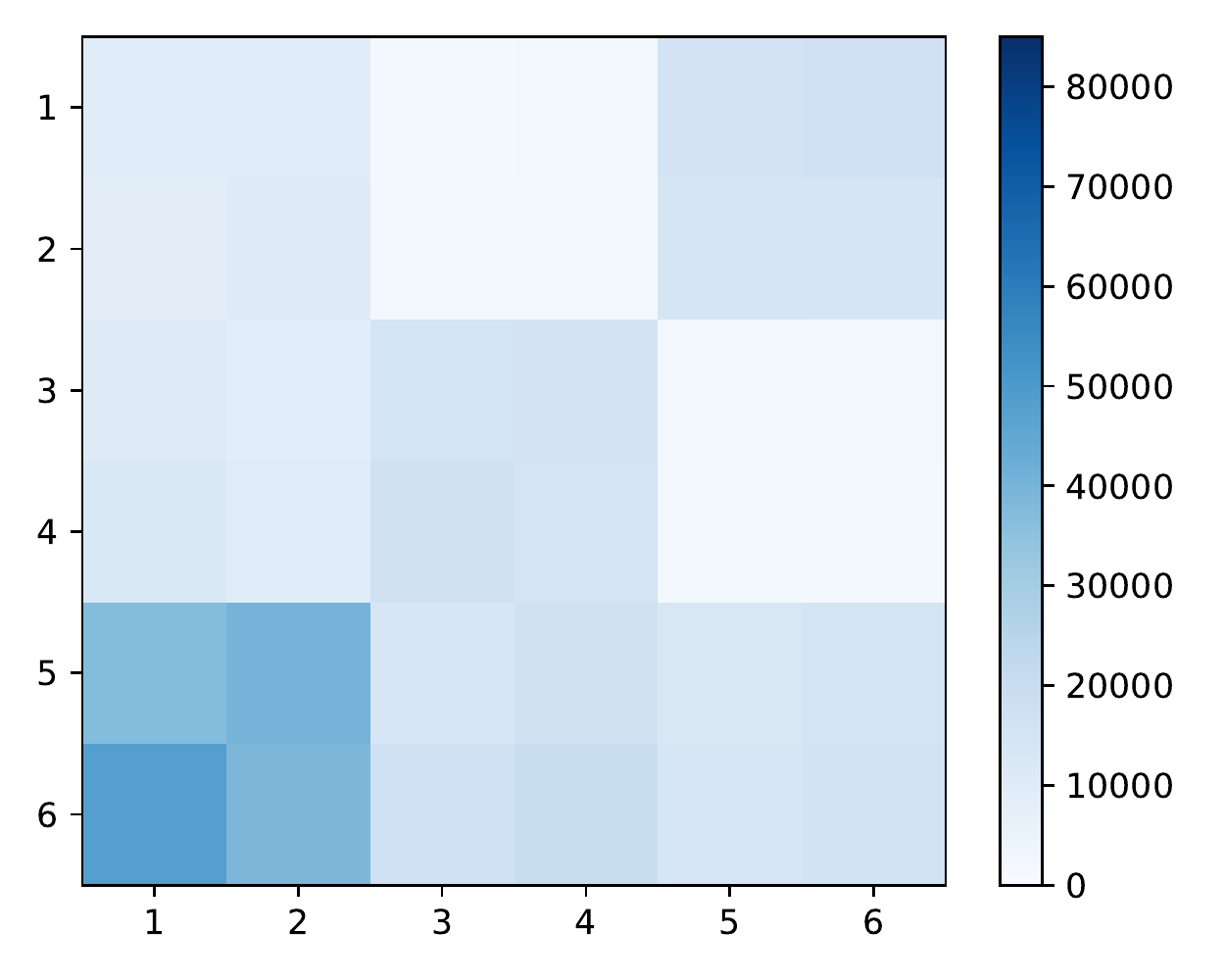}
        \caption{Actual network}
        \label{fig:MotifCountsR2}
    \end{subfigure}
    \hfill
    \begin{subfigure}[c]{\figwidth}
        \centering
        \includegraphics[width=\textwidth]{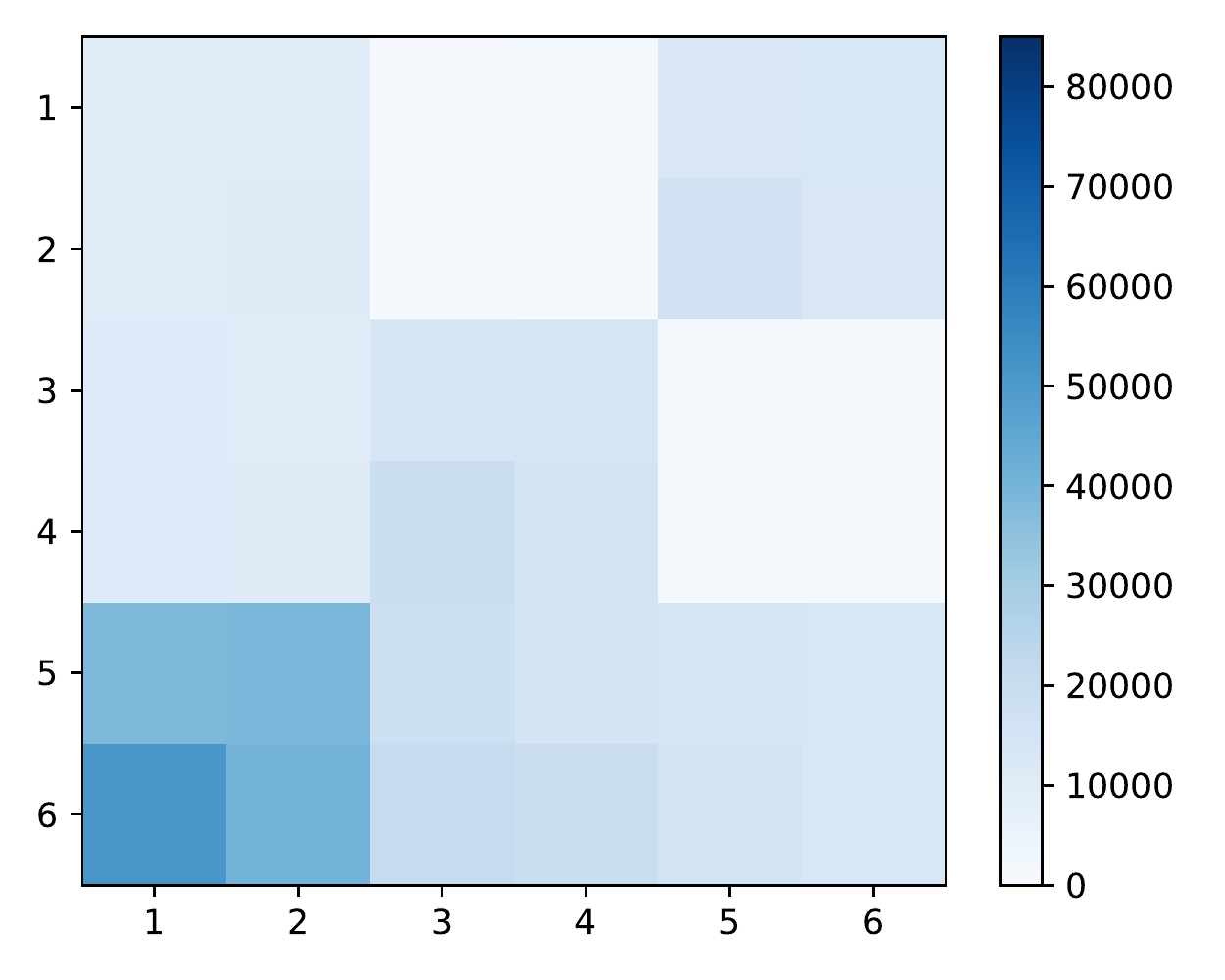}
        \caption{Six-$\alpha$ ($\text{MAPE}=16.5$)}
        \label{fig:6alphaMotifCountsR}
    \end{subfigure}
    \hfill
    \begin{subfigure}[c]{\figwidth}
        \centering
        \includegraphics[width=\textwidth]{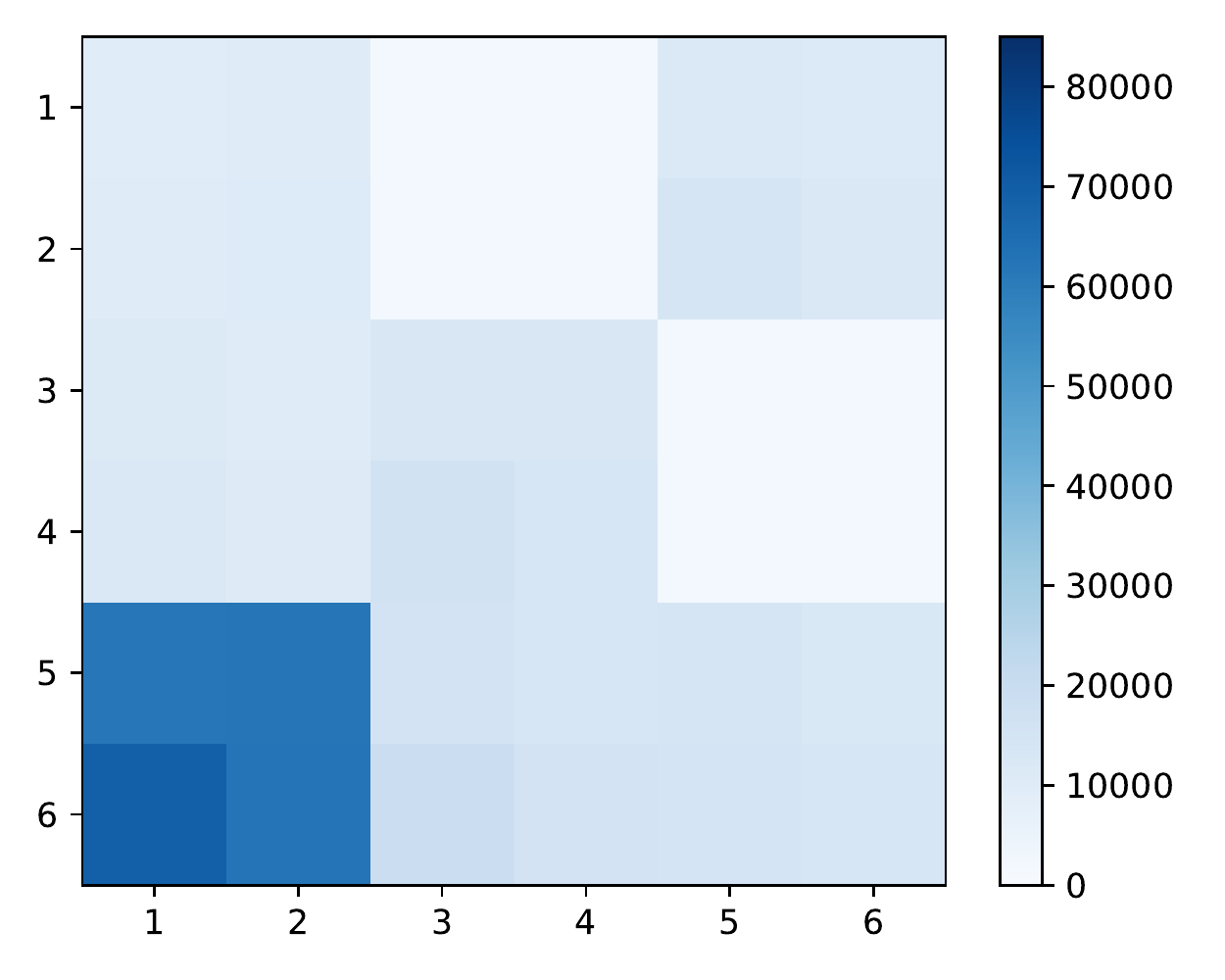}
        \caption{Four-$\alpha$ ($\text{MAPE}=16.1$)}
        \label{fig:4alphaMotifCountsR}
    \end{subfigure}
    \hfill
    \begin{subfigure}[c]{\figwidth}
        \centering
        \includegraphics[width=\textwidth]{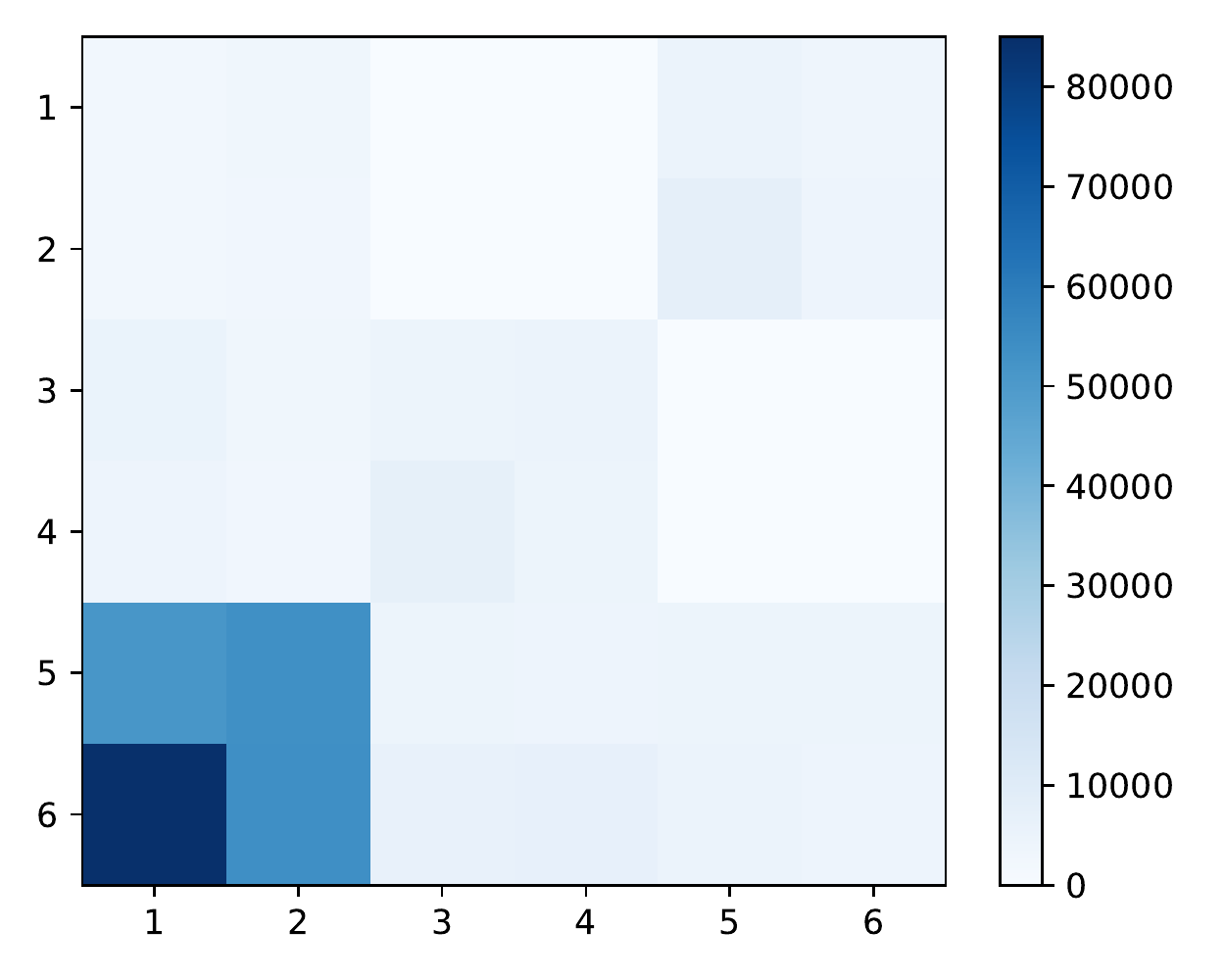}
        \caption{Two-$\alpha$ ($\text{MAPE}=69.8$)}
        \label{fig:2alphaMotifCountsR}
    \end{subfigure}
    \hfill
    \caption{Average temporal motif counts for time window $\delta=1$ week on $10$ simulated networks from six-$\alpha$, four-$\alpha$, and two-$\alpha$ MULCH models fitted on the Reality Mining dataset.}
    \label{fig:AlpahsMotifCounts}
\end{figure*}

%%%%%%%%%%%%%%%%%%%%%%%%%%%%%%%%%%%%%%%%%%%%%%%%%%%%%%%%%%%%%%%%%%%%%%%%%%%%%%%
%%%%%%%%%%%%%%%%%%%%%%%%%%%%%%%%%%%%%%%%%%%%%%%%%%%%%%%%%%%%%%%%%%%%%%%%%%%%%%%

\end{document}